\documentclass[11pt]{article}
\usepackage{amsmath,amssymb,graphicx,float,epsf}
\usepackage{csquotes}
\usepackage{appendix}
\usepackage{multirow}
\def\RR{\hbox{I\kern-.2em\hbox{R}}}
\newcommand{\qed}{\hbox to 0pt{}\hfill$\rlap{$\sqcap$}\sqcup$ \vspace{3mm}}
\numberwithin{equation}{section}
\restylefloat{figure}
\usepackage{graphicx} 
\usepackage{amsmath}
\usepackage{amsthm}
\usepackage{amsfonts}
\usepackage[caption = false]{subfig}
\usepackage{wrapfig}
\usepackage{enumerate}
\usepackage{enumitem}
\usepackage{array}
\usepackage{caption}
\usepackage[margin=1cm]{caption}
\usepackage{authblk}

\usepackage[utf8]{inputenc}

\newtheorem{Lemma}{\textbf{Lemma}}
\newtheorem{Th}{\textbf{Theorem}}

\usepackage{graphicx}
\usepackage{tikz}
\usetikzlibrary{shapes.geometric, arrows}
\tikzstyle{rect} = [draw, rectangle, fill=blue!20, text width=6em, text centered, minimum height=2em]
\tikzstyle{elli} = [draw, ellipse, fill=red!20, minimum height=2em]
\tikzstyle{circ} = [draw, circle, fill=white!20, minimum width=8pt, inner sep=5pt]
\tikzstyle{diam} = [draw, diamond, fill=white!20, text width=6em, text badly centered, inner sep=0pt]
\tikzstyle{line} = [draw, -latex']
\usepackage{a4wide} 
\usepackage{authblk} 

\usepackage{varioref}
\usepackage{hyperref}
\hypersetup{colorlinks=true, linkcolor=violet, citecolor=cyan, urlcolor=teal} 

\usepackage{tikz}
\usetikzlibrary{shapes.geometric, arrows}
\usetikzlibrary{calc}
\tikzstyle{cs} = [rectangle,minimum width=1cm, minimum height=1cm, text centered,
draw=blue,fill=blue!5, thick]
\tikzstyle{ce} = [rectangle,minimum width=1cm, minimum height=1cm, text centered,draw=orange,fill=orange!5,thick]
\tikzstyle{ci} = [rectangle,minimum width=1cm, minimum height=1cm, text centered,
draw=red!70!red,fill=red!5, thick]
\tikzstyle{cr} = [rectangle,minimum width=1cm, minimum height=1cm, text centered,draw=green,fill=green!5,thick]
\tikzstyle{cv} = [rectangle,minimum width=1cm, minimum height=1cm, text centered,draw=green,fill=green!5,thick]
\tikzstyle{ct} = [rectangle,minimum width=1cm, minimum height=1cm, text centered,draw=green,fill=green!5,thick]
\tikzstyle{arrow} = [thick,->,>=stealth]

\graphicspath{{Figures/}}
\date{}

\begin{document}

	\title{{Bifurcation Analysis of an Influenza A (H1N1) Model with Treatment and Vaccination}}

	\author[1]{\small  Kazi Mehedi Mohammad\thanks{ Email: mehedimim.me@gmail.com}}
 \author[1]{\small Asma Akter Akhi \thanks{Email: akhiasma752@gmail.com}}
	\author[1]{\small Md. Kamrujjaman\thanks{Corresponding author Email: kamrujjaman@du.ac.bd}}

	\affil[1]{\footnotesize Department of Mathematics, University of Dhaka, Dhaka 1000, Bangladesh}

	\maketitle
	\vspace{-1.0cm}
	\noindent\rule{6.35in}{0.02in}\\
	{\bf Abstract.}\\
	This study focuses on the modeling, mathematical analysis, developing theories, and numerical simulation of Influenza virus transmission. We have proved the existence, uniqueness, positivity, and boundedness of the solutions. Also, investigate the qualitative behavior of the models and find the basic reproduction number $(\mathcal{R}_0)$ that guarantees the asymptotic stability of the disease-free and endemic equilibrium points. The local and global asymptotic stability of the disease free state and endemic equilibrium of the system is analyzed with the Lyapunov method, Routh-Hurwitz, and other criteria and presented graphically. {This study helps to investigate the effectiveness of control policy and makes suggestions for alternative control policies. Bifurcation analyses are carried out to determine prevention strategies.} Transcritical, Hopf, and backward bifurcation analyses are displayed analytically and numerically to show the dynamics of disease transmission in different cases. Moreover,
	analysis of contour plot, box plot, relative biases, phase portraits are presented to show the
	influential parameters to curtail the disease outbreak.
	 We are interested in finding the nature of $\mathcal{R}_0$, which determines whether the disease dies out or persists in the population. The findings indicate that the dynamics of the model are determined by the threshold parameter $\mathcal{R}_0$. \\

 
	\noindent{\it \footnotesize Keywords}: {\small SVEIRT model; Influenza H1N1; Equilibrium States; Stability Analysis; Bifurcation.}\\
	\noindent
	\noindent{\it \footnotesize AMS Subject Classification 2010}: 53C25, 83C05, 57N16. \\
	\noindent\rule{6.35in}{0.02in}
\section{Introduction}
{Based on historical records, the flu has existed for at least 1,500 years. Hippocrates (5th century BC) is credited with first describing the history of influenza when he wrote that a disease like the flu had traveled from northern Greece to the southern islands and beyond. When a flu outbreak struck Florence, Italy in the 1300s, the city's officials dubbed it influenza di freddo, or ``cold influence", presumably a reference to their theory of the disease's cause. 
	Numerous flu epidemics have been documented throughout history, ranging from one that originated in Asia and moved to Europe and Africa in 1580 to others that have occurred throughout centuries in Europe and Britain. With a third of the world's population affected and an estimated 50 million deaths, the 1918 ``Spanish flu" pandemic is referred to as the ``mother of all pandemics" and was the deadliest pandemic in history \cite{CDC, WHO}.}
	
Throughout human history, pandemics and epidemics have decimated the population several times, frequently changing the course of history drastically and putting an end to entire civilizations. Mathematical models that explain the dynamics of infectious diseases are crucial to public health because they shed light on how to adopt effective and feasible disease-control measures. Flu, or influenza, is a condition brought on by the Orthomyxoviridae virus that mostly affects the nose, bronchi, throat, and occasionally the lungs. Humans who have influenza frequently have fever over 38 degrees Celsius, coughing, headaches, sniffles, and anorexia. ILI, or Influenza like Illness, is an unpleasant body condition. Virus transmission lasts 2–7 days, and it typically goes away on its own \cite{Stability Bound-5, Stability Bound-8, Stability Bound-9}. This flu is often referred to as a self-limiting disease by Indonesians. If there are no complications from other disorders, the illness will be gone in 4–7 days. The immune system of a person has a significant impact on the disease's severity. The generation of airborne particles and aerosols containing viruses is necessary for respiratory transmission. When people speak and breathe normally, aerosols are created. Sneezing is a method of expulsion from the nasal cavity that is more efficient if the infection generates more snot \cite{Stability Bound-14, Stability Bound-17, Stability Bound-18}.

 The World Health Organization (WHO) proclaimed the influenza A (H1N1) virus, a new virus, to be a pandemic on June 11, 2009, after it was discovered in Mexico and the United States \cite{WHO}. 
 A (H1N1) virus strain that primarily affected children and young people without immunity to the new strain and started in North America but spread globally also became a pandemic \cite{Stability Bound-19, Stability Bound-20}. Since many elderly persons had previously been exposed to a similar H1N1 virus strain, they were shielded by their antibodies. However, it caused the deaths of more than 200,000 individuals worldwide \cite{WHO}.

Despite the availability of vaccinations for numerous infectious diseases, World continues to experience significant suffering and mortality from these diseases. Given this context, additional research is essential to determine the control of the widespread transmission of the H1N1 influenza A pandemic virus. 
Sequencing the mathematical models, Nguyen Huu Khanh has considered the SEIR model, which describes the influenza virus's propagation while taking human illness resistance into account \cite{Stability Bound-13}. Consequently, in the simulation, an exposed or infected individual. Without therapy, a person could return to a vulnerable one.
In a series, many authors looked to the well-known SEIR model or its adaptations to explain how people move through various compartments, which stand in for the phases of disease across the entire population over time \cite{Stability Bound-14, AAA1, Stability Bound-21}. By receiving an annual influenza vaccination, people can avoid influenza. Due to the virus's fast mutation, a vaccine created for one year might not be effective the following year \cite{Stability Bound-23}. Additionally, the virus's antigenic drift may take place after the year's vaccine has been developed, making it less protective. As a result, outbreaks are more likely to happen, especially among high-risk populations. Other precautions include avoiding ill people, hiding coughs and sneezes, and often washing your hands \cite{CDC, WHO, Stability Bound-14}.

Considering the existing literature's and historical notes, the main objectives of this article are:
\begin{itemize}
    \item Vaccination and treatment strategy to control the disease spreading and outbreaks. 
    \item Perform theoretical observation of the model by examining the existence, positivity, and boundedness of the solutions.
    \item Observe the basic reproduction number with vaccination and without vaccination. 
    \item Perform Hopf, forward and backward bifurcation analysis of disease-free and endemic equilibrium and analyze their local and global stability. 
    \item Contour plots, box plots, and phase plane analyses were conducted to scrutinize both single and multi-compartment interactions.
\end{itemize}
The findings of this article regarding the goal are:
\begin{itemize}
    \item Disease-free equilibrium (DFE) and endemic equilibrium (EE) points persist in the system. 
    \item Persistence theorem considering the eigenvalue and their corresponding eigenfunctin analysis.
    \item After numerical simulation, it is evident that transcritical bifurcation occurs at DFE when the basic reproduction number $(\mathcal{R}_0)$ is equal to one. That means the DFE becomes stable situation to unstable.
    \item Backward bifurcation property arises in the model because of the reinfection of the susceptible population. Moreover, the disease transmission rate, population density, interventions, and contact patterns influence the relationship between the basic reproduction number and the force of infection.  
    \item When a Hopf bifurcation occurs, the disease-free equilibrium becomes unstable, and a stable limit cycle appears. This cyclic pattern shows the periodic oscillations in the dynamics of the disease, with the number of individuals in each compartment varying over time.
    \item We conducted contour plot, box plot, and relative influence analyses to illustrate various scenarios and their effects on the basic reproduction number $\mathcal{R}_0$.
\end{itemize}

This paper is organized as follows: mathematical model is discussed elaborately in Sections \ref{Section-Mathematical Model SVEIRT}. Existence, positivity, and boundedness of solutions are described in Sections \ref{Subsection-Existence of Solution}, \ref{Subsection-Positivity of Solution}, \ref{Subsection-Boundedness of Solution}. Further, the determination of fixed points such as disease free equilibirum (DFE) and endemic equilibrium  (EE) and calculation of basic reproduction number $(\mathcal{R}_0)$ (with control and without control) are presented in Sections \ref{Section-Determination Fixed Points}, \ref{Subsection-DFE point}, \ref{Subsection-EE point}, \ref{Section-Reproduction Number}, \ref{Subsection-R0-with-Control}, and \ref{Subsection-R0-without control}, respectively. Then, the existence of the endemic equilibrium point is calculated, and forward-backward and Hopf bifurcation analysis are carried out in Sections \ref{Section-Existence of EE}, \ref{Section-Hopf-Bifurcation}. In Sections \ref{Section-Local-global-stability of DFE EE}, \ref{Subsection-Stability and Persistence} local, global stability of DFE and EE is presented, and stability and persistence of solutions are discussed. In Section \ref{Section-Numerical-Simulation}, we presented a variety of numerical examples, elucidating the results of contour plot and box plot analyses, phase plane analysis focusing on the relative influence on $\mathcal{R}_0$, and computational biological findings.
The outcomes are summarized and discussed in Section \ref{Section-Concluding-Remarks}.

\section{Mathematical Model Formulation of Influenza}\label{Section-Mathematical Model SVEIRT}
Most pandemic situations observe an exponential curve followed by gradual flattening \cite{Marcheva Book}, that is, reducing the epidemic peak. In the absence of any established treatment or vaccination, understanding the transmission dynamics of a new infectious disease outbreak is imperative for flattening the curve. Mathematical models are thus considered a crucial tool to the public health authorities in this regard as their decision to optimize the control measures depends hugely on the short- and long-term predictions of these models \cite{Stability Bound-4}. Amongst the various models used for describing such epidemic evolutions, classic compartmental models such as SIR and SEIR are of immense value to decision-makers and even non-expert operators for their simplicity, reliability, and usage of multiple data sources \cite{CDC, WHO, Stability Bound-6}. During the outbreak of seasonal Influenza virus, these epidemiological models have been of great significance to many countries \cite{WHO} including Italy, Mexico who have widely adopted the models for obtaining insights into the recent situations, assessing the impact of the outbreak control measured, searching for alternative interventions and providing a roadmap to other similar settings. Moreover, in disadvantaged settings, prediction models with multiple features will be of great value to healthcare workers for monitoring patients within limited resources \cite{CDC, Marcheva Book, Stability Bound-7, Stability Bound-10}.\\
In the case of the modeling of infectious disease, the traditional SIR model permits the determination of critical conditions of disease occurrence in the population with total population size \cite{Marcheva Book}. Influenza is classified as a person-to-person transmissible disease. In several cases, infected people have no apparent symptoms, and, in those cases, an SEIR model is principally used since exposed class (E) individuals spread the infection rapidly. {In contrast to the SIR epidemic model, SEIR is a more updated and sophisticated model that is biologically plausible regarding numerous pandemics and infectious disorders \cite{Stability Bound-11, Stability Bound-4}.} Therefore, a model with numerous compartments is a helpful tool for forecasting the nature of current influenza disease patterns. In this study, we consider a modified version of the typical SEIR model, we propose the following six compartments' potential SVEIRT (Susceptible-Vaccinated-Exposed-Infectious-Treatment-Removal) mathematical model as a system of ordinary differential equations:
\begin{align} \label{new_model}
	\begin{cases}
		\vspace{0.2cm}
		\displaystyle\frac{dS}{dt} = \Lambda - \left(\beta_1E+\beta_2I\right)S-(\mu+\phi) S. \\
		\vspace{0.2cm}
		\displaystyle\frac{dV}{dt} = \phi S-(1-\varepsilon)\left(\beta_1E+\beta_2I\right) V-\mu V.\\
		\vspace{0.2cm}
		\displaystyle\frac{dE}{dt} =\left(\beta_1E+\beta_2I\right)S-(\alpha+\mu) E.\\
		\vspace{0.2cm}
		\displaystyle\frac{dI}{dt} =\alpha E+ (1-\varepsilon)\left(\beta_1E+\beta_2I\right) V-(\mu+\delta+\gamma+\gamma_1) I.\\
		\vspace{0.2cm}
		\displaystyle\frac{dR}{dt} = \gamma I-\mu R.\\
		\displaystyle\frac{dT}{dt} =\gamma_1 I-\mu T. 
	\end{cases}
\end{align}
for $ t \in (0,\infty) $ with initial conditions,
\begin{align}\label{ic}
	S(0) = S_{0},\;\;V(0) = V_0. \;\;  E(0) = E_0, \;\;  I(0) = I_0, \;\;  R(0) = R_0, \; \text{and}\;  T(0) = T_0,
\end{align}
and the total population for the SVEIRT model is found by,
\begin{equation}\label{equ_2.5}
	N(t) \equiv S(t) +V(t)+ E(t)+ I(t) + R(t) +T(t).
\end{equation} 
\begin{figure}[H]
	\centering
	\includegraphics[width=4 in]{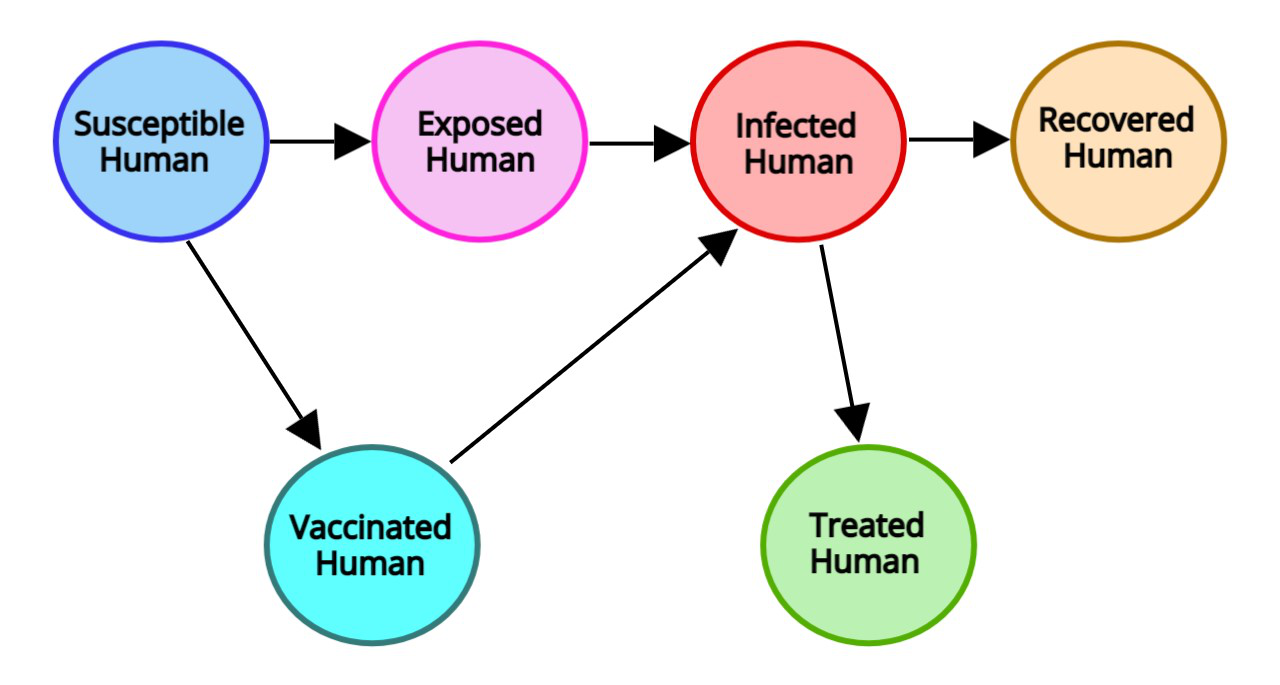} 
	\caption{The transmission cycle of disease in the SVEIRT model.}
	\label{transmission-cycle-SVEIRT}
\end{figure}

\noindent	
The definitions of all state variables and model parameters with brief descriptions are presented in Table~\ref{tableparameter}, and they are non-negative because of the dynamics of a population. The whole cycle and the flow diagram for the suggested model \eqref{new_model} are illustrated in Figure~\ref{transmission-cycle-SVEIRT}. 
Assume that the following guidelines control the transmission of disease:
\begin{enumerate}
	\item[(i)] The total population stays fixed at a level $N$ over the interval.
	\item[(ii)] The susceptible population becomes infected by contracting the disease. The rate of change of the susceptible population is proportional to the number of contracts between the $S(t)$ and $E(t)$ also between the $S(t)$ and $I(t)$. This number is proportional to the number of susceptible, exposed, and infected persons.
\end{enumerate}

\begin{table}[H]
 	\begin{center}
 		\caption{Model parameters values with descriptions.}
 		\scriptsize
 		\label{tableparameter} 
 		\begin{tabular}{|l|l|l|l|}
 			\hline\noalign{\smallskip}
 			\textbf{Notation} & \textbf{Definition} & \textbf{Value}  & \textbf{Source}  \\
 			\noalign{\smallskip}\hline\noalign{\smallskip}
    $N$ & Total number of human population & &\\
			$S$ & Total number of susceptible population & & \\
			$V$ & Total number of vaccinated population & & \\
			$E$ & Total number of exposed population & & \\
			$I$ & Total number of infected population & & \\
			$R$ & Total number of recovered population & & \\
			$T$ & Total number of treated population & &\\
 			$ \alpha $ & Transition rate from $E$ to $I$  &  $0.75$ week$ ^{-1} $  & \cite{Stability Bound-5,Stability Bound-13} \\
 			$\Lambda$ & Recruitment rate in $S$ class & $5\times 10^2$ week$ ^{-1} $ & \cite{Stability Bound-14,Stability Bound-17} \\
 			$ \beta_1 $ & Transmission rate from contact with $E$ to class $S$  & $ [0.0045,0.0055 ]$ week$ ^{-1} $ & \cite{Stability Bound-18,Stability Bound-20}\\
 			$ \beta_2  $ & Transmission rate from contact with $I$ to class $S$&  $ [0.0045,0.0055 ]$ week$ ^{-1} $  & \cite{Stability Bound-21,Stability Bound-23} \\
 			$\gamma$ & Recovery rate of $I$  & $ 0.65 $ week$ ^{-1} $ & \cite{Stability Bound-19,Stability Bound-24}\\
 			$\gamma_1$ & Treatment progression rate of $I$ &  $ 0.25$ week$ ^{-1} $ & \cite{Stability Bound-8, Stability Bound-9} \\
 			$ \mu $ & Natural death rate & $ 5 \times 10^{-2}$ week$ ^{-1} $ & \cite{Stability Bound-5,Stability Bound-8}\\
 			$ \delta $ & Disease induced death rate & $0.3$ week $^{-1} $ & \cite{Stability Bound-13,Stability Bound-19}\\
 			$ N $ & Total Population in Mexico (assumed) & $520$  & \cite{Stability Bound-8,Stability Bound-20}\\
 			$ \lambda $ & Vaccine inefficiency rate & $0.55$  & \cite{Stability Bound-14,Stability Bound-19}\\
 			\noalign{\smallskip}\hline
 		\end{tabular}
 	\end{center}
 \end{table}
 
 In the following section, we stated a few auxiliary results and Appendix \ref{allproofs} contains the proofs of the existence, positivity, and boundedness of the solution.
\subsection{Existence of Solution}\label{Subsection-Existence of Solution}
\begin{Th}\label{theorem01}
	(Existence of solution). Let $\{S_0, V_0, E_0, I_0, R_0, T_0\} \in \mathbb{R}$ be presented. There subsists $t_0 >0$, and continuously differentiable functions $\{S, V, E, I, R, T : [0,t_0) \rightarrow \mathbb{R}\}$ such that the ordered pairs of states $(S, V, E, I, R, T)$ satisfies \eqref{new_model} and $(S, V, E, I, R, T)(0)= (S_0, V_0, E_0, I_0, R_0, T_0).$
\end{Th}

\subsection{Positivity of Solution}\label{Subsection-Positivity of Solution}
The boundedness and positivity of the solutions are two main constituents of an epidemic model. To convey that any solution with positive beginning values stays positive for all times $t>0$, it is necessary to establish that all parameters and variables are always positive for $t>0$. Positive behavior is biologically interpreted as the long-term survival of the population \cite{Stability Bound-14, Stability Bound-17}. 
\begin{Th}\label{th1}
	(Positivity of solution) Consider the initial conditions of the system \eqref{new_model} are $S(0) \geq 0 ,\; V(0) \geq 0, \; E(0) \geq 0,\; I(0) \geq 0,\; R(0) \geq 0,\; and \; T(0) \geq 0 $;
	the solutions $ S(t),\; V(t),\; E(t),\; I(t),\; R(t),\; and \; T(t) $ are non negative $\forall\;t > 0 $.
\end{Th}

\subsection{Boundedness of Solution}\label{Subsection-Boundedness of Solution}
 Boundedness can be understood as a natural growth constraint resulting from scarce resources, while positivity suggests that every member of the compartment population survives \cite{Stability Bound-14, AAA1, AAA2}. 
 
\begin{Th}\label{th2}
	\cite{Stability Bound-17}. (Positive invariance and boundedness of solutions)
	The closed region $ \Omega = \{(S,V,E, I, R,T) \in \mathbb{R}_{+}^{6} : 0 < N \leq \dfrac{\Lambda}{\mu} \} $ is positively invariant and attracting set for the system (\ref{new_model}).
\end{Th}

\begin{Th}\label{th3} \cite{Stability Bound-17}.
	The feasible region $\Omega$ is determined by, $$\Omega=\left\{(S(t),\;V(t),\;E(t),\;I(t),\;R(t),\;T(t))\in \mathbb{R}_{+}^6 \;\;\vert \;\;0\leq N \leq \max\left\{N(0),\frac{\Lambda}{\mu}\right\} \right\}.$$
	with initial conditions $S(t)>0,\;V(t)>0,\;E(t)>0,\;I(t)>0,\;R(t)>0,\;T(t)>0,\;$is positively invariant and attracting with respect to system (\ref{new_model})\; $\forall \;t>0$.
\end{Th}

\section{Determination of Fixed Points}\label{Section-Determination Fixed Points}
To identify the equilibrium points $(\widetilde{S},\widetilde{V},\widetilde{E},\widetilde{I},\widetilde{R},\widetilde{T})$ of the system \eqref{new_model}, we set the value of each derivative to zero. Thus, in situations of equilibrium, we obtain,

\begin{align} \label{equi}
	\begin{cases}
		\vspace{0.2cm}
		\displaystyle\Lambda - \left(\beta_1\widetilde{E}+\beta_2\widetilde{I}\right)\widetilde{S}-(\mu+\phi) \widetilde{S} =0.\\
		\vspace{0.2cm}
		\displaystyle\phi \widetilde{S}-(1-\varepsilon)\left(\beta_1\widetilde{E}+\beta_2\widetilde{I}\right) \widetilde{V}-\mu \widetilde{V} =0.\\
		\vspace{0.2cm}
		\displaystyle\left(\beta_1\widetilde{E}+\beta_2\widetilde{S}\right)\widetilde{S}-(\alpha+\mu) \widetilde{E} =0.\\
		\vspace{0.2cm}
		\displaystyle\alpha \widetilde{E}+ (1-\varepsilon)\left(\beta_1\widetilde{E}+\beta_2\widetilde{I}\right) \widetilde{V}-(\mu+\delta+\gamma+\gamma_1) \widetilde{I} =0.\\
		\vspace{0.2cm}
		\displaystyle\gamma \widetilde{I}-\mu \widetilde{R}=0.\\
		\displaystyle\gamma_1 \widetilde{I}-\mu \widetilde{T}=0.
	\end{cases}
\end{align}
Now, we have to solve the right side of the equations \eqref{equi} in the derivative terms to find steady states, and by solving the consequent equations for compartments $S, V, E, I, R$, and $T$, we discover that there are only two biologically significant equilibria in total. These events can be divided into two categories: when the influenza virus either becomes extinct in the area, i.e., $E=I=0$, or when the virus persists within the region  $(E\neq 0,\; I\neq 0)$.

\subsection{Disease-Free Equilibrium (DFE) Point}\label{DFE point}\label{Subsection-DFE point}
The disease-free equilibrium (DFE) point of \eqref{new_model} is, $\mathcal{E}^0 \equiv\left( \dfrac{\mu N}{\mu + \phi},\dfrac{\phi N}{\mu + \phi}, 0, 0, 0,0 \right).$

\subsection{Endemic Equilibrium (EE) Point}\label{Subsection-EE point}
The endemic equilibrium (EE) point for the system \eqref{new_model} is, 
$\mathcal{E}^*=\left(S^*, V^*, E^*, I^*, R^*, T^*\right)$. Here,
$S^{*}=\frac{\displaystyle \Lambda-(\alpha+\mu)E^{*}}{\displaystyle (\mu+\phi)}$, $V^{*}= \frac{\displaystyle \phi(\Lambda-a_1 E^{*})}{\displaystyle a_2(\mu+\lambda\lambda_1)},$ $R^* = \frac{\displaystyle \lambda I^*}{\displaystyle \mu}, \;
	T^* = \frac{\displaystyle \lambda_1 I^*}{\displaystyle \mu},$\\
$E^*=\frac{\displaystyle (\Lambda \beta_1-a_1a_2-a_1\beta_2 I^{*})\pm \sqrt{(\Lambda \beta_1-a_1a_2-a_1\beta_2 I^{*})^2 + 4\Lambda\beta_2I^* a_1\beta_1}}{\displaystyle 2a_1\beta_1}.$\\ Where, $a_1=\alpha+\mu, \; a_2=\mu+\phi, \; \lambda_1=\beta_1E^*+\beta_2I^*, \;\text{and}\; \lambda=1-\epsilon.$\\

A concise analysis of the endemic equilibrium point can be found in the Appendix \ref{endemic_calculation}.

\section{Calculation of Basic Reproduction Number }\label{Section-Reproduction Number}
The fundamental reproduction number, a crucial threshold quantity for examining infectious disease modeling, has been calculated in this section. To mathematically quantify the volatility of an infectious disease, it was developed for the study of epidemiology. It establishes whether the disease will disappear over time or remain prevalent in the community. This threshold quantity, usually denoted by $\mathcal{R}_{0}$, is defined as the expected number of secondary infections resulting from a single primary infection in a population where everyone is susceptible. When $\mathcal{R}_0 > 1$, referring that one primary infection can lead to several subsequent infections, the disease-free equilibrium (DFE) becomes unstable, causing an epidemic in the population. Conversely, if $\mathcal{R}_{0}< 1$, the DFE is locally asymptotically stable, preventing the disease from persisting in the community
 \cite{Marcheva Book}. That reflects the scenario that the situation is under control. Therefore, an efficient plan should be created as soon as a pandemic emerges to ensure that the reproduction number falls to less than zero.
Since, the considered model \eqref{new_model} has DFE, $\displaystyle {\mathcal{E}}^0\equiv\Big(\frac{\mu N}{\mu+\phi},\frac{\phi N}{\mu+\phi},0,0,0,0\Big);$ hence, $\mathcal{R}_0$ can be calculated analytically. In this section, we have employed the next-generation matrix method to determine the basic reproduction number for the Influenza model presented in \eqref{new_model}. The calculation is based on the formula $\mathcal{R}_0 = \rho(FV^{-1})$, that represents the spectral radius of $FV^{-1}$.

 \noindent
 Succinct computation is shown in Appendix \ref{BRN} 
\subsection{Basic Reproduction Number with Control}\label{Subsection-R0-with-Control}
The threshold value for the system (\ref{new_model}) associated with controlling strategies can be presented as follows,
\begin{align} \label{repro1}
	\mathcal{R}_{0V} 
	&=\frac{\mu N\left[\alpha\beta_2+\beta_1(\mu+\delta+\gamma+\gamma_1)\right]}{(\mu+\phi)(\alpha+\mu)(\mu+\delta+\gamma+\gamma_1)}+\frac{N\phi\beta_2\lambda}{(\mu+\phi)(\mu+\delta+\gamma+\gamma_1)}.
\end{align}

\subsection{Basic Reproduction Number without Control}\label{Subsection-R0-without control}
The fundamental reproduction number for the system (\ref{new_model}) in the absence of control strategies can be expressed as follows,
\begin{align} \label{repro2}
	\mathcal{R}_{0} &=\frac{S_0\left[\alpha\beta_2+\beta_1(\gamma+\gamma_1+\mu+\delta)\right]}{(\alpha+\mu)(\gamma+\gamma_1+\delta+\mu)}.
\end{align}

\section{Presence of Endemic Equilibrium}\label{Section-Existence of EE}
The existence of the endemic equilibrium point and its uniqueness is dependent upon the corresponding threshold number $ \mathcal{R}_{0} >1 $. There endures a solitary endemic equilibrium $\mathcal{E}^*\equiv(S^*,V^*,E^*,I^*,R^*,T^*)$ for the model \eqref{new_model}.
From system \eqref{equiEE} we have, 
\begin{align*} \label{eeI}
S^*=\frac{\displaystyle \Lambda}{\displaystyle \lambda_1+\mu+\phi},\;
	R^* = \frac{\lambda I^*}{\mu}, \;\text{and}\;\; T^* = \frac{\lambda_1 I^*}{\mu}.
\end{align*}
Now, from second equation of \eqref{equiEE},
\begin{align*}
	&\phi S^*=(\lambda\lambda_1+\mu)V^*.\\
	&V^*=\frac{\phi S}{\lambda\lambda_1+\mu}=\frac{\Lambda\phi}{(\lambda\lambda_1+\mu)(\lambda_1+\mu+\phi)}.\\
	&V^*=\frac{\Lambda\phi}{\lambda\lambda_1^2+\mu\lambda_1+\mu\lambda\lambda_1+\mu^2+\phi\lambda\lambda_1+\phi\mu}.
\end{align*}
Let, $p=\lambda\lambda_1^2+\mu\lambda_1+\mu\lambda\lambda_1+\mu^2+\phi\lambda\lambda_1+\phi\mu$.\\
Now, adding the second and fourth equation of \eqref{equiEE} we have,
\begin{align*}
	&\phi S^*-\mu V^*+\alpha E^*-(\mu+\delta+\gamma+\gamma_1)I^*=0\\
	&\Rightarrow \frac{\phi\Lambda}{\lambda_1+\mu+\phi}-\frac{\mu\phi\Lambda}{p}+\alpha E^* - (\mu+\delta+\gamma+\gamma_1)I^*=0\\
	&\Rightarrow I^*= \frac{E^*\alpha}{(\mu+\delta+\gamma+\gamma_1)}-\frac{\alpha\mu\phi\Lambda}{\alpha p (\mu+\delta+\gamma+\gamma_1)}+\frac{\phi\Lambda\alpha}{\alpha(\lambda_1+\mu+\phi)(\mu+\delta+\gamma+\gamma_1)}.
\end{align*}
Now, from third equation of \eqref{equiEE} we have get,
$E^*=\frac{\displaystyle \Lambda\lambda_1}{\displaystyle (\lambda_1+\mu+\phi)(\alpha+\mu)}.$
Substituting this in above expression we have,
\begin{align*}
	I^*&=\frac{\Lambda\lambda_1\alpha}{a_1(\alpha+\mu)(\lambda_1+\mu+\phi)}-\frac{\alpha\mu\phi\Lambda}{a_1\alpha p}	+\frac{\phi\Lambda\alpha}{\alpha(\lambda_1+\mu+\phi)a_1}.
\end{align*} 
Let, $a_1=(\mu+\delta+\gamma+\gamma_1).$ From \eqref{repro2} we have, $\displaystyle R_0=\frac{\Lambda(\alpha\beta_2+\beta_1 a_1)}{(\mu+\phi)(\alpha+\mu)a_1}$.
Now,
\begin{align*}
	I^*=&\frac{\Lambda\lambda_1\alpha^2p-(\alpha+\mu)(\lambda_1+\mu+\phi)+p(\alpha+\mu)\phi\alpha\Lambda}{a_1(\alpha+\mu)(\lambda_1+\mu+\phi)\alpha p}\\
	=&\frac{\Lambda\lambda_1\alpha^2p}{a_1(\alpha+\mu)(\lambda_1+\mu+\phi)\alpha p}+\frac{\left(\left(\frac{\Lambda(\alpha\beta_2+\beta_1 a_1)}{(\mu+\phi)(\alpha+\mu)a_1} \right)-1\right)p\phi\{1-(\mu+\phi+\lambda_1)\}(\mu+\phi)}{(\mu+\phi+\lambda_1)}\\
	=&\frac{\Lambda\lambda_1\alpha^2p}{a_1(\alpha+\mu)(\lambda_1+\mu+\phi)\alpha p}+\frac{(\mathcal{R}_0-1)p\phi\{1-(\mu+\phi+\lambda_1)\}(\mu+\phi)}{(\mu+\phi+\lambda_1)}.
\end{align*}
Let, the infectious force at the endemic steady state,
\begin{align*}
	\lambda_1^* &=(\beta_1E^*+\beta_2I^*)\\
	\Rightarrow \lambda_1^* &=\frac{\beta_1\Lambda\lambda_1}{(\lambda_1+\mu+\phi)(\alpha+\mu)}+\frac{\beta_2(\Lambda\alpha^2p\lambda_1-(\alpha+\mu)(\lambda_1+\mu+\phi)\alpha\mu\phi\Lambda+p(\alpha+\mu)\phi\Lambda\alpha)}{a_1\alpha p(\alpha+\mu)(\lambda_1+\mu+\phi)}\\
	\Rightarrow \lambda_1^* &=\frac{a_1\alpha p\beta_1\Lambda\lambda_1+\Lambda\alpha^2p\lambda_1-(\alpha+\mu)(\lambda_1+\mu+\phi)\alpha\mu\phi\Lambda+p(\alpha+\mu)\phi\Lambda\alpha}{a_1\alpha p(\alpha+\mu)(\lambda_1+\mu+\phi)}
\end{align*}
\begin{align*}
	\Rightarrow & \lambda_1^*a_1\alpha p\{\lambda_1^*(\alpha+\mu)+(\mu+\phi)(\alpha+\mu)\}=p[a_1\alpha\beta_1\Lambda\lambda_1^*+\lambda_1^*\Lambda\alpha^2+(\alpha+\mu)\phi\Lambda\alpha]-\\
	&\alpha\mu\phi\Lambda\{\lambda_1^*(\alpha+\mu)+(\mu+\phi)(\alpha+\mu)\}\\
	\Rightarrow & \lambda_1^*a_1\alpha\{\lambda_1^{*2}\lambda+\lambda_1^*(\mu+\lambda(\mu+\phi))+\mu(\mu+\phi)\}[\lambda_1^*(\alpha+\mu)+(\mu+\phi)(\alpha+\mu)]=\\
	&\{\lambda_1^{*2}\lambda+\lambda_1^*(\mu+\lambda(\mu+\phi))+\mu(\mu+\phi)\}[a_1\alpha\beta_1\Lambda\lambda_1^*+\lambda_1^*\Lambda\alpha^2+(\alpha+\mu)\phi\Lambda\alpha]\\
	&-\lambda_1^*(\alpha+\mu)\alpha\mu\phi\Lambda-(\alpha+\mu)(\mu+\phi)\alpha\mu\phi\Lambda\\
	\Rightarrow & [\lambda_1^{*3} a_1\alpha\lambda + \lambda_1^{*2} a_1\alpha(\mu+\lambda(\mu+\phi))+\lambda_1^* a_1\alpha \mu(\mu+\phi)][\lambda_1^*(\alpha+\mu)+(\mu+\phi)(\alpha+\mu)]=\\
	&\lambda_1^{*3}\lambda a_1\alpha\beta_1\Lambda+\lambda_1^{*3}\lambda\Lambda\alpha^2+\lambda_1^{*2}\lambda\alpha\Lambda\phi(\alpha+\mu)+\lambda_1^{*2}\Lambda\alpha a_1\beta_1(\mu+\lambda(\mu+\phi))+\\
	&\lambda_1^{*2}(\mu+\lambda(\mu+\phi))\Lambda\alpha^2+\lambda_1^*(\mu+\lambda(\mu+\phi))(\alpha+\mu)\phi\Lambda\alpha+\mu(\mu+\phi) a_1\alpha\beta_1\Lambda\lambda_1^*+\\
	&\mu(\mu+\phi)\lambda_1^*\Lambda\alpha^2+(\mu+\phi)\mu(\alpha+\mu)\phi\Lambda\alpha-\lambda_1^*(\alpha+\mu)\alpha\mu\phi\Lambda-(\alpha+\mu)(\mu+\phi)\alpha\mu\phi\Lambda\\
	\Rightarrow & \lambda_1^{*4} (\alpha+\mu)a_1\alpha\lambda+\lambda_1^{*3} a_1\alpha(\alpha+\mu)(\mu+\lambda(\mu+\phi))+\lambda_1^*(\alpha+\mu)a_1\alpha\mu(\mu+\phi)+\\
	&\lambda_1^{*3} a_1\alpha\lambda(\mu+\phi)(\alpha+\mu)+\lambda_1^{*2}a_1\alpha(\mu+\lambda(\mu+\phi))(\mu+\phi)(\alpha+\mu)+\lambda_1^*a_1\alpha\mu(\mu+\phi)^2(\alpha+\mu)=\\
	&\lambda_1^{*3}[\lambda a_1\alpha\beta_1\Lambda+\lambda\Lambda\alpha^2]+\lambda_1^{*2}[\lambda\alpha\Lambda\phi(\alpha+\mu)+\Lambda\alpha a_1\beta_1(\mu+\lambda(\mu+\phi))+(\mu+\lambda(\mu+\phi))\Lambda\alpha^2]+\\
	&\lambda_1^*[(\mu+\lambda(\mu+\phi))(\alpha+\mu)\phi\Lambda\alpha+\mu(\mu+\phi)a_1\alpha\beta_1\Lambda+\mu(\mu+\phi)\Lambda\alpha^2+(\alpha+\mu)\alpha\mu\phi\Lambda]+\\
	&(\mu+\phi)\mu(\alpha+\mu)\phi\Lambda\alpha-(\mu+\phi)(\alpha+\mu)\alpha\mu\phi\Lambda\\
	\Rightarrow & \lambda_1^{*4}(\alpha+\mu)a_1\alpha\lambda+\lambda_1^{*3}[a_1\alpha(\alpha+\mu)(\mu+\lambda(\mu+\phi))+a_1\alpha\lambda(\mu+\phi)(\mu+\alpha)-\\
	&\lambda a_1\alpha\beta_1\Lambda+\lambda\Lambda\alpha^2]+\lambda_1^{*2}[a_1\alpha(\mu+\lambda(\mu+\phi))(\mu+\phi)(\alpha+\mu)-\lambda\alpha\Lambda\phi(\alpha+\mu)-\\
	&\Lambda\alpha a_1\beta_1(\mu+\lambda(\mu+\phi))-(\mu+\lambda(\mu+\phi))\Lambda\alpha^2]+\lambda_1^*[(\alpha+\mu)(\mu+\phi)a_1\alpha\mu+\\
	&a_1\alpha\mu(\mu+\phi)^2(\alpha+\mu)-(\mu+\lambda(\mu+\phi))(\alpha+\mu)\phi\Lambda\alpha-\mu(\mu+\phi)a_1\alpha\beta_1\Lambda-\\
	&\mu(\mu+\phi)\Lambda\alpha^2-(\alpha+\mu)\alpha\mu\phi\Lambda]-(\mu+\phi)(\alpha+\mu)\mu\phi\Lambda\alpha-(\mu+\phi)(\alpha+\mu)\alpha\mu\phi\Lambda=0.
\end{align*}
Here, the basic reproduction number,
\begin{align*}
	\mathcal{R}_0 &=\frac{\Lambda[\alpha\beta_2+\beta_1(\gamma+\gamma_1+\mu+\delta)]}{(\alpha+\mu)(\mu+\phi)(\mu+\delta+\gamma+\gamma_1)}.\\
	1-\mathcal{R}_0 &= \frac{(\alpha+\mu)(\mu+\phi)a_1-\Lambda[\alpha\beta_2+\beta_1a_1]}{(\alpha+\mu)(\mu+\phi)a_1}.
\end{align*}
Thus, substituting the above formulas into the infectious force, at a steady state, we have the polynomial in the form,
$$b_4\lambda_1^4+b_3\lambda_1^3+b_2\lambda_2+b_1\lambda_1+b_0=0.$$
Where
\begin{flalign*}
	b_0 =&(\mu+\phi)(\alpha+\mu)\mu\phi\Lambda\alpha-(\mu+\phi)(\alpha+\mu)\alpha\mu\phi\Lambda=(1-\mathcal{R}_0)(\alpha+\mu)(\mu+\phi)a_1\alpha\mu\phi. &\\
	b_1 =&\left[\frac{(\alpha+\mu)(\mu+\phi)a_1-\Lambda[\alpha\beta_2+\beta_1 a_1]}{(\alpha+\mu)(\mu+\phi)a_1}\right]\mu\alpha(\alpha+\mu)(\mu+\phi)a_1+(\alpha+\mu)(\mu+\phi)a_1\alpha\mu(1-\mathcal{R}_{0V})-&\\
	&\mu(\mu+\phi)\Lambda\alpha^2-(\alpha+\mu)\alpha\mu\phi\Lambda &\\
	=&(1-\mathcal{R}_0)\mu\phi+(1-\mathcal{R}_{0V})\alpha\mu-\mu(\mu+\phi)\Lambda\alpha^2-(\alpha+\mu)\alpha\mu\phi\Lambda. &\\
	b_2=&\alpha(\mu+\lambda(\mu+\phi))(\alpha+\mu)(\mu+\phi)a_1(1-\mathcal{R}_0)-\lambda\alpha\Lambda\phi(\alpha+\mu)-(\mu+\lambda(\mu+\phi))\Lambda\alpha^2. &\\
	b_3=&a_1\alpha(\alpha+\mu)(\mu+\lambda(\mu+\phi))+\alpha\lambda a_1(\alpha+\mu)(\mu+\phi)(1-\mathcal{R}_0)+\lambda\Lambda\alpha^2. &\\
	b_4=&(\alpha+\mu)a_1\alpha\lambda.
\end{flalign*}
After calculating $\lambda_1^{*}$ from the polynomial and inserting positive values of $\lambda_1^{*}$ in the formulas of $S^*, V^*, E^*, I^*, R^*,$ and $T^*$, the components of the endemic equilibrium point (EEP) are determined. Moreover, it follows from the expression of polynomial coefficients that $b_4$ is always positive, and $b_0, b_1,b_2,b_3$ are positive (negative) if $\mathcal{R}_0$ and $\mathcal{R}_{0V}$ is less(greater) than one. Hence, the  following outcomes can be deduced,
\begin{itemize}
	\item Four or two endemic equilibria if $b_2>0, b_3<0, b_4>0$, and $\mathcal{R}_0<1$, $\mathcal{R}_{0V}<1$.
	\item Two endemic equilibria if $b_2>0, b_3>0, b_4<0$, and $\mathcal{R}_0<1$, $\mathcal{R}_{0V}<1$.
	\item No endemic equilibria otherwise if $\mathcal{R}_0<1$ and $\mathcal{R}_{0V}<1$.
\end{itemize}
Hence, by applying the Descartes rule of signs \cite{Stability Bound-26, Stability Bound-27}, the endemic equilibrium point of the model \eqref{new_model},  $\mathcal{E}^*$ exists iff $\mathcal{R}_0>1$.

\subsection{Forward and Backward Bifurcation Analysis of the Equilibrium States}\label{Subsection-Forward-Backward_Bif}
We proceed to establish the bifurcation condition of the equilibrium point that exists. In Section \ref{DFE point}, we have obtained the DFE point. We will now define the characteristics of the disease-free equilibrium point, $\mathcal{E}^0$. We have observed that the disease-free equilibrium state is locally asymptotically stable if $\mathcal{R}_0<1$ and unstable if $\mathcal{R}_0>1$ \cite{Bifurcation of R0-7, Bifurcation of R0-8}. It clearly reveals that when $\mathcal{R}_0=1$, then the above analysis becomes ineffective. The crucial factor $\mathcal{R}_0=1$ is equivalent to,
\begin{equation*}
	\beta_2={\beta_2}^{[TC]}=\frac{(\mu+\phi)(\alpha+\mu)(\mu+\delta+\gamma+\gamma_1)-\mu N\beta_1(\mu+\delta+\gamma+\gamma_1)}{\phi N\lambda(\alpha+\mu)+\mu N\alpha}.
\end{equation*} 
In the upcoming theorem, we will depict that the model system \eqref{new_model} undergoes a transcritical bifurcation (TC) at the disease-free equilibrium (DFE) point $\mathcal{E}_0$ when the critical parameter $\beta_2$ reaches its critical value $\beta_2={\beta_2}^{[TC]}$.
\begin{Th}
	Transcritical bifurcation of the system \eqref{new_model} obtains at the point of no disease $(\mathcal{E}_0)$ when model parameter $\beta_2$ goes through the critical value $\beta_2={\beta_2}^{[TC]}.$
\end{Th}
\begin{proof}
	When $\beta_2={\beta_2}^{[TC]}$, one of the eigenvalues becomes zero, causing the collapse of standard eigen-method analysis. In such cases, we employ Somtomayor's Theorem \cite{Stability Bound-11, Bifurcation of R0-6} to examine the characteristics of the DFE point. Let $V$ and $W$ represent the eigenvectors with respect to the zero eigenvalue of $J(\mathcal{E}_0)$ and $J[(\mathcal{E}_0)]^T$, respectively.
 So,
	\[V=\begin{pmatrix}
		-\frac{(\mu+\delta+\gamma+\gamma_1)(\alpha+\mu)}{\alpha(\mu+\phi)} \\ \frac{\mu+\delta+\gamma+\gamma_1}{\alpha}\\ 1\\ \frac{\delta}{\mu}
	\end{pmatrix},\; \text{and}\;\; W=\begin{pmatrix}
		0 \\ 2\\ \frac{(\mu+\phi)(\mu+\alpha)-\beta_1\Lambda}{\alpha(\mu+\phi)} \\ 0
	\end{pmatrix}.\]\\
	Considering the sub-model of the above  model \eqref{new_model}, by taking the SEIR compartment we have,\\
	\[F=\begin{pmatrix}
		\Lambda -(\beta_1 E+\beta_2 I)S-(\mu+\phi)S \\ (\beta_1 E+\beta_2 I)S-(\alpha+\mu)E \\ \alpha E-(\mu+\delta+\gamma+\gamma_1)I \\ \gamma I-\mu R
	\end{pmatrix}.\]
	Next, we examine the system's dynamic behavior by methodically adjusting the parameters close to each equilibrium point. We adapt Sotomayor's theorem for local bifurcation analysis \cite{Bifurcation of R0-6}. The modified theorem states that the Jacobian matrix of the modified SEIR system at the disease-free equilibrium point $\mathcal{E}_0$ appears transcritical bifurcation.\\
	It is demonstrated that the Jacobian matrix of the system $(\mathcal{E}_0, \beta^*=\beta_2)$ can be evaluated as $J=Df(\mathcal{E}_0,\beta^*)$. Here,\\
	\[J=\begin{pmatrix}
		-(\mu+\phi) & -\beta_1 S & -{\beta_2}^* S & 0\\
		0 & \beta_1 S-(\alpha+\mu) & \beta_2 S & 0\\
		0 & \alpha & -(\mu+\delta+\gamma+\gamma_1) & 0\\ 
		0 & 0 & \gamma & -\mu
	\end{pmatrix}.\]
	Where, \begin{equation*}
		{\beta_2}^*=\frac{(\mu+\phi)(\alpha+\mu)(\mu+\delta+\gamma+\gamma_1)-\mu N\beta_1(\mu+\delta+\gamma+\gamma_1)}{\phi N\lambda(\alpha+\mu)+\mu N\alpha}.
	\end{equation*}
	From the Jacobian matrix, the third $\lambda_R$ eigen value  $\lambda_I$ in the direction of $I$ is $-(\mu+\delta+\gamma+\gamma_1)$ while $\lambda_S$ and $\lambda_R$ are negative. Further, the eigenvector $V=(v_1,v_2,v_3,v_4)^T$ corresponding to $\lambda_I$ satisfying the condition $Jz=\lambda z$ then $Jz=0$ gives, \\
	\[\begin{pmatrix}
		-(\mu+\phi) & -\beta_1 S & -{\beta_2}^* S & 0\\
		0 & \beta_1 S-(\alpha+\mu) & \beta_2 S & 0\\
		0 & \alpha & -(\mu+\delta+\gamma+\gamma_1) & 0\\ 
		0 & 0 & \gamma & -\mu
	\end{pmatrix}\begin{pmatrix}
		v_1 \\ v_2\\v_3\\v_4
	\end{pmatrix}=\begin{pmatrix}
		0 \\ 0\\0\\0
	\end{pmatrix}\]
	From which we get, \\
	\begin{align*}
		-(\mu+\phi)v_1-\beta_1 S v_2-{\beta_2}^*S v_3=0\\
		(\beta_1 E+\beta_2 I)v_1
		+\{\beta_1 S-(\alpha+\mu)\}v_2+\beta_2 Sv_3=0\\
		0+\alpha v_2-(\mu+\delta+\gamma+\gamma_1)v_3=0 \\
		\gamma v_3-\mu v_4=0
	\end{align*}
	The finding of the above system of equations is, 
	$$\displaystyle V=(v_1,v_2,v_3,v_4)^T=\left(\frac{\beta_1S(\mu+\delta+\gamma+\gamma_1)\mu}{-(\alpha+\mu)\alpha\gamma} v_4-\frac{{\beta_2}^*S\mu}{\gamma(\alpha+\mu)}v_4\;,\; \frac{(\mu+\delta+\gamma+\gamma_1)\mu}{\alpha\gamma}v_4\;,\;\frac{\mu v_4}{\gamma}\;,\; v_4 \right).$$
	Similarly, the eigenvector $W=(w_1,w_2,w_3,w_4)^T$ can be written as,\\
	\[J^Tw=\begin{pmatrix}
		-(\mu+\phi) & \beta_1 E+\beta_2 I & 0 & 0\\
		-\beta_1 S & \beta_1 S-(\alpha+\mu) & \alpha & 0\\
		-{\beta_2}^*S & \beta_2 S & -(\mu+\delta+\gamma+\gamma_1) & \gamma \\
		0 &0 &0 &-\mu
	\end{pmatrix}\begin{pmatrix}
		w_1 \\w_2 \\w_3 \\w_4 
	\end{pmatrix}=0.\]
	We have the solutions, 
	\begin{align*}
		w_4=0 \;,\; w_1=\frac{(\beta_1E+\beta_2I)}{(\mu+\phi)}w_2, \; \text{and}\\
		w_3= \frac{\beta_1S(\beta_1 E+\beta_2I)+(\alpha+\mu)(\mu+\phi)}{\alpha(\mu+\phi)}w_2.
	\end{align*}
	Where $w_2$ is a free variable. Now, it is possible to write the simplified SEIR system in vector form,
	\begin{equation*}
		\frac{dX}{dt}=f(X).
	\end{equation*}
	Here, $X=(S,E,I,R)^T$ and $F=(F_1,F_2,F_3,F_4)^T$ with $F_i(i=1,2,3,4)$ then calculate\\ $\frac{\displaystyle dF}{\displaystyle d\beta_2}=F_{\beta_2}$. From which we get that, \[F_\beta^*=\begin{pmatrix}
		-SI \\ SI \\ 0\\0
	\end{pmatrix}.\] Then, \[F_\beta^*(\mathcal{E}_0,\beta^*)=\begin{pmatrix}
		0\\0\\0\\0
	\end{pmatrix},\;\; w^T\centerdot F_\beta(\mathcal{E}_0,\beta^*)=0. \]
	\[DF_\beta(\mathcal{E}_0,\beta^*)=\begin{pmatrix}
		0& -S_1 & 0 & 0\\ 0 &S_1 & 0& 0 \\ 0&0&0&0\\0&0&0&0
	\end{pmatrix},\;\;\text{where}\; S_1=\frac{\pi}{\mu}.\] \\
	\[w^T\centerdot\left[DF_\beta(\mathcal{E}_0,\beta^*)\centerdot z\right]=w_2v_2S_1 \neq 0. \]
\end{proof}
Based on Sotomayor's theorem, when the parameter $\beta_2$ gets over the bifurcation value $\beta_2^{[TC]}$, then the transcritical bifurcation arises at DFE. Then according to \cite{Stability Bound-23}
we have obtain, 
\begin{align*}
	W^T F_{\beta_2}|_{\mathcal{E}_0,\beta_2=\beta_2^{[TC]}}=0\\
	W^T DF_{\beta_2}|_{\mathcal{E}_0,\beta_2=\beta_2^{[TC]}}V=-\frac{\Lambda(\mu+\delta+\gamma+\gamma_1)}{\alpha(\mu+\phi)}\neq 0\\
	W^TD^2F_{\beta_2}|_{\mathcal{E}_0,\beta_2=\beta_2^{[TC]}}(V,V)=\frac{-2(\mu+\delta+\gamma+\gamma_1)^2(\mu+\alpha)\beta_1}{\alpha^2(\mu+\phi)}\neq 0
\end{align*}
Consequently, the framework undergoes transcritical bifurcation as the rate of infection $I$ class $(\beta_2)$ exceeds the critical value $\beta_2=\beta_2^{[TC]}$. There endures a critical infection rate for the $I$ class, beyond which the endemic diseases will spread throughout the civilization; however, below this threshold, the disease is easily curable.
\begin{figure}[H]
	\centering  
	\subfloat[]{\includegraphics[width=2.5 in]{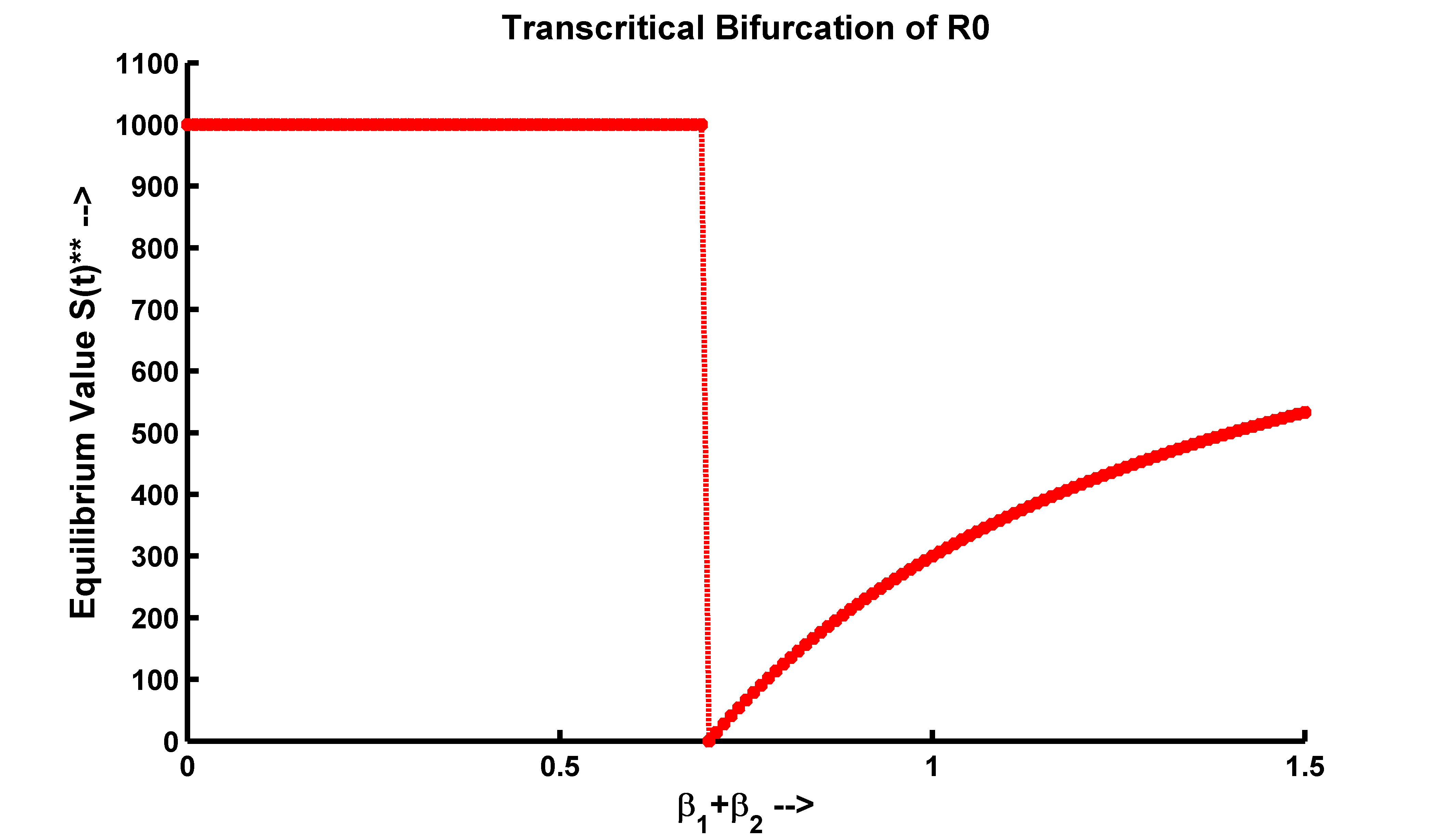}}
	\subfloat[]{\includegraphics[width=2.5 in]{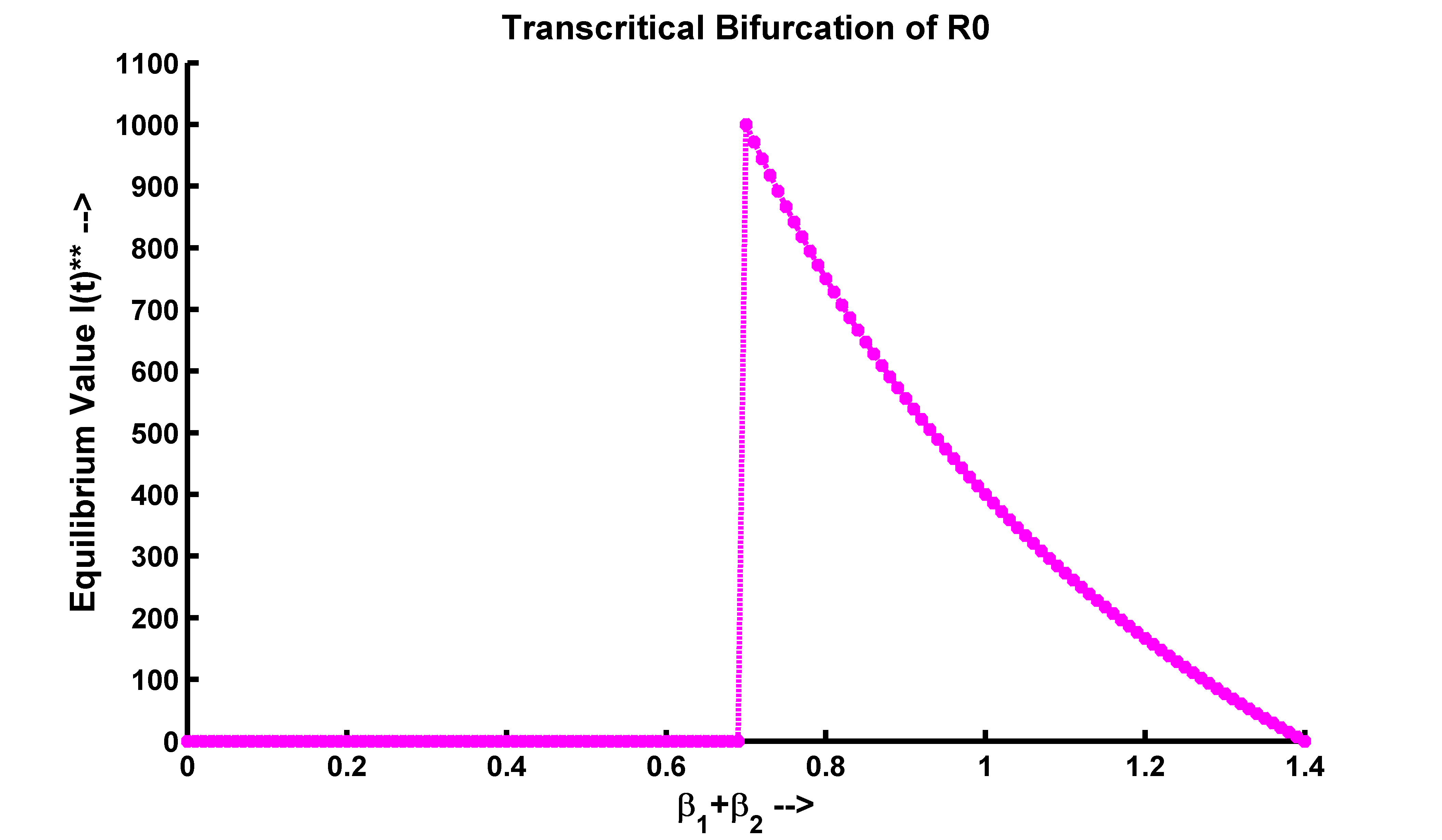}}\\
	\subfloat[]{\includegraphics[width=2.5 in]{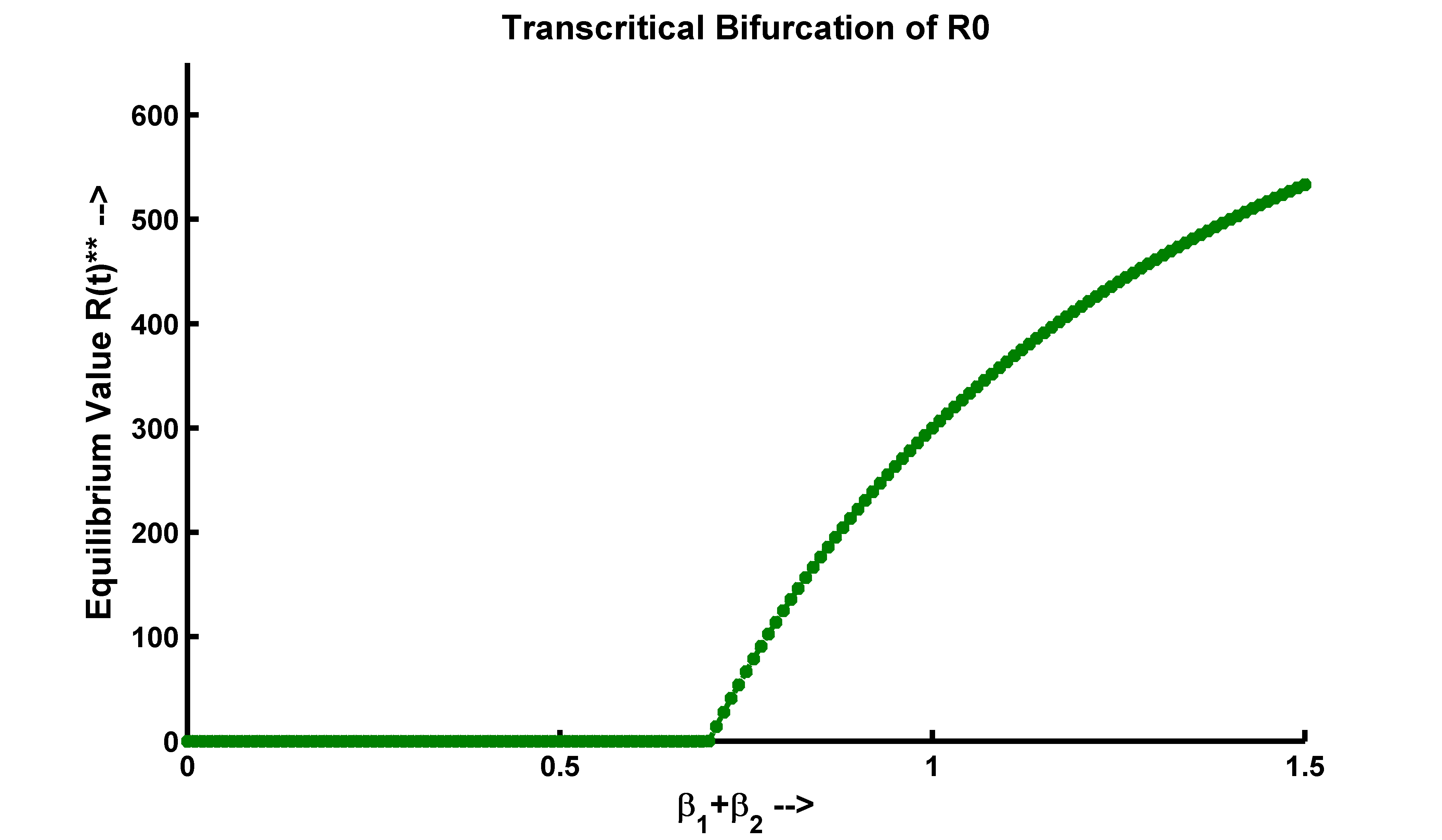}}
	\subfloat[]{\includegraphics[width=2.5 in]{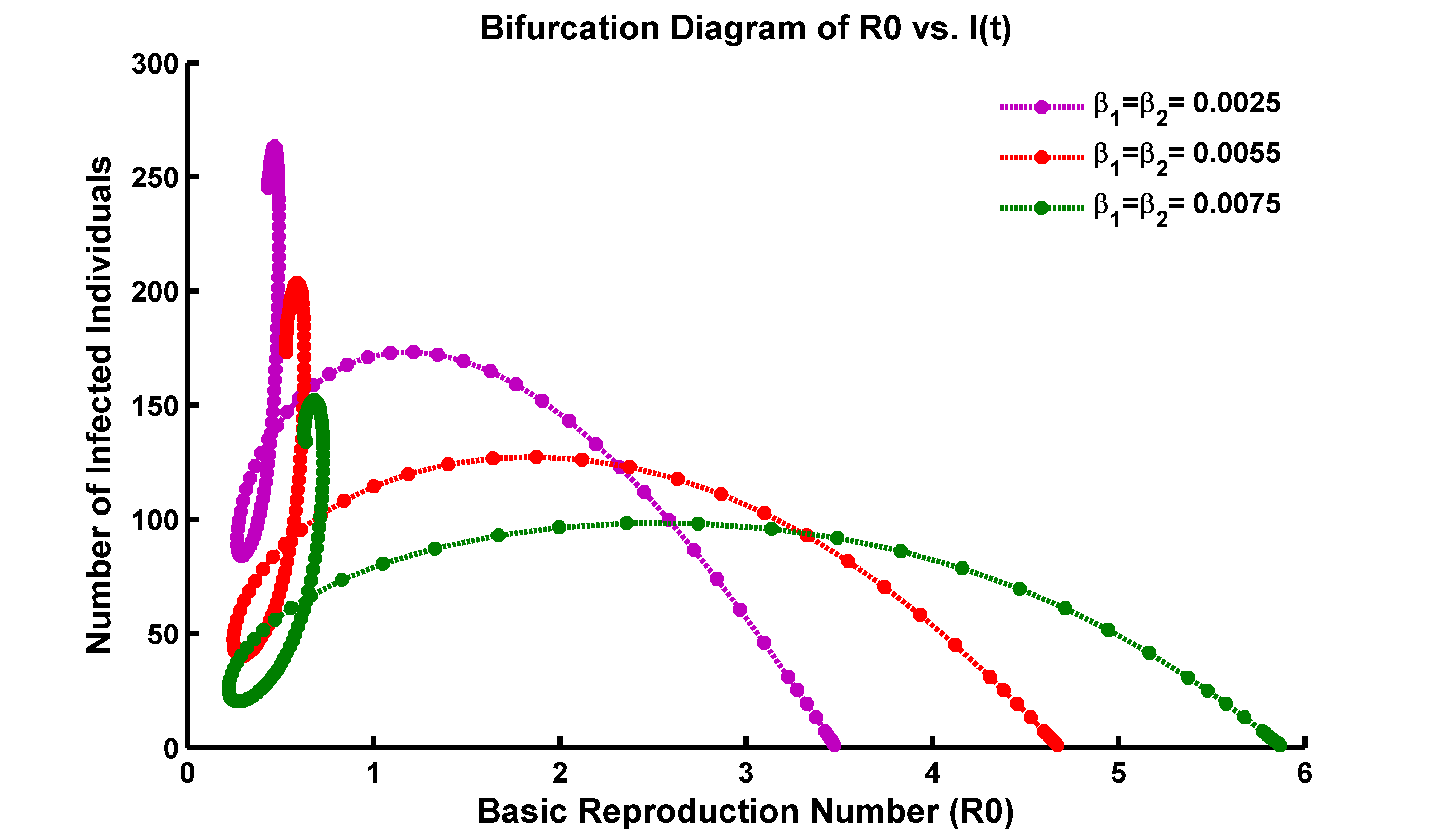}}
	\caption{{Transcritical bifurcation of the (a) $S(t)$ compartment, (b) $I(t)$ compartment, (c) $R(t)$ compartment, and (d) bifurcation of $I(t)$ concerning $\mathcal{R}_0$ at the disease-free equilibrium point, where all the parameters are taken from Table \ref{tableparameter}}.}
	\label{transcritical-bifurcation}
\end{figure}
\noindent
In a transcritical bifurcation of the basic reproduction number ($\mathcal{R}_0$), the scenario of the susceptible population versus the force of infection can be described as follows,\\
Before the bifurcation point, $\mathcal{R}_0$ is typically below a critical threshold value. The force of infection is relatively low, meaning that the rate at which susceptible individuals become infected is not very high. The susceptible population gradually decreases over time as some individuals get infected and transition into the infected or recovered compartments. Figure \ref{transcritical-bifurcation}(a) indicates that when the force of infection $\lambda_1$ lies in 0 to 0.75 the $S(t)$ population goes parallelly. After crossing $\lambda_1$ value 0.75, $S(t)$ population increase gradually. At the bifurcation point, $\mathcal{R}_0$ reaches a critical threshold value. This threshold represents a transition in the dynamics of the infectious disease system. The force of infection $\lambda_1$ increases abruptly, causing a sudden rise in the rate at which susceptible individuals become infected (Figure \ref{transcritical-bifurcation}(a)). This rise can lead to a rapid spread of the infection throughout the population. After the bifurcation, the force of infection remains high, and the susceptible population decreases at a faster rate. The number of infected individuals rises significantly, potentially resulting in an epidemic or an outbreak of the disease. The system enters a new equilibrium state characterized by a higher number of infected individuals compared to the pre-bifurcation phase.

\noindent
From Figure \ref{transcritical-bifurcation}(b), the transcritical bifurcation of the basic reproduction number ($\mathcal{R}_0$) refers to a critical point in a disease transmission model where the dynamics of the infected population can undergo a qualitative change. At this bifurcation point, the force of infection, which represents the rate at which susceptible individuals become infected, plays a significant role in determining the behavior of the infected population. Before the bifurcation, when the force of infection is relatively low (lies at 0 to 0.65), the infected population remains at a low level or may even die out. The disease is not able to sustain itself in the population, and there is no sustained transmission. At the transcritical bifurcation, a critical threshold of the force of infection is reached. Beyond this threshold, the force of infection becomes sufficient to sustain the disease transmission, leading to an increase in the infected population. The infected population transitions from a low or dying-out state to a persistent or growing state. After the bifurcation, when the force of infection exceeds the critical threshold (grows from 0.65 and progresses to value 1.4), the infected population grows exponentially or stabilizes at a higher level. The disease becomes endemic, with sustained transmission and a significant proportion of the population being infected.\\
\noindent
In a transcritical bifurcation of the threshold value ($\mathcal{R}_0$), the scenario of the recovered population vs. force of infection ($\lambda_1$) can be described as follows,\\
Figure \ref{transcritical-bifurcation}(c) reflects that when the force of infection is low (lies in 0 to 0.65), the recovered population $R(t)$ is also low. This suggests that a relatively small number of people have contracted the sickness, and then recovered from it. The disease may not be spreading efficiently, and the epidemic is unlikely to sustain itself. At a critical point, the force of infection reaches a threshold value. This is the point where the bifurcation occurs. In transcritical bifurcation, this threshold value is denoted as $\mathcal{R}_0^*$. At this point, the disease transitions from a non-endemic state to an endemic state, meaning it becomes self-sustaining. Beyond the critical point (when $\lambda_1$ exceeds 0.65), as the force of infection increases, the recovered population starts to rise exponentially. This indicates that more and more individuals are becoming infected and subsequently recovering from the disease. The epidemic is spreading efficiently, and the recovered population grows over time.
From Figure \ref{transcritical-bifurcation}(d), the relationship between the infected population and the basic reproduction number can be described as follows,\\
When $\mathcal{R}_0>1$, it implies that, on average, each infected person is spreading the infection to more than one susceptible individual. In this scenario, the infected population tends to increase over time. The epidemic spreads in the population, and if no interventions or control measures are implemented, it has the potential to cause a large-scale outbreak. When $\mathcal{R}_0=1$ each infected person, on average, transmits the infection to precisely one susceptible individual. In this case, the infected population tends to remain stable over time. The epidemic is said to be in an endemic state, meaning it persists within the population without causing exponential growth or decline. When $\mathcal{R}_0<1$, each infected person, on average, transmits the infection to fewer than one susceptible individual. In this situation, the infected population tends to decline over time. The epidemic subsides and eventually dies out, as the number of new infections generated by each infected individual is insufficient to sustain the transmission within the population. Thus, $I(t)$ individuals grows with the increasing amount of $\mathcal{R}_0$. This indicates that (Figure \ref{transcritical-bifurcation}(d)), the basic reproduction number is a crucial determinant of the behavior of infectious diseases. It provides insights into the potential for outbreaks, the sustainability of transmission, and the effectiveness of control measures in reducing the infected population.\\
From Figure \ref{transcritical-bifurcation}, it is noticeable that when $\mathcal{R}_0<1$, the system described by \eqref{new_model} shows only a stable disease-free equilibrium point. Conversely, when $\mathcal{R}_0>1$, a stable endemic equilibrium occurs, leading to the instability of the disease-free equilibrium (DFE). This instability, denoting a transition from stable to unstable, happens precisely at the critical point $\mathcal{R}_0=1$, leading to a transcritical bifurcation at DFE points. Therefore, when the model parameter $\beta_2$ exceeds its critical value $\beta_2^{[TC]}$, the stability of the disease-free equilibrium shifts from a stable state to an unstable one.\\
In the next part, we explore the phenomenon of backward bifurcation of the modified system of \eqref{new_model} as follows,
\begin{align} \label{Bifur_New_Model}
	f_1=x_1'=&\Lambda-(\beta_1x_3+\beta_2x_4)x_1-(\mu+\phi)x_1+\phi_1x_2. \nonumber\\
	f_2=x_2'=&\phi x_1-(1-\varepsilon)(\beta_1x_3+\beta_2x_4)x_2-(\mu+\phi_1)x_2. \nonumber\\
	f_3=x_3'=&(\beta_1x_3+\beta_2x_4)x_1-(\alpha+\mu)x_3. \nonumber\\
	f_4=x_4'=&\alpha x_3+(1-\varepsilon)(\beta_1x_3+\beta_2x_4)x_2-(\mu+\delta+\gamma+\gamma_1)x_4. \nonumber\\
	f_5=x_5'=&\gamma x_4-\mu x_5. \nonumber\\
	f_6=x_6'=&\gamma_1x_4-\mu x_6.
\end{align}
Where $(S,V,E,I,R,T)=(x_1,x_2,x_3,x_4,x_5,x_6).$ Here, we have considered that, after vaccination, a portion of the vaccinated population progresses to susceptible compartments by losing immunity at a constant rate $\phi_1$. Firstly, we examine bifurcation analysis with the help of the center manifold theorem \cite{Bifurcation of R0-7, Bifurcation of R0-8}. We investigate the properties of the equilibrium solutions in the vicinity of the bifurcation point $x=x_0$ where $\mathcal{R}_0=1$. As $\mathcal{R}_0$ can be difficult as a direct bifurcation parameter, we introduce a new parameter $\mu_1$ for this purpose. We define $\mu_1$ as a bifurcation parameter such that $\mathcal{R}_0 < 1$ for $\mu_1 < 0$ and $\mathcal{R}_0 > 1$ for $\mu_1 > 0$. Moreover, we ensure that $x_0$ persists in a Disease-Free Equilibrium (DFE) for all values of $\mu_1$. Consider the structure, $$\dot{x}=f(x,\mu_1).$$
Here, the restriction is that $f$ possesses at least two continuous derivatives with respect to both $x$ and $\mu_1$. The line $(x_0,\mu_1)$ represents the Disease-Free Equilibrium (DFE), and the local stability of the DFE undergoes a shift at the point $(x_0,0)$. Applying center manifold theory, we demonstrate the existence of non-trivial (endemic) equilibria close to the bifurcation point $(x_0,0)$. Before exploring these findings, we introduce some notation and gather relevant facts.\\
We represent the partial derivative of $f$ with respect to $x$ at $x=x_0,\;\mu_1=0$ as $D_xf(x_0,0)$. Assuming that $D_xf(x_0,0)$ has simple zero eigenvalues, we define $v$ and $w$ as the corresponding left and right null vectors such that $vw=1$. We also ensure that the remaining eigenvalues of $D_xf(x_0,0)$ possess negative real parts. Let's,
\begin{align*}
	a=&\frac{v}{2}D_{xx}f(x_0,0)w^2=\frac{1}{2}\sum_{i,j,k=1}^{n}v_iw_jw_k\frac{\partial^2f_i}{\partial x_j \partial x_k}(x_0,0).\\
	b=&vD_{x\mu_1}f(x_0,0)w=\sum_{i,j=1}^{n}v_iw_j\frac{\partial^2f_i}{\partial x_j\partial \mu_1}(x_0,0).
\end{align*}
We will demonstrate that the sign of $a$ dictates the characteristics of the endemic equilibria near the bifurcation point. Before delving into this analysis, it is essential to mention that the expression for $a$ can be revised using the outcomes established in the preceding sections. We utilize center manifold theory to theoretically divine the existence of the backward bifurcation phenomenon in the model \eqref{Bifur_New_Model}, as outlined below.
\begin{Th}\label{Center Manifold Theorem}

	Consider the following system of ordinary differential equations, incorporating a parameter $\phi$,
	\begin{equation}\label{center_reqn}
		\frac{dx}{dt}=f(x,\phi), f:\mathbb{R}^n\times\mathbb{R} \rightarrow \mathbb{R}, f\in \mathbb{C}^2(\mathbb{R}^n\times\mathbb{R}).
	\end{equation}
	Without departing generality, we assume that $x=0$ serves as an equilibrium for the system \eqref{center_reqn} across all parameter values of $\phi$. Suppose that
	\begin{enumerate}
		\item [(a)] The matrix $A=D_xf(0,0)$ represents the linearized matrix of system \eqref{center_reqn} at the equilibrium $x=0$, where $\phi$ evaluated at $0$. In this case, $0$ is a simple eigenvalue of $A$, and all other eigenvalues of $A$ reflect negative real parts.

		\item [(b)] The matrix $A$ possesses a non-negative right eigenvector $w$ and a left eigenvector $v$ associated with the zero eigenvalue. Let $f_k$ denote the $k^{th}$ component of $f$, and
		\begin{align*}
			a= \sum_{k,i,j=1}^{n} v_kw_iw_j\frac{\partial^2f_k}{\partial x_i\partial x_j}(0,0),\;\; b= \sum_{k,i=1}^{n} v_kw_i\frac{\partial^2f_k}{\partial x_i\partial \beta}(0,0).
		\end{align*}
		Subsequently, the signs of $a$ and $b$ have a complete influence on the local dynamics of system \eqref{center_reqn} around zero.
		\begin{enumerate}
			\item [(i)] Case-1: For $\phi<0$ with $|\phi|\ll1$, the equilibrium at $x=0$ is locally asymptotically stable and a positive unstable equilibrium exists if $a>0$ and $b>0$. When $0<\phi\ll 1 $, a negative and locally asymptotically stable equilibrium is present, and the equilibrium at $x=0$ is unstable.
			\item[(ii)] Case-2: In the scenario where $a<0$ and $b<0$, for $|\phi|\ll1$, the equilibrium at $x=0$ is unstable. There is a negative unstable equilibrium and the equilibrium at $x=0$ becomes locally asymptotically stable when $0<\phi\ll1$.

			\item[(iii)] Case-3: With $a>0$ and $b<0$, the equilibrium at $x=0$ is unstable and a locally asymptotically stable negative equilibrium emerges when $\phi<0$ with $|\phi|\ll1$. When $0<\phi\ll1$, the equilibrium at $x=0$ becomes stable, and a positive unstable equilibrium emerges.
			\item[(iv)] Case-4: $\phi$ changes from negative to positive, $x=0$ changes its stability from stable to unstable when $a<0,\;b>0$. A negative unstable equilibrium correspondingly becomes positive and locally asymptotically stable.
		\end{enumerate}
	\end{enumerate}
	In particular, if $a>0$ and $b>0$, a backward bifurcation occurs at $\phi=0$. These conditions, delineating the bifurcation locally at $\mathcal{R}_0=1$, align with the scenarios illustrated in Figure \ref{forward-backward-bifurcation}. Specifically, conditions (ii) and (iv) denote a forward bifurcation scenario, while conditions (i) and (iii) indicate the occurrence of a backward bifurcation \cite{Bifurcation of R0-3, Bifurcation of R0-4}.
\end{Th}

Let, $x=(x_1,x_2,x_3,x_4,x_5,x_6)^T=(S,V,E,I,R,T)^T$. Thus, the model \eqref{Bifur_New_Model} is modified with the previous model \eqref{new_model} by taking reinfection term $\phi_1$ which goes backward from $V$ compartment to $S$. Thus, model \eqref{Bifur_New_Model} is in the form $\displaystyle \frac{\displaystyle dx}{\displaystyle dt}=f(x)$, with $f(x)=(f_1(x),f_2(x),\cdots f_6(x))$. The Jacobian matrix of the system \eqref{Bifur_New_Model} at DFE $\mathcal{E}^0$ is given as, 
\begin{align*}
	&J^*(E_{0V})|_{\beta_2=\beta_2^*}=
	\begin{pmatrix}
		a_{11} & \phi_1 &-\beta_1x_1 &-\beta_2x_1& 0 & 0 \\
		\phi & a_{22} & -(1-\varepsilon)\beta_1x_1 & -(1-\varepsilon)\beta_2x_2& 0 &0 \\
		\beta_1x_3+\beta_2x_4 & 0 & a_{33} & \beta_2x_1 & 0 &0 \\
		0 & a_{42} & \alpha+\lambda\beta_1x_2 &a_{44}&0&0\\
		0&0&0&\gamma &-\mu & 0\\
		0&0&0&\gamma_1 &0 &-\mu
\end{pmatrix} \end{align*}
where $a_{11}=-(\beta_1x_3+\beta_2x_4)-(\mu+\phi)$, $a_{22}=-\lambda(\beta_1x_3+\beta_2x_4)-(\mu+\phi_1)$,
$a_{33}=\beta_1x_1-(\alpha+\mu)$, $a_{42}=\lambda(\beta_1x_3+\beta_2x_4)$, and
$a_{44}=\lambda\beta_2x_2-(\mu+\delta+\gamma+\gamma_1)$.\\
We have considered contact rate $\beta_2=\beta_2^*$ as the bifurcation parameter, setting $\mathcal{R}_0=1$ gives,
\begin{align*}
&S_0\alpha\beta_2^*+V_0\beta_2\lambda(\alpha+\mu)+S_0\beta_1(\mu+\delta+\gamma+\gamma_1)=(\alpha+\mu)(\mu+\delta+\gamma+\gamma_1)\\
	&\beta_2^*=\frac{(\alpha+\mu)(\mu+\delta+\gamma+\gamma_1)-S_0\beta_1(\mu+\delta+\gamma+\gamma_1)}{V_0\lambda(\alpha+\mu)+S_0\alpha}\\
	& \beta_2=\beta^*=\frac{(\mu+\phi)(\alpha+\mu)(\mu+\delta+\gamma+\gamma_1)-\mu N\beta_1(\mu+\delta+\gamma+\gamma_1)}{\phi N\lambda(\alpha+\mu)+\mu N\alpha} \\ &[\text{By putting}\; S_0\; \text{and} \;V_0 \;\text{at the DFE value.}]
\end{align*}
At the DFE point, we have $\displaystyle x_1=\frac{\Lambda}{\mu+\phi}=\frac{\mu N}{\mu+\phi}=\frac{k_1}{k_2}\;,\; x_2=\frac{N\phi}{\mu+\phi}=\frac{k_3}{k_2}.$\\
Here, we let $k_1=\mu N$, $k_3=N\phi$, and $k_2=\mu+\phi$. Hence,
\begin{align*}
	J^*(E_{0V})|_{\beta_2=\beta_2^*}=
	&\begin{pmatrix}
		-(\mu+\phi)&\phi_1&\frac{-\beta_1k_1}{k_2}&\frac{-\beta_2k_1}{k_2}& 0&0\\
		\phi & -(\mu+\phi_1)&-\lambda\beta_1\frac{k_3}{k_2}&-\lambda\beta_2\frac{k_3}{k_2} & 0&0\\
		0 & 0& \beta_1\frac{k_1}{k_2}-(\alpha+\mu) &\beta_2\frac{k_1}{k_2}&0&0\\
		0&0& \alpha+\lambda\beta_1\frac{k_3}{k_2} &a_{44}&0&0\\
		0&0&0&\gamma&-\mu &0\\
		0&0&0&\gamma_1 &0&-\mu
	\end{pmatrix}
\end{align*}
where $a_{44}=\lambda\beta_2\frac{\displaystyle k_3}{\displaystyle k_2}-(\mu+\delta+\gamma+\gamma_1).$\\
The transformed system \eqref{Bifur_New_Model} at the DFE $\mathcal{E}^0$ calculated for $\beta_2=\beta_2^*$ has a hyperbolic equilibrium point i.e. a simple eigenvalue with $0$ and all other eigenvalues have a negative real part. We therefore apply the Centre Manifold  Theorem in order to analyze the dynamics of \eqref{Bifur_New_Model} near $\beta_2=\beta_2^*$.\\
{The Jacobian of \eqref{Bifur_New_Model} at  $\beta_2=\beta_2^*$ denoted by $J(\mathcal{E}^0)|_{\beta_2=\beta^*}$. Now, the right $(w)$ and left $(v)$ eigenvector are computed from $J(\mathcal{E}^0)|_{\beta_2=\beta^*}$ (associated with zero eigenvalues) are given as}\\
\textbf{Right Eigenvector:} $J^*(E_{0V})\underline{w}=0.$
\; Where \[\underline{w}=\begin{pmatrix}
	w_1\\w_2\\w_3\\w_4\\w_5\\w_6
\end{pmatrix}, \;\; \text{and}\;\; \underline{0}=\begin{pmatrix}
	0\\0\\0\\0\\0\\0
\end{pmatrix}.\]
{For zero eigenvalue we have obtained,}
\begin{align}\label{case_1Bif}
	&-k_2w_1+\phi_1w_2-\frac{\beta_1k_1}{k_2}w_3-\frac{\beta_2k_1}{k_2}w_4=0. \nonumber\\
	&\phi w_1-(\mu+\phi_1)w_2-\lambda\beta_1\frac{k_3}{k_2}w_3-\lambda\beta_2\frac{k_3}{k_2}w_4=0.\nonumber\\
	&\left(\beta_1\frac{k_1}{k_2}-(\alpha+\mu)\right)w_3+\beta_2\frac{k_1}{k_2}w_4=0. \nonumber\\
	&\left(\alpha+\lambda\beta_1\frac{k_3}{k_2}\right)w_3+\left(\lambda\beta_2\frac{k_3}{k_2}-(\mu+\delta+\gamma+\gamma_1)\right)w_4=0. \nonumber\\
	&\gamma w_4-\mu w_5=0. \nonumber\\
	&\gamma_1 w_4-\mu w_6=0.
\end{align}
Let, $w_4>0$ be the free variable.  Here,
\begin{align}\label{case_w_5 and w_6}
	w_5= \frac{\gamma w_4}{\mu},\; \text{and}\; w_6=\frac{\gamma_1 w_4}{\mu}.
\end{align}
Also, from the third equation of \eqref{case_1Bif} we have,
\begin{align}\label{case_w_3}
	 w_3 =\frac{\beta_2k_1w_4}{k_2(\alpha+\mu)-\beta_1k_1}.
\end{align}
Multiply first equation of \eqref{Bifur_New_Model} by $\phi$, and second equation of \eqref{Bifur_New_Model} by $k_2$,
\begin{align*}
	-k_2^2\phi w_1+\phi\phi_1k_2w_2-\beta_1k_1\phi w_3-\beta_2k_1\phi w_4=0.\\
 \phi k_2^2 w_1-k_2^2(\mu+\phi_1)w_2-\lambda\beta_1k_3k_2w_3-\lambda\beta_2k_2k_3w_4=0.
\end{align*}
By adding these two equations we obtain,
\begin{align}\label{case_w_2}
	&w_2\{k_2\phi\phi_1-k_2^2(\mu+\phi_1)\}-w_3(\beta_1k_1\phi+\lambda\beta_1k_3k_2)-w_4(\beta_2k_1\phi+\beta_2k_3k_2\lambda)=0 \nonumber\\
	\Rightarrow & w_2\{k_2\phi\phi_1-k_2^2(\mu+\phi_1)\}-\frac{\beta_2k_1(\beta_1k_1\phi+\lambda\beta_1k_3k_2)}{k_2(\alpha+\mu)-\beta_1k_1}w_4-w_4(\beta_2k_1\phi+\beta_2k_3k_2\lambda)=0 \nonumber\\
	\Rightarrow &w_2= \frac{\beta_2k_1(\beta_1k_1\phi+\lambda\beta_1k_3k_2)+(\beta_2k_1\phi+\beta_2k_3k_2\lambda)(k_2(\alpha+\mu)-\beta_1k_1)}{(k_2\phi\phi_1-k_2^2(\mu+\phi_1))(k_2(\alpha+\mu)-\beta_1k_1)}w_4=\frac{C_{11}}{C_{22}}w_4.
\end{align}
Substituting the expressions of $w_2$ and $w_3$ in the first equation of \eqref{case_1Bif} we get,
\begin{align}\label{case_w_1}
	w_1=\frac{1}{k_2}\left[\frac{\phi C_{11}}{C_{22}}-\frac{\beta_1k_1^2\beta_2}{k_2^2(\alpha+\mu)-\beta_1k_1k_2}-\frac{\beta_2k_1}{k_2}\right]w_4.
\end{align}
Thus, all parameters of {{right eigenvectors}} $w_1$, $w_2$, $w_3$, $w_5$, and $w_6$ can be expressed in terms of $w_4$.\\ 
\textbf{Left eigenvector:} Similarly, from $J(E_{0V})|_{\beta_2=\beta^*}$ we obtain,
$$[v1,v2,v3,v4,v5,v6]J^*(E_{0V})=[0,0,0,0,0,0].$$
Then, 
\begin{align} \label{case-2Bif}
	&-k_2v_1+\phi v_2=0, \nonumber\\
	&\phi_1v_1-(\mu+\phi_1)v_2=0,\nonumber\\
	&\frac{-\beta_1k_1}{k_2}v_1-\frac{\lambda\beta_1k_3}{k_2}v_2+\left(\frac{\beta_1k_1}{k_2}-(\alpha+\mu)\right)v_3+\left(\alpha+\lambda\beta_1\frac{k_3}{k_2}\right)v_4=0, \\
	&-\frac{\beta_2k_1}{k_2}v_1-\frac{\lambda\beta_2k_3}{k_2}v_2+\frac{\beta_2k_1}{k_2}v_3 +\left[\lambda\beta_2\frac{k_3}{k_2}-(\mu+\delta+\gamma+\gamma_1)\right]v_4+\gamma v_5+\gamma_1v_6 =0,\nonumber\\
	&-\mu v_5=0, \;\; \text{and} \;\; -\mu v_6=0.\nonumber
\end{align}
Therefore, $\displaystyle v_1=\frac{\displaystyle \phi v_2}{\displaystyle k_2}$, $v_5=0,\;\text{and}\; v_6=0$.
Putting $v_5=v_6=0$ on the third and fourth equation of \eqref{case-2Bif},
\begin{align*}
	&\beta_1k_1v_1+\lambda\beta_1k_3v_2-(\beta_1k_1-k_2(\alpha+\mu))v_3-(\alpha k_2+\lambda\beta_1k_3)v_4=0.\\
	&-\beta_2k_1v_1-\lambda\beta_2k_3v_2+\beta_2k_1v_3+[\lambda\beta_2k_3-k_2(\mu+\delta+\gamma+\gamma_1)]v_4=0.
\end{align*}
Now, multiplying third equation by $\beta_2$ and fourth equation by $\beta_1$ of \eqref{case-2Bif},
\begin{align*}
	&\beta_2\beta_1k_1v_1+\lambda\beta_1\beta_2k_3v_2-(\beta_1\beta_2k_1-\beta_2k_2(\alpha+\mu))v_3-\beta_2(\alpha k_2+\lambda\beta_1k_3)v_4=0.\\
	&-\beta_1\beta_2k_1v_1-\lambda\beta_1\beta_2k_3v_2+\beta_2\beta_1k_1v_3+\beta_1[\lambda\beta_2k_3-k_2(\mu+\delta+\gamma+\gamma_1)]v_4=0.
\end{align*}
By doing addition,
\begin{align*}
	[\beta_2k_2(\alpha+\mu)]v_3=[\beta_2(\alpha k_2+\lambda\beta_1k_3)-\beta_1(\lambda\beta_2k_3-k_2(\mu+\delta+\gamma+\gamma_1))]v_4.
\end{align*}
That implies,
\begin{align}\label{case_bif_v3}
	v_3=\frac{C_{33}}{C_{44}}v_4.
\end{align}
Where $C_{33}=[\beta_2(\alpha k_2+\lambda\beta_1k_3)-\beta_1(\lambda\beta_2k_3-k_2(\mu+\delta+\gamma+\gamma_1))]$, and $C_{44}=[\beta_2k_2(\alpha+\mu)].$
In the third equation of \eqref{case-2Bif} we have,
\begin{align*}
	\left[\frac{\beta_1k_1\phi}{k_2^2}+\frac{\lambda\beta_1k_3}{k_2}\right]v_2=\left[\left(\frac{\beta_1k_1}{k_2}-(\alpha+\mu)\right)\frac{C_{33}}{C_{44}}+\left(\alpha+\lambda\beta_1\frac{k_3}{k_2}\right)\right]v_4.
\end{align*}
Thus,
\begin{align}\label{case_bif_v2}
	v_2=\frac{C_{32}}{C_{42}}v_4.
\end{align}
Where $C_{32}=\left[\left(\frac{\displaystyle \beta_1k_1}{\displaystyle k_2}-(\alpha+\mu)\right)\frac{\displaystyle C_{33}}{\displaystyle C_{44}}+\left(\alpha+\lambda\beta_1\frac{\displaystyle k_3}{\displaystyle k_2}\right)\right]$, and $C_{42}= \left[\frac{\displaystyle \beta_1k_1\phi}{\displaystyle k_2^2}+\frac{\displaystyle \lambda\beta_1k_3}{\displaystyle k_2}\right].$\\
Hence, 
\begin{align}\label{case_bif_v1}
	v_1=\frac{\phi}{k_2}v_2=\frac{\phi C_{32}}{k_2 C_{42}}v_4.
\end{align}
Thus, $v_1$, $v_2$, and $v_3$ can be expressed in terms of $v_4$. Hence, in the {left eigenvector}, we can assume $v_4$ as a free variable. Since $v_4$ is a free variable, we evaluate the second-order partial derivatives $f_i$ at the disease-free equilibrium point $\mathcal{E}_0$ to exhibit the existence of backward bifurcation. Furthermore, in $f_4$ there are $\beta_1$ and $\beta_2$ terms (which are contact rate related to disease transmission). So, the associate non-zero second partial derivative of the model \eqref{Bifur_New_Model} evaluated at $(E_{0V},\beta^*).$
Now, $$f_4=\alpha x_3+(1-\varepsilon)(\beta_1x_3+\beta_2x_4)x_2-(\mu+\delta+\gamma+\gamma_1)x_4.$$
Taking derivative with respect to $x_3$ and $x_4$ with possible other combinations, as $x_3$ and $x_4$ indicates $E$ and $I$ terms.
\begin{align*}
	\frac{\partial f_4}{\partial x_3}&=\alpha+(1-\varepsilon)\beta_1 x_2^0
	=\alpha+\frac{(1-\varepsilon)(x_1+x_2+x_3+x_4+x_5+x_6)\phi\beta_1}{\mu+\phi},\\
	\frac{\partial^2 f_4}{\partial x_3\partial x_1}&=\frac{\partial^2 f_4}{\partial x_1\partial x_3}=\frac{(1-\varepsilon)\phi\beta_1}{\mu+\phi},\;
	 \frac{\partial^2 f_4}{\partial x_3\partial x_2}=\frac{\partial^2 f_4}{\partial x_2\partial x_3}=\frac{(1-\varepsilon)\phi\beta_1}{\mu+\phi},\;
	 \frac{\partial^2 f_4}{\partial x_3\partial x_3}=\frac{\partial^2 f_4}{\partial x_3\partial x_3}=\frac{(1-\varepsilon)\phi\beta_1}{\mu+\phi},\\
	 \frac{\partial^2 f_4}{\partial x_3\partial x_4}&=\frac{\partial^2 f_4}{\partial x_4\partial x_3}=\frac{(1-\varepsilon)\phi\beta_1}{\mu+\phi},\;
	 \frac{\partial^2 f_4}{\partial x_3\partial x_5}=\frac{\partial^2 f_4}{\partial x_5\partial x_3}=\frac{(1-\varepsilon)\phi\beta_1}{\mu+\phi}, \; 
 \frac{\partial^2 f_4}{\partial x_3\partial x_6}=\frac{\partial^2 f_4}{\partial x_6\partial x_3}=\frac{(1-\varepsilon)\phi\beta_1}{\mu+\phi}.
\end{align*}
Now, $\frac{\displaystyle \partial f_4}{\displaystyle \partial x_4}=(1-\varepsilon)\beta_2 x_2^0-(\mu+\delta+\gamma+\gamma_1)$.
Thus, putting the above $E_{0V}$ value,
\begin{align*}
	\frac{\partial f_4}{\partial x_4}&=\frac{(1-\varepsilon)\beta_2(x_1+x_2+x_3+x_4+x_5+x_6)\phi}{\mu+\phi}-(\mu+\delta+\gamma+\gamma_1),\\
	 \frac{\partial^2 f_4}{\partial x_4\partial x_1}&=\frac{\partial^2 f_4}{\partial x_1\partial x_4}=\frac{(1-\varepsilon)\phi\beta_2}{\mu+\phi},\;
	\frac{\partial^2 f_4}{\partial x_4\partial x_2}=\frac{\partial^2 f_4}{\partial x_2\partial x_4}=\frac{(1-\varepsilon)\phi\beta_2}{\mu+\phi},\;
	 \frac{\partial^2 f_4}{\partial x_4\partial x_3}=\frac{\partial^2 f_4}{\partial x_3\partial x_4}=\frac{(1-\varepsilon)\phi\beta_2}{\mu+\phi},\\
	 \frac{\partial^2 f_4}{\partial x_4\partial x_4}&=\frac{\partial^2 f_4}{\partial x_4\partial x_4}=\frac{(1-\varepsilon)\phi\beta_2}{\mu+\phi},\;
	 \frac{\partial^2 f_4}{\partial x_4\partial x_5}=\frac{\partial^2 f_4}{\partial x_5\partial x_4}=\frac{(1-\varepsilon)\phi\beta_2}{\mu+\phi},\;
	 \frac{\partial^2 f_4}{\partial x_4\partial x_6}=\frac{\partial^2 f_4}{\partial x_6\partial x_4}=\frac{(1-\varepsilon)\phi\beta_2}{\mu+\phi}.
\end{align*}
Now, we determine the bifurcation coefficients of $\bar{a}$ and $\bar{b}$ defined in Theorem \ref{Center Manifold Theorem}, stated by Castillo-Chavez and Song which is given as follows,
\begin{align*}
	\bar{b}=&\sum_{k,i=1}^{6}v_kw_i\frac{\partial^2 f_k(E_{0V},\beta_2^*=0)}{\partial x_i\partial \beta_2^*}\\
	&=v_1w_4\frac{\partial^2f_1}{\partial x_4\partial\beta_2^*}+v_2w_4\frac{\partial^2f_2}{\partial x_4\partial\beta_2^*}+v_3w_4\frac{\partial^2f_3}{\partial x_4\partial\beta_2^*}+v_4w_4\frac{\partial^2f_4}{\partial x_4\partial\beta_2^*}\\
	&+v_1w_3\frac{\partial^2f_1}{\partial x_3\partial\beta_2^*}+v_2w_3\frac{\partial^2f_2}{\partial x_3\partial\beta_2^*}+v_3w_3\frac{\partial^2f_3}{\partial x_3\partial\beta_2^*}+v_4w_3\frac{\partial^2f_4}{\partial x_3\partial\beta_2^*}.
\end{align*}
Here, we take combination of $v_1$, $v_2$, $v_3$, $v_4$ in expression of $\bar{b}$ as $\beta_1$, $\beta_2$ are present, also take combination of $w_3, w_4$ as they related to $E$ and $I$ compartments \cite{Bifurcation of R0-1, Bifurcation of R0-2}. Now, putting the DFE value $(\mathcal{E}^0)$ we have,
\begin{align*}
	\frac{\partial f_1}{\partial x_4}&=-\beta_2^*x_1 \;\Rightarrow \frac{\partial^2 f_1}{\partial x_4\partial \beta_2^*}=-x_1,\;
	\frac{\partial f_2}{\partial x_4}=-(1-\varepsilon)\beta_2^*x_2 \;\Rightarrow \frac{\partial^2 f_2}{\partial x_4\partial \beta_2^*}=-(1-\varepsilon)x_2,\\
 \frac{\partial f_3}{\partial x_4}&=\beta_2^*x_1\; \Rightarrow \frac{\partial^2 f_3}{\partial x_4\partial \beta_2^*}=x_1,\; \frac{\partial f_4}{\partial x_4}=(1-\varepsilon)\beta_2^*x_2-(\mu+\delta+\gamma+\gamma_1)\;\Rightarrow \frac{\partial^2 f_3}{\partial x_4\partial \beta_2^*}=(1-\varepsilon)x_2.\\
\frac{\partial f_1}{\partial x_3} &=-\beta_1 x_1\;\Rightarrow\frac{\partial^2 f_1}{\partial x_3\partial \beta_2^*}=0,\;\frac{\partial f_2}{\partial x_3}=-(1-\varepsilon)\beta_1 x_2\;\Rightarrow\frac{\partial^2 f_2}{\partial x_3\partial \beta_2^*}=0,\\
\frac{\partial f_3}{\partial x_3}&=-\beta_1 x_1-(\alpha+\mu)\;\Rightarrow\frac{\partial^2 f_3}{\partial x_3\partial \beta_2^*}=0,\;
	\frac{\partial f_4}{\partial x_3}=\alpha+(1-\varepsilon)\beta_1 x_2\;\Rightarrow\frac{\partial^2 f_4}{\partial x_3\partial \beta_2^*}=0.
\end{align*}
Putting these values in the above expression of $\bar{b}$,
\begin{align*}
	\bar{b}=&-v_1x_1^*-v_2(1-\varepsilon)x_2^*
	+v_3x_1^*+(1-\varepsilon)x_2^*.
\end{align*}
As $v_4$ and $w_4$ are free variables we put $v_4=w_4=1$. Now, substituting the DFE point we have,
\begin{align*}
	 \bar{b}=\frac{-v_1\mu N-v_2(1-\varepsilon)N\phi+v_3\mu N+(1-\varepsilon) N\phi}{\mu+\phi}.
\end{align*}
Putting the expressions of $v_1$, $v_2$, and $v_3$ from \eqref{case_bif_v1}, \eqref{case_bif_v2}, and \eqref{case_bif_v3} we get explicit expression of $\bar{b}$ in terms of model parameters and it reflects that $\bar{b}>0$ automatically.\\
Here, vaccine efficiency is $0<\varepsilon<1$, and the reinfection rate from the vaccination compartment is $0<\phi_1<1$.
Now, $$\bar{a}=\sum_{k,i,j=1}^{6}v_kw_iw_j\frac{\partial^2 f_k}{\partial x_i \partial x_j}(E_{0V},\beta_2^*).$$
Here, we take $v_k=v_4$ as $v_4$ is a free variable, and $w_j=\{w_3, w_4\}$ as they correspond to the contact rates $\beta_1$ and $\beta_2$. Then,
\begin{align*}
	\bar{a}=&v_4w_1\left[w_3\frac{\partial ^2 f_4}{\partial x_1\partial x_3}+w_4\frac{\partial ^2 f_4}{\partial x_1\partial x_4}\right]+v_4w_2\left[w_3\frac{\partial ^2 f_4}{\partial x_2\partial x_3}+w_4\frac{\partial ^2 f_4}{\partial x_2\partial x_4}\right]+\\
	&v_4w_3\left[w_3\frac{\partial^2f_4}{\partial x_3\partial x_3}+w_4\frac{\partial^2f_4}{\partial x_3\partial x_4}\right]+v_4w_4\left[w_3\frac{\partial^2f_4}{\partial x_4\partial x_3}+w_4\frac{\partial^2f_4}{\partial x_4\partial x_4}\right]+\\
	&v_4w_5\left[w_3\frac{\partial^2f_4}{\partial x_5\partial x_3}+w_4\frac{\partial^2f_4}{\partial x_5\partial x_4}\right]+v_4w_6\left[w_3\frac{\partial^2f_4}{\partial x_6\partial x_3}+w_4\frac{\partial^2f_4}{\partial x_6\partial x_4} \right].
\end{align*}
As $v_4$ and $w_4$ are free variables, so putting $v_4=1$ and $w_4=1$ in the expression of $\bar{a}$ we have,
\begin{align*}
	\bar{a}
	=\;&\frac{(1-\varepsilon)\phi\beta_1(w_1w_3+w_2w_3+w_3w_3+w_3+w_5w_3+w_6w_3)}{(\mu+\phi)}+\\
	&\frac{(1-\varepsilon)\phi\beta_2^*(w_1+w_2+w_3+1+w_5+w_6)}{\mu+\phi}.
\end{align*}
Substituting the expression of $w_1,$ $w_2,$ $w_3,$ $ w_4,$ $ w_5,$ \text{and} $w_6$ from \eqref{case_w_1}, \eqref{case_w_2} and \eqref{case_w_3} we have the explicit expression of $\bar{a}$ in terms of model parameters. Hence, $\bar{a} > 0$  automatically, as all parameters are non-negative, reinfection rate is $0<\phi_1<1$, and vaccine efficiency rate is $0<\varepsilon<1$.\\
Since $\bar{a}>0$ and $\bar{b}>0$, by Theorem \ref{Center Manifold Theorem} the modified system \eqref{Bifur_New_Model} undergoes backward bifurcation at $R_0=1$. So, the condition to occur backward bifurcation in mode \eqref{Bifur_New_Model} is for $\bar{a}>0$. This follows that the bifurcation parameter $\bar{a}>0$ whenever,
\begin{align*}
	&\beta_2^*>\frac{(\varepsilon-1)\beta_1(w_1w_3+w_2w_3+w_3w_3+w_3+w_5w_3+w_6w_3)}{(1-\varepsilon)(w_1+w_2+w_3+1+w_5+w_6)}.
\end{align*}
as required. Hence, the condition or crucial parameter for bifurcation $(\beta^*)$ is obtained. Additionally, it's important to emphasize that in situations where susceptible individuals under lockdown do not become infected during the lockdown period (i.e., $\beta_2^*=0$), the bifurcation coefficient $\bar{a}$ takes on a negative value.
Thus, $\bar{a}<0$ is the case according to Theorem \ref{Center Manifold Theorem}, for which no backward bifurcation occurs, hence forward bifurcation occurs. In other terms, this research demonstrates that the occurrence of backward bifurcation in the model \eqref{Bifur_New_Model} is a result of susceptible individuals getting reinfected during the lockdown period. This result is consistent when the DFE of the model is globally asymptotically stable with $\beta_2^*=0$. Since $\bar{a}<0$ and $\bar{b}>0$ at $\beta^*=\beta_2^*$,  a transcritical bifurcation occurs according to Theorems in \cite{Bifurcation of R0-7,Bifurcation of R0-8}.\\
In a forward bifurcation plot (Figure \ref{forward-backward-bifurcation}(a)), we typically have $\mathcal{R}_0$ on the x-axis and the force of infection ($\lambda_1$) on the y-axis. From Figure \ref{forward-backward-bifurcation}(a), we see that up to $\mathcal{R}_0\in [0,1]$, the force of infection lies in zero level. The plot shows different branches or curves that represent the possible force of infection for different values of $\mathcal{R}_0$. A forward bifurcation occurs when there is a critical value of $\mathcal{R}_0$ at which a qualitative change in the behavior of the force of infection occurs. This means that as $\mathcal{R}_0$ increases past this critical value ($\mathcal{R}_0=1$), the force of infection exhibits a sudden change in its behavior, such as transitioning from a low to a high level of infection or from a stable to an unstable state \cite{Bifurcation of R0-5}.\\
In Figure \ref{forward-backward-bifurcation}(b), we plot the backward bifurcation of the force of infection concerning $\mathcal{R}_0$, the figure describes the qualitative behavior of the force of infection as the basic reproduction number ($\mathcal{R}_0$) changes. The force of infection represents the rate at which susceptible individuals become infected in a population. In a backward bifurcation plot, we typically have $\mathcal{R}_0$ on the x-axis and the force of infection ($\lambda_1$) on the y-axis. The plot shows different branches or curves that represent the possible force of infection for different values of $\mathcal{R}_0$. A backward bifurcation occurs when there is a critical value of $\mathcal{R}_0$ at which the force of infection exhibits a sudden change in its behavior. However, in this case, the change occurs in the opposite direction compared to a forward bifurcation. Instead of transitioning from a low to a high level of infection, a backward bifurcation indicates a transition from a high to a low level of infection when $\mathcal{R}_0$ exceeds 1. This phenomenon is often associated with complex transmission dynamics, such as the presence of a secondary reservoir or a significant population of infectious individuals with prolonged infectivity. A backward bifurcation implies that even if the basic reproduction number is reduced below a certain threshold, the infection can persist in the population \cite{Bifurcation of R0-4}. The plot helps us understand the conditions under which an infection can persist or resurge despite efforts to control or reduce $\mathcal{R}_0$. It highlights the potential challenges in eradicating or controlling the disease and emphasizes the importance of considering additional factors beyond $\mathcal{R}_0$ in public health interventions.\\
From Figure \ref{forward-backward-bifurcation}(c), when plotting the basic reproduction number ($\mathcal{R}_0$) for the force of infection ($\lambda_1$), the scenario of $\mathcal{R}_0$ depends on how $\mathcal{R}_0$ changes as the force of infection increases or decreases. Figure \ref{forward-backward-bifurcation}(c) shows that if the force of infection increases, it means that the rate at which susceptible individuals become infected rises. Meanwhile, $\mathcal{R}_0$ increases with increasing the force of infection, which indicates that the disease becomes more transmissible. A higher force of infection leads to a greater number of new infections caused by each infected individual. On the other hand, $\mathcal{R}_0$ decreases with the decreasing of the force of infection, which means that the disease becomes less transmissible. A lower force of infection results in a reduced number of new infections caused by each infected individual \cite{Bifurcation of R0-6}. Transmission with contact rate $\beta_{1}$ and $\beta_{2}$, and factors such as population density, contact patterns, and interventions can also influence the relationship between $\mathcal{R}_0$ and the force of infection ($\lambda_1$). 
\begin{figure}[H]
	\centering  
	\subfloat[]{\includegraphics[width=2.5 in]{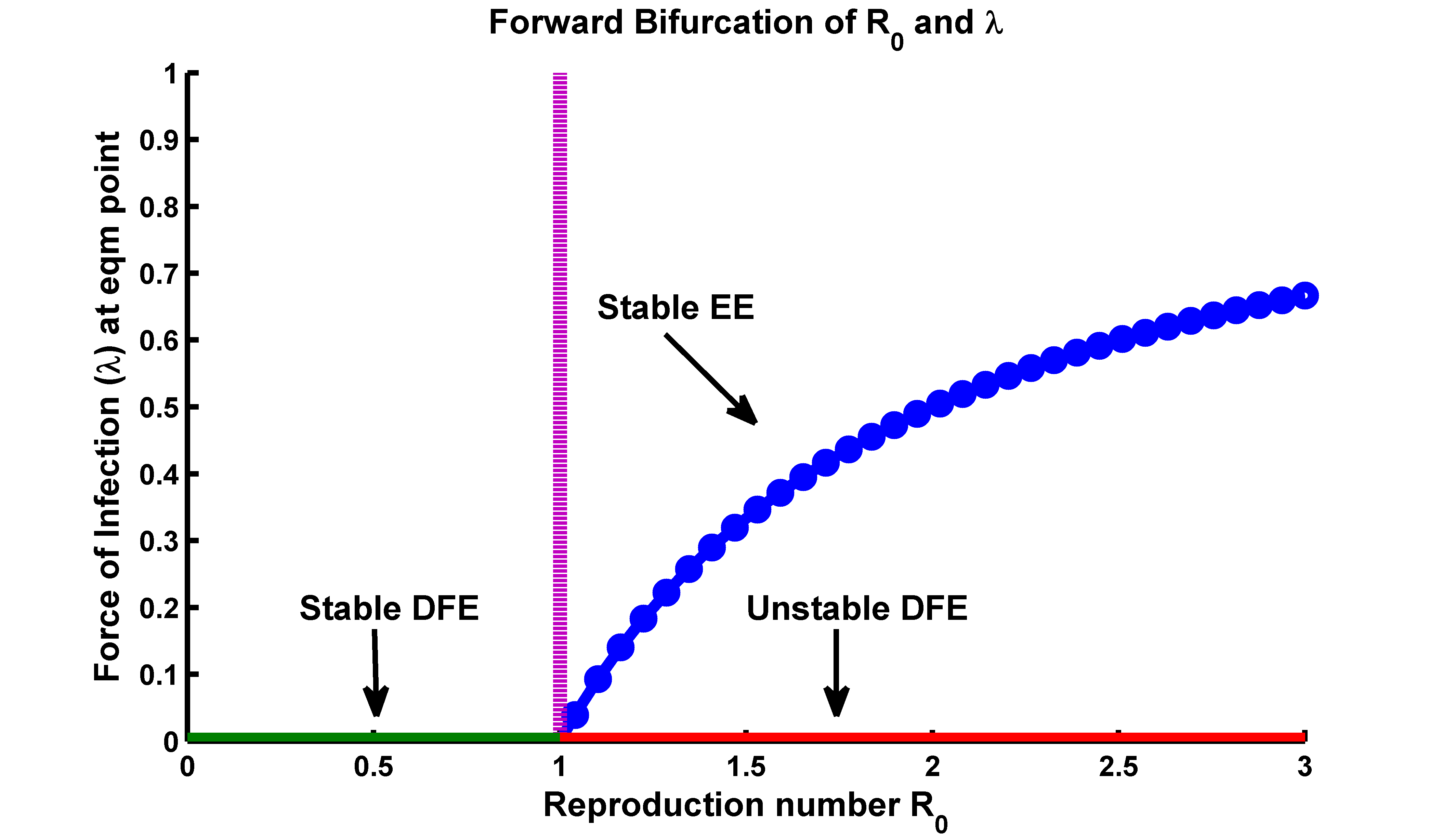}}
	\subfloat[]{\includegraphics[width=2.5 in]{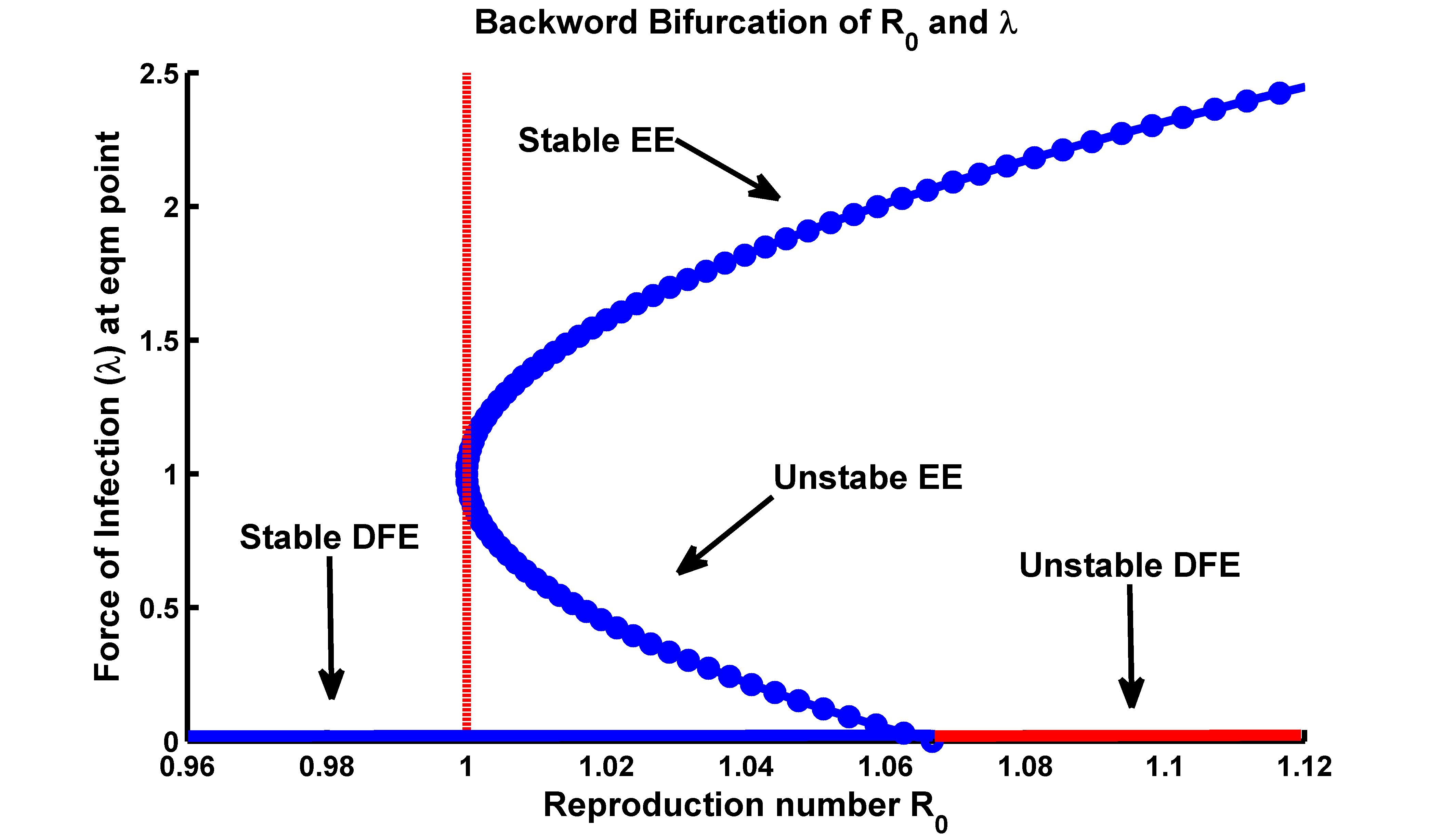}}\\
	\subfloat[]{\includegraphics[width=2.5 in]{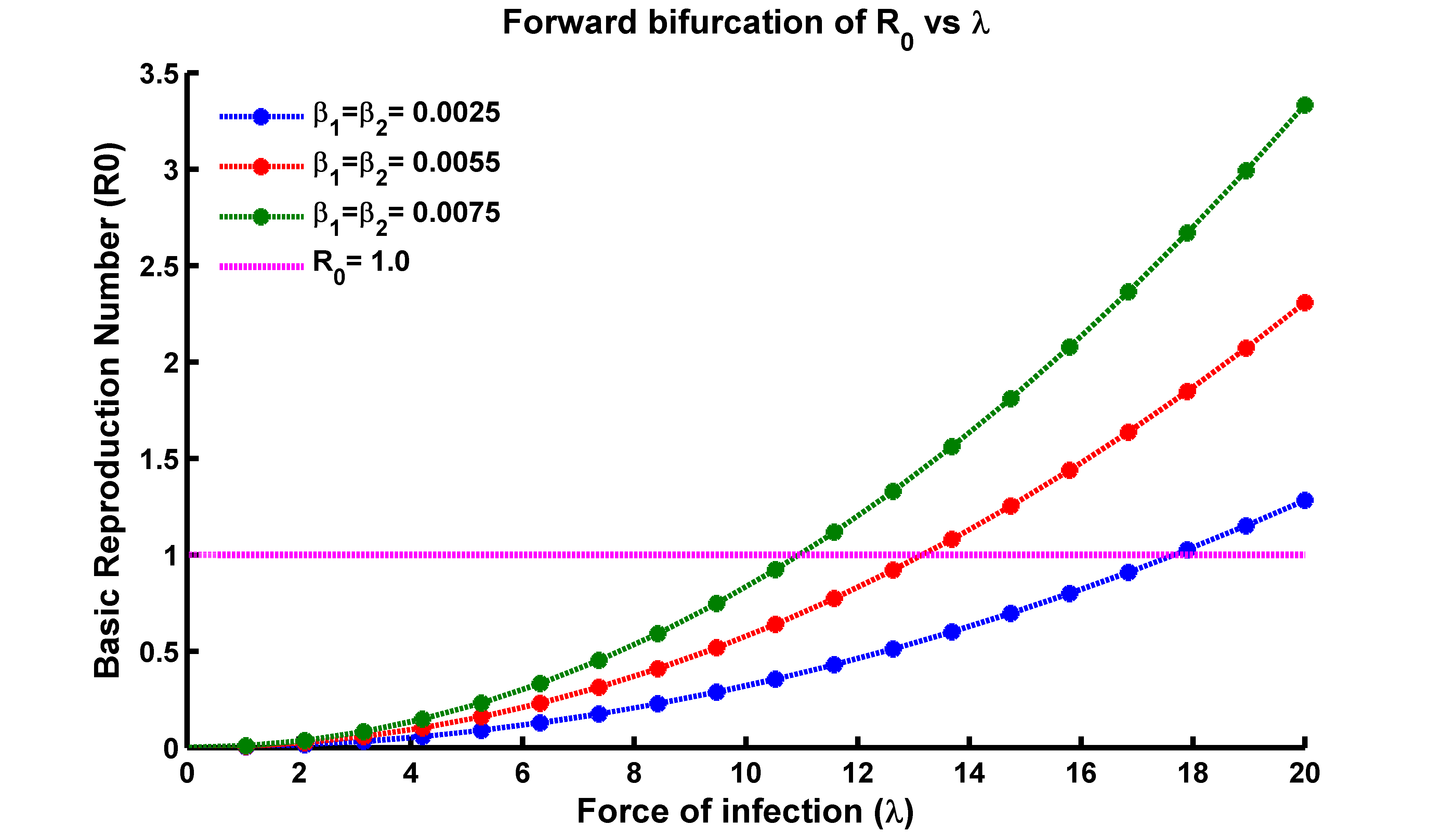}}
	\caption{{Simmulation of model \eqref{new_model} showing (a) forward bifurcation vs $\mathcal{R}_0$, (b) backward bifurcation vs $\mathcal{R}_0$, and (c) bifurcation of $\mathcal{R}_0$ vs the force of infection $\lambda$ at the equilibrium point respectively, where all the parameters are taken from Table \ref{tableparameter}.}}
	\label{forward-backward-bifurcation}
\end{figure}
\noindent
Figure \ref{forward-backward-bifurcation} illustrates that the bifurcation coefficient $\bar{b}$ is positive. Consequently, based on theorem \cite{Bifurcation of R0-4, Bifurcation of R0-5}, the model \eqref{Bifur_New_Model} represents a backward bifurcation phenomenon when the backward coefficient, $\bar{a}$, is positive. Figure \ref{forward-backward-bifurcation} depicts the corresponding forward and backward bifurcation diagram. Importantly, when setting the Influenza reinfection term $\phi_1$ to 0 as the modification parameter for enhanced susceptibility from the vaccination compartment, it was perceived that the bifurcation coefficient $\bar{a}$ is less than 0. Thus, in the Influenza co-infection scenario, when there is no reinfection after recovery from Influenza and effective measures are taken to prevent vaccinated-susceptible individuals from contracting Influenza Influenza \cite{Bifurcation of R0-4}, backward bifurcation does not occur. This outcome aligns with the earlier result. From an epidemiological viewpoint, the implication is that if this phenomenon occurs, controlling Influenza at the community level becomes more challenging, even with persistent vaccination programs and an associated reproduction number $\mathcal{R}_0<1$.

\section{Hopf-Bifurcation Analysis of the Model}\label{Section-Hopf-Bifurcation}
Hopf bifurcation occurs in an epidemic model when there is a transition from a stable equilibrium point to a stable limit cycle. In other words, the system starts to exhibit periodic oscillations instead of converging to a steady state \cite{Hopf LEAST LHS PRCC-1}. This is indicated in the SVEIRT model when the basic reproduction number $\mathcal{R}_0$ is greater than $1$ and there is a change in stability of the DFE equilibrium point. This can happen when the reproduction number, which represents the average number of new infections caused by a single infected individual, crosses a certain threshold value. The occurrence of Hopf bifurcation in an epidemic model is significant because it can lead to the emergence of sustained, periodic outbreaks of the disease \cite{Hopf LEAST LHS PRCC-2}. This can happen even if the disease would have otherwise died out in the absence of intervention or natural immunity. Sometimes this occurs when the model includes time delays and the steady-state equilibrium becomes unstable, leading to the emergence of limit cycles or oscillatory behavior in the system. The periodicity of the outbreaks can also make it more difficult to control the disease using traditional methods such as vaccination or quarantine. Specifically, the Hopf bifurcation occurs when a pair of complex conjugate eigenvalues of the  Jacobian matrix cross the imaginary axis as the value of the threshold quantity $\mathcal{R}_0$ increases. This leads to the emergence of limit cycles or periodic oscillations in the model. Hopf bifurcation is an important phenomenon in the study of infectious disease dynamics as it can lead to complex and unpredictable patterns in disease spread, which can have significant public health implications. Therefore, understanding the conditions under which Hopf bifurcation occurs in epidemic models can provide insights into the dynamics of infectious diseases and help inform public health policies to prevent and control outbreaks \cite{Hopf LEAST LHS PRCC-4}.
To determine the conditions for a Hopf bifurcation in the SVEIRT model, one can use the Routh-Hurwitz criterion, which involves computing the coefficients of the characteristic polynomial of the linearized system evaluated at the disease-free equilibrium. We need to compute the Jacobian matrix evaluated at each equilibrium point. By setting the determinant of the Jacobian matrix evaluated at this point to zero, $$|J(S^*, V^*, E^*, I^*, R^*, T^*) - \lambda I|= 0.$$
The characteristic polynomial can be written as
\begin{align*}
	p(\lambda) = \lambda^6 + a_1\lambda^5 + a_2\lambda^4 + a_3\lambda^3 + a_4\lambda^2 + a_5\lambda + a_6.
\end{align*}
Where the coefficients $a_1, a_2, a_3, a_4, a_5$, and $a_6$ depend on the model parameters and $\lambda$ is a complex eigenvalue with a positive real part that determines the stability of the limit cycle that arises from the bifurcation. To find the condition for a Hopf bifurcation, we need to calculate the sign of the coefficient of the linear term in the normal form of the system near the equilibrium point. The normal form is given by,
$$\dot{z} = (\alpha + i\omega)z - \mu|z|^2z.$$
Where $z$ is a complex variable representing the deviation from the equilibrium point, $\alpha$ and $\omega$ are real constants, and $\mu$ is a small parameter. To obtain the normal form, we need to examine the eigenvalues of the Jacobian matrix evaluated at the equilibrium point. If the real part of one of the eigenvalues changes sign as a parameter is varied, a Hopf bifurcation occurs. The characteristic polynomial gives us the eigenvalues of the Jacobian matrix. Therefore, we can use the coefficients of the characteristic polynomial to find the normal form coefficients. The coefficients are related to the normal form coefficients as follows,
\begin{align*}
	\alpha =& \frac{1}{2}(a_5 - a_1),\;\;\;\omega = \frac{1}{2}(a_4 - a_2), \; \text{and}\\
	\mu =& \frac{1}{4}(a_1a_5 - a_2a_4 + a_3a_3) - \frac{1}{2}(a_1a_4 + a_2a_5) + a_3a_6.
\end{align*}
The condition for a Hopf bifurcation is that $\alpha = 0$ and $\omega \neq 0$. Therefore, we need to set $a_5 - a_1 = 0$ and $a_4 - a_2 \neq 0$. This gives us the conditions,
$$a_5 = a_1,\;\; a_4 \neq a_2.$$
If the Routh-Hurwitz criterion yields that all the coefficients of the polynomial are positive, then the disease-free equilibrium is stable, and there is no Hopf bifurcation. However, if one or more of the coefficients are negative, then the disease-free equilibrium is unstable, and a Hopf bifurcation can occur \cite{Hopf LEAST LHS PRCC-1, Hopf LEAST LHS PRCC-2, Hopf LEAST LHS PRCC-4}. In summary, the condition for a Hopf bifurcation is that the characteristic polynomial has a repeated eigenvalue, i.e., $a_5 = a_1$, and the coefficient of the quartic term is different from the coefficient of the quadratic term, i.e., $a_4 \neq a_2$.\\
In epidemiology, when a Hopf bifurcation is represented in the susceptible population in Figure \ref{Hopf-Bifurcation-analysis}(a), with oscillation after 20 weeks, the Figure Figure \ref{Hopf-Bifurcation-analysis}(a) describes a complex dynamic pattern. It indicates that the number of susceptible individuals $(S(t))$ in the population exhibits periodic fluctuations over time. This suggests that the transmission dynamics of the infectious disease may undergo a qualitative change, leading to sustained oscillations in the susceptibility of the population.\\
When this bifurcation occurs in the vaccinated population represented in (Figure \ref{Hopf-Bifurcation-analysis}(b)), it suggests that the dynamics of the disease transmission are changing due to the vaccination efforts. Meanwhile, Figure \ref{Hopf-Bifurcation-analysis}(b) represents the temporal pattern of the vaccinated population's susceptibility to the disease. Initially, after the introduction of vaccination, the vaccinated population increases, and disease transmission decreases. However, after a certain period (20 weeks in this case), the vaccinated population's susceptibility starts oscillating, leading to periodic fluctuations in disease transmission. These oscillations can occur due to various factors, such as waning immunity, changes in contact patterns, or the emergence of new variants of the pathogen. The specific characteristics of the oscillations (e.g., amplitude, frequency) can provide insights into the dynamics of the disease and the effectiveness of vaccination strategies.\\
The specific interpretation of Figure \ref{Hopf-Bifurcation-analysis}(c) describing the oscillation after 20 weeks in the exposed population $E(t)$. In the case of a Hopf bifurcation, the oscillations in the exposed population could suggest periodic fluctuations in the number of individuals transitioning from the susceptible to the exposed state and vice versa.
\begin{figure}[H]
	\centering  
	\subfloat[]{\includegraphics[width=2.5 in]{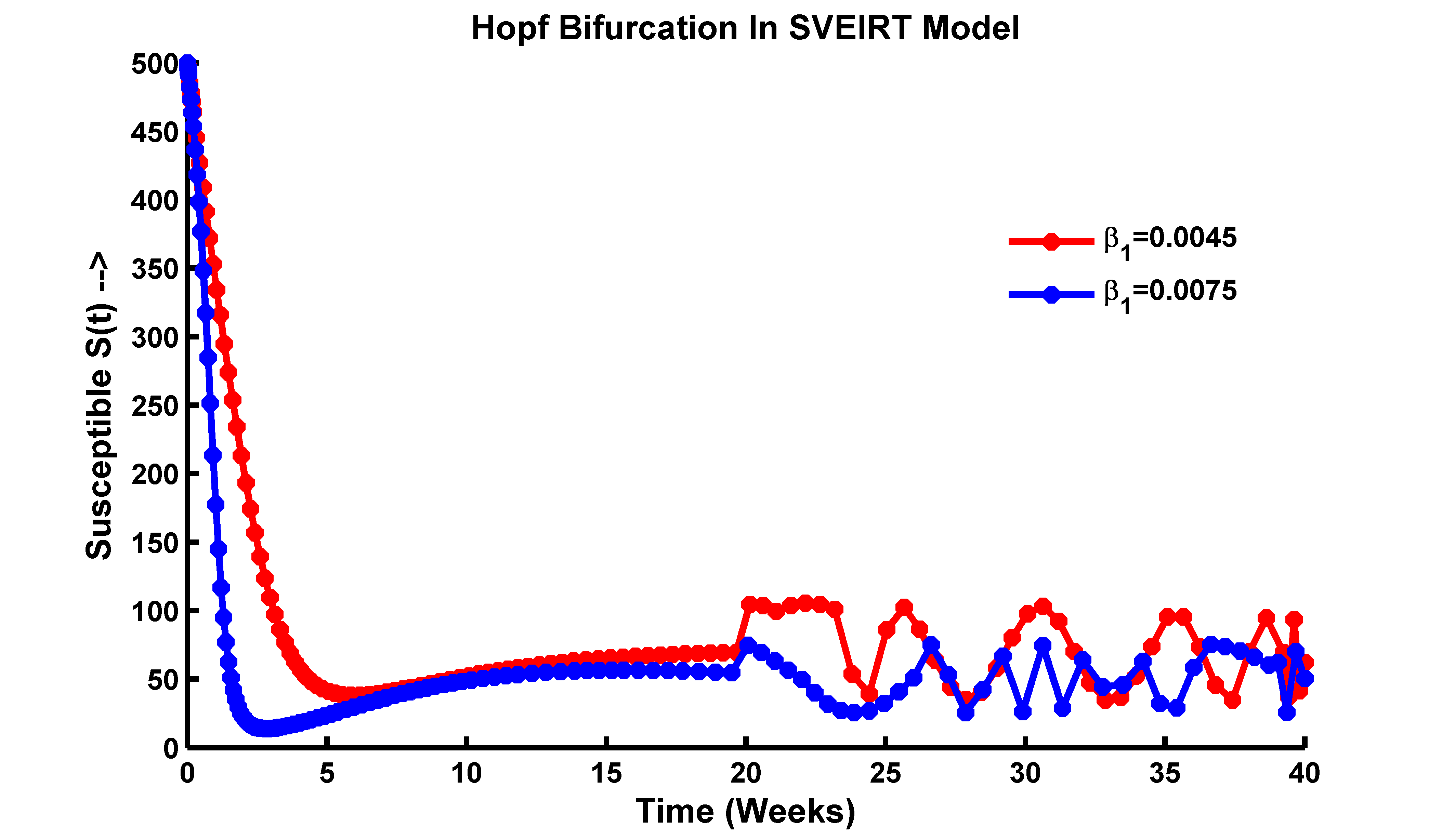}}
	\subfloat[]{\includegraphics[width=2.5 in]{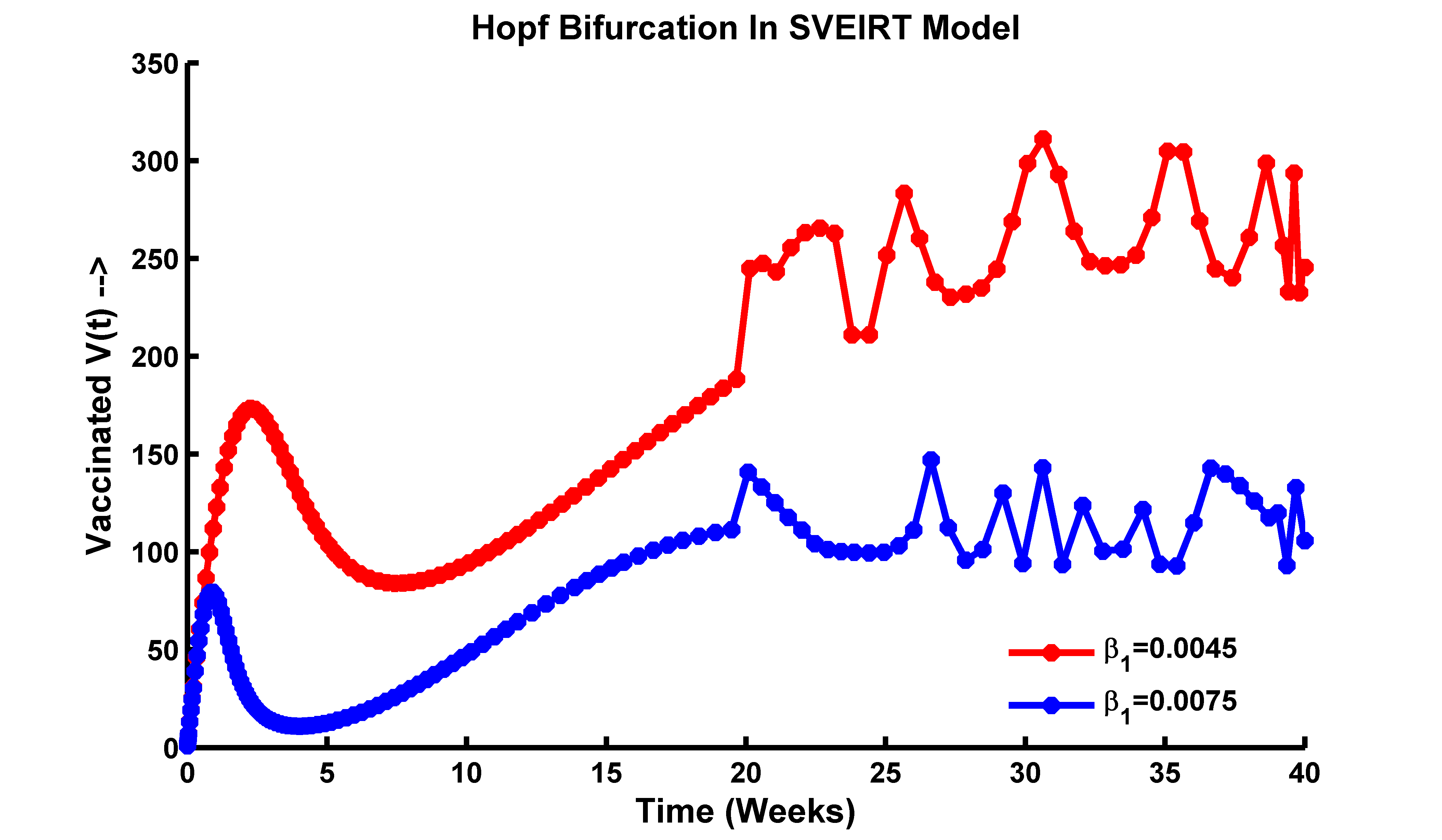}}\\
	\subfloat[]{\includegraphics[width=2.5 in]{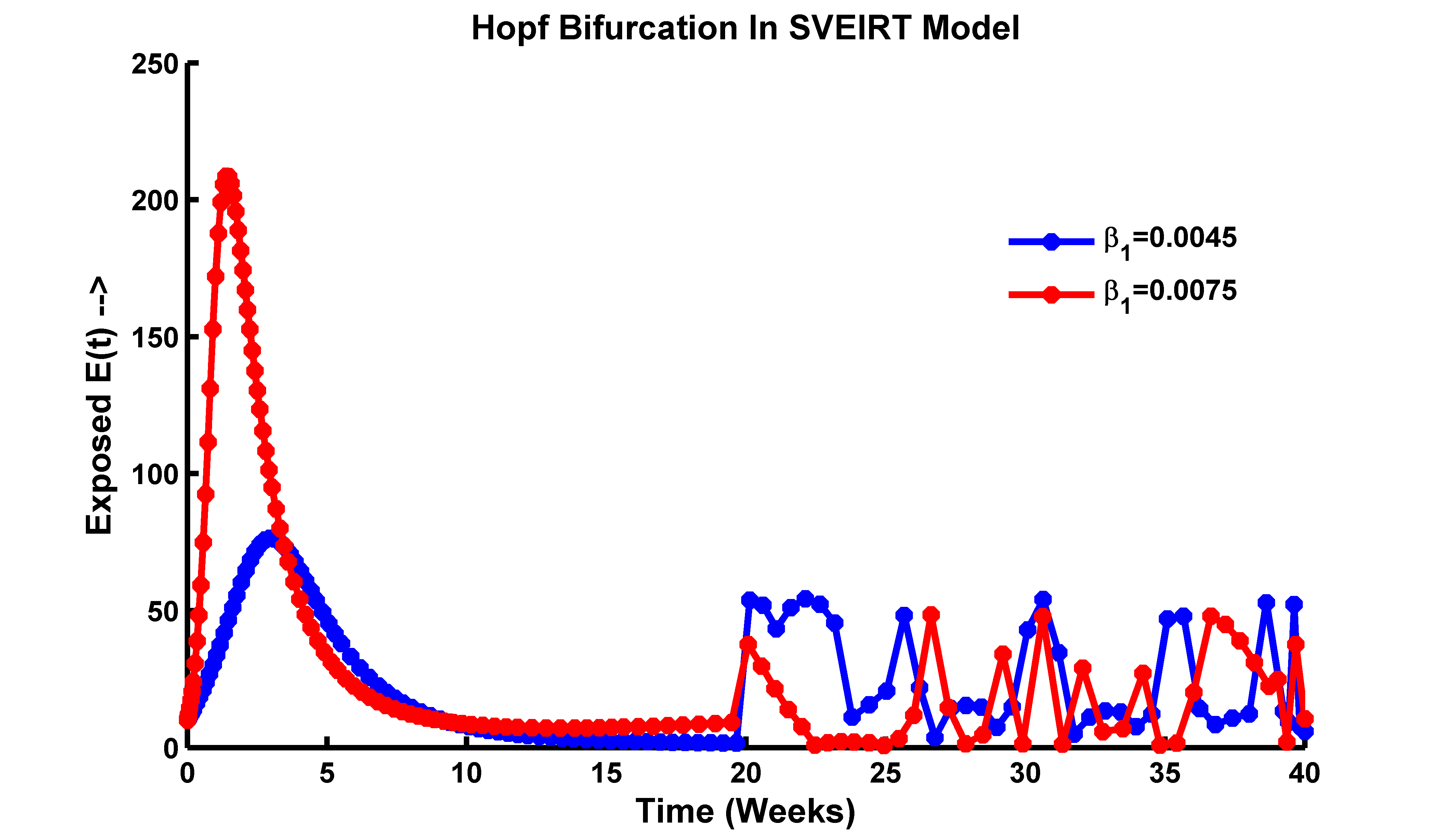}}
	\subfloat[]{\includegraphics[width=2.5 in]{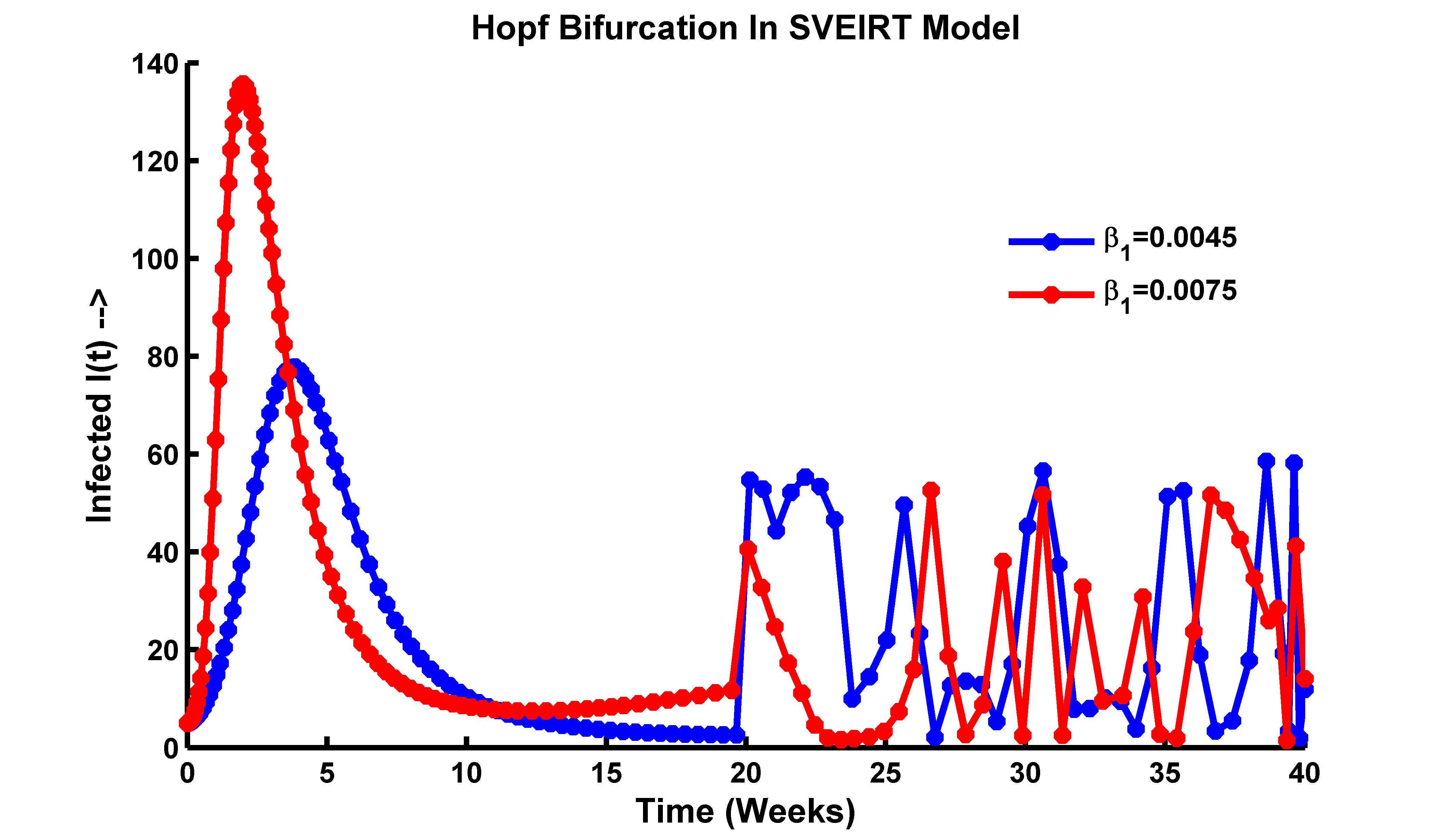}}\\
	\subfloat[]{\includegraphics[width=2.5 in]{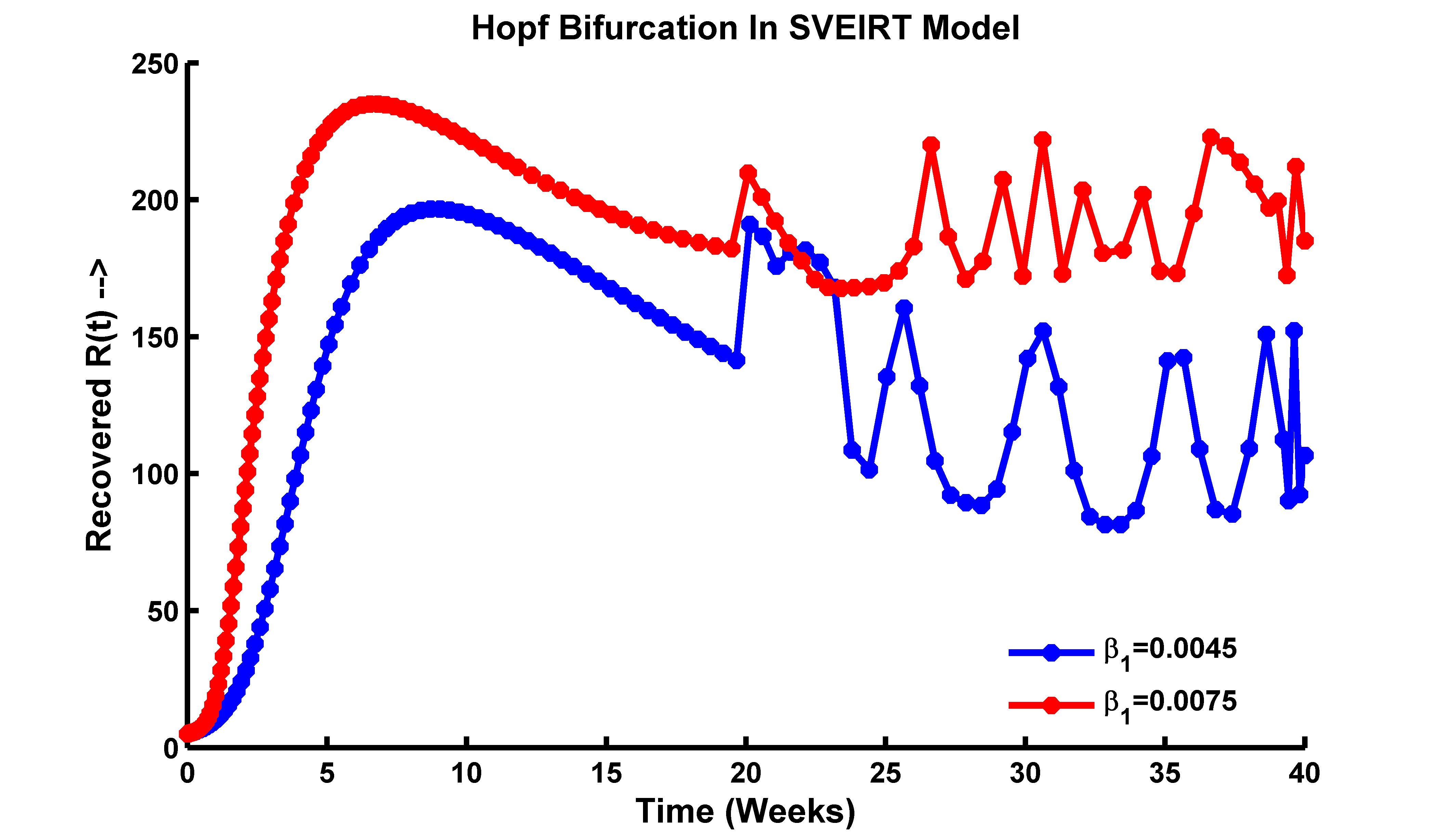}}
	\subfloat[]{\includegraphics[width=2.5 in]{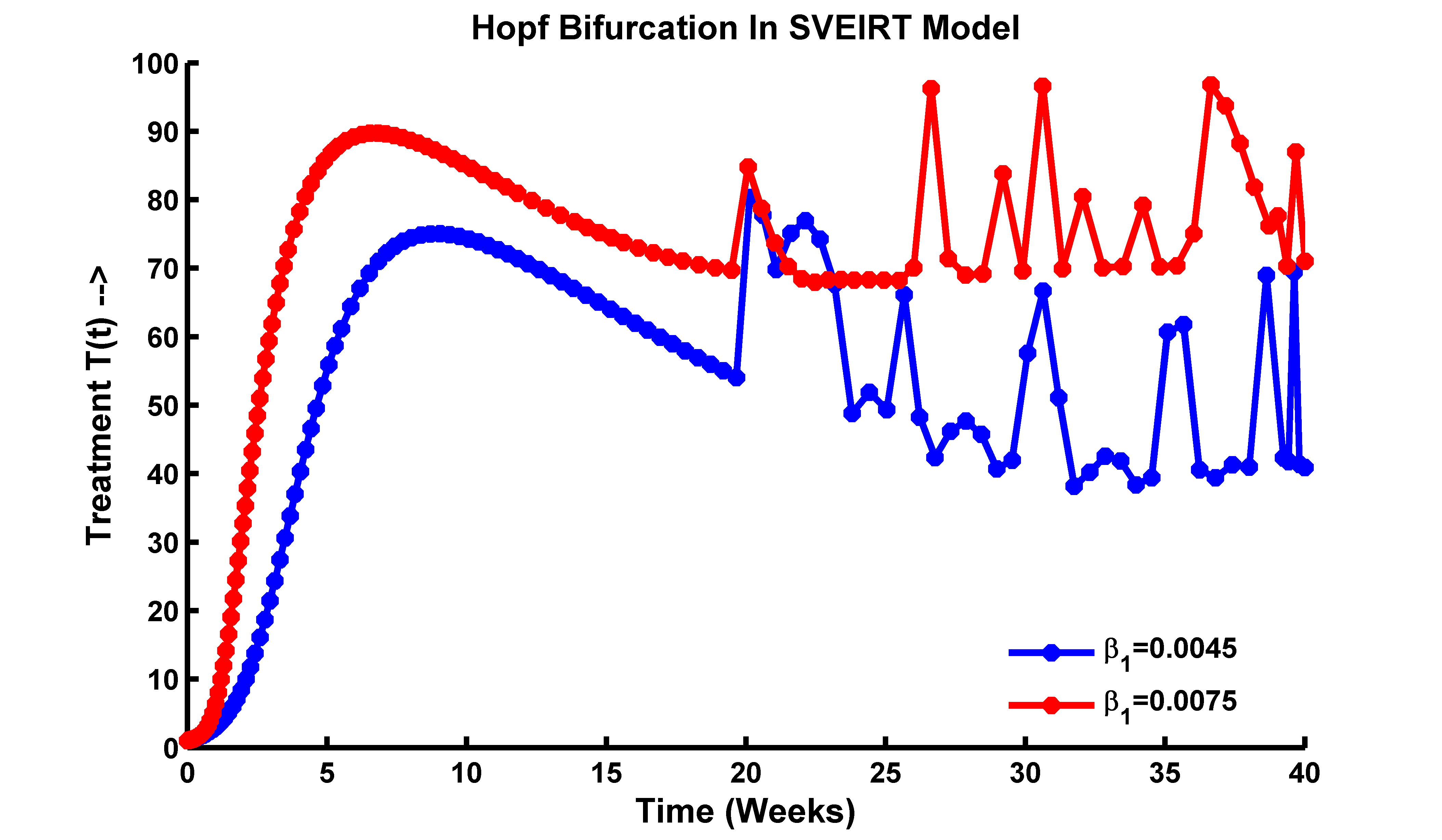}}
	\caption{Hopf bifurcation occurs at (a) 
 $S(t)$ compartment, (b) $V(t)$ compartment, (c) $E(t)$ compartment, (d) $I(t)$ compartment, (e) $R(t)$ compartment, and (f) $T(t)$ compartment when the endemic equilibrium point becomes unstable, using $\beta_1 = 0.0065,\;\beta_2 = 0.0065,$ \;$\mu=2$, and rest of the parameter values are taken from Table \ref{tableparameter}.}\label{Hopf-Bifurcation-analysis}
\end{figure}
\noindent
The precise implications of such oscillations would depend on various factors, including the specific disease being modeled, the population dynamics, and the model parameters. It reflects seasonal patterns, cyclic changes in human behavior or interventions, or other factors that influence the transmission dynamics of the disease.\\
In Figure \ref{Hopf-Bifurcation-analysis}(d), Hopf bifurcation is represented in the infected population in epidemiology, it implies that the number of infected individuals exhibits sustained oscillations after a certain period, such as 20 weeks. This means that the population of infected individuals cyclically fluctuates over time. Some possible explanations for the observed oscillations could include:\\
The oscillations may be driven by seasonal factors that influence the transmission of the disease. For example, certain diseases may exhibit increased transmission during specific seasons, leading to periodic spikes in the number of infections. Oscillations could be the result of human behavior or interventions that vary periodically. For instance, changes in contact patterns, adherence to preventive measures, or the implementation of control measures might vary over time, leading to fluctuations in the number of infections. The dynamics of the Influenza disease and its interaction with the population could create self-regulating feedback mechanisms that give rise to oscillatory behavior. These mechanisms might involve factors such as immunity, host susceptibility, or pathogen evolution.\\
In Figures \ref{Hopf-Bifurcation-analysis}(e) and (f), the representation of a Hopf bifurcation in the recovered and treatment population with oscillation after 20 weeks suggests a specific dynamic behavior in the disease system. A Hopf bifurcation occurs when a system undergoes a qualitative change in its behavior as a parameter (in this case, possibly an infection rate or treatment effectiveness) crosses a critical threshold. In this scenario, Figure \ref{Hopf-Bifurcation-analysis}(e) and (f), describes the fluctuation or oscillation observed in the population of individuals who have recovered from the disease and those undergoing treatment over 20 weeks. The oscillation indicates that the population sizes of these two groups are changing periodically, possibly with a recurring pattern. Some observations are: the Hopf bifurcation and subsequent oscillation suggest that the disease dynamics have transitioned from a stable state to cyclical or periodic behavior. The patterns indicate that factors such as seasonal variations, behavioral changes, or the dynamics of immunity and treatment influence the population sizes of the recovered and treatment groups. Understanding these oscillations and their underlying causes is crucial for developing effective interventions and strategies to control and manage the disease.\\
Once a Hopf bifurcation occurs, the disease-free equilibrium becomes unstable, and a stable limit cycle emerges. The limit cycle represents the oscillatory behavior of the disease dynamics, where the number of individuals in each compartment (susceptible, vaccinated, exposed, infected, recovered, and treated) varies periodically over time. The exact conditions for the emergence of the limit cycle depend on the specific values of the model parameters and cannot be determined analytically. However, numerical simulations can be used to explore the behavior of the system near the Hopf bifurcation and to estimate the parameters that lead to the emergence of the limit cycle.

\section{Local and Global Stabilities of DFE and EE}\label{Section-Local-global-stability of DFE EE}
We know that the stability characteristics of linear ODEs only depend on the system's eigenvalues. Since our suggested model \eqref{new_model} is non-linear, we must use Hartman and Grobman's theorem, and linearization to combine the local behavior of linear and non-linear systems. We now analyze the local stability of the equilibria at points $\mathcal{E}^0$ and $\mathcal{E}^*$ by approximating the non-linear system of differential equations with a linear system. The system is then locally perturbed from equilibrium, and the long-term behavior that results is then examined. This is performed by linearizing the system about each equilibrium, using the Jacobian approach for \eqref{new_model}. Analyzing the linearized system, $$\dot{z}=J(\mathbb{E})z(t).$$
Here, $\mathbb{E}$ represents the equilibrium space. We can look into the stability of each equilibrium point $\mathbb{E}=\mathcal{E}^0$ and $\mathbb{E}=\mathcal{E}^*$. We will obtain that the property depends on a crucial factor, referred to as basic reproduction number $\mathcal{R}_0$ which is estimated above. As a result, the nature of $\mathcal{R}_0$ can be examined to determine whether persistence or extinction of disease occurs as $t\rightarrow\infty$ \cite{Stability Bound-18, Stability Bound-16}.\\
In this section, we shall proceed to analyze the stability properties of the DFE and EE. Firstly, we analyze the following results regarding to local and global stability of DFE. 
\subsection{Local Stability of Disease-Free Equilibrium State $(\mathcal{E}^0)$}
\begin{Th}\label{lsDFE}
	The disease-free equilibrium point is locally asymptotically stable if $\mathcal{R}_0<1$ and unstable if $\mathcal{R}_0>1$.
\end{Th}
\begin{proof}
	The next step is to linearize the system and use the Routh-Hurwitz criterion to identify the circumstances in which the linear system has only negative eigenvalues. Because a point is deemed an \enquote{attractor} if its Jacobian matrix's eigenvalues at that location have negative real portions, even slight disturbances from the equilibrium induce the system to gradually return there. Alternatively, if any eigenvalues represent positive real parts, slight deviations from equilibrium lead to amplification, causing the system to diverge. This results in a ``repeller" point, where the local behavior of the linearized system aligns with the non-linear system, following the Hartman-Grobman theorem. For the system, the Jacobian at the disease-free equilibrium point $\mathcal{E}^0=(S_0, V_0, E_0, I_0, R_0, T_0)$ results in,
	\begin{equation}
		J(\mathcal{E}^0)=
		\begin{pmatrix}
			-(\mu+\phi)& 0 & -\beta_1S_0 & -\beta_2S_0 & 0 & 0 \\
			\phi & -\mu & -\lambda\beta_1 V_0 & -\lambda\beta_2 V_0 & 0 & 0 \\
			0 & 0 & a_{33} & \beta_2S_0 & 0 & 0 \\
			0 & 0 & \alpha+\lambda\beta_1V_0 & a_{44} & 0 & 0\\
			0 & 0 & 0 & \gamma & -\mu & 0 \\
			0 & 0 & 0 & \gamma_1 & 0 & -\mu 
		\end{pmatrix}
	\end{equation}
	Here, we let, $\lambda=(1-\varepsilon),\; a_{33}=\beta_1S_0-(\alpha+\mu)$, and $a_{44}=\lambda\beta_2V_0-(\mu+\delta+\gamma+\gamma_1).$\\
	At the DFE point we have substituted  $\displaystyle S_0=\frac{\mu N}{\mu+\phi}$, $\displaystyle V_0=\frac{\phi N}{\mu+\phi}$, and rest of all state variables $0$.
	The characteristic polynomial of $J(\mathcal{E}^0)$ is got as,
	$$|J(\mathcal{E}^0)-x\mathbb{I}|=0.$$
	Expanding the terms as well as ordering by powers of $x$ we obtain the required characteristic polynomial. This equation ultimately simplifies to,
	\begin{align*}
		&(x+\mu)^3(x^2+(\alpha+\mu)(\gamma+\gamma_1+\delta+\mu-\beta_2V_0\lambda)+x(\alpha-\beta_1S_0+\gamma+\gamma_1+\delta+2\mu-V_0\beta_2\lambda)-\\
		&S_0(\alpha\beta_2+\beta_1(\mu+\delta+\gamma+\gamma_1)))(x+\mu+\phi)=0.\\
  \Rightarrow&(x+\mu)^3\{x^3+x(\alpha+\mu)(\gamma+\gamma_1+\delta+\mu-\beta_2V_0\lambda)+x^2(\alpha+2\mu+\delta+\gamma+\gamma_1-S_0\beta_1-V_0\beta_2\lambda)\\
		&-S_0x(\alpha\beta_2+\beta_1(\mu+\delta+\gamma+\gamma_1))+x^2(\mu+\phi)+(\alpha+\mu)(\mu+\phi)(\mu+\gamma+\gamma_1+\delta-\beta_2V_0\lambda)\\
		&+x(\mu+\phi)(2\mu+\alpha+\delta+\gamma+\gamma_1-\beta_1S_0-V_0\beta_2\lambda)-S_0(\mu+\phi)(\alpha\beta_2+\beta_1(\mu+\delta+\gamma+\gamma_1))\}=0\\
		\Rightarrow& (x+\mu^3)(x^3+A_1x^2+A_2x+A_3)=0.
	\end{align*}
	Where 
	\begin{align*}
		A_1=&(\alpha+2\mu+\delta+\gamma+\gamma_1-S_0\beta_1-V_0\beta_2\lambda)+(\mu+\phi).\\
		A_2=&(\alpha+\mu)(\gamma+\gamma_1+\mu+\delta-\beta_2V_0\lambda)-S_0(\alpha\beta_2+\beta_1(\mu+\delta+\gamma+\gamma_1))+\\
		&(\mu+\phi)(2\mu+\alpha+\delta+\gamma+\gamma_1-\beta_1S_0-V_0\beta_2\lambda).\\
		A_3=&(\alpha+\mu)(\mu+\phi)(\mu+\gamma+\gamma_1+\delta-\beta_2V_0\lambda)-S_0(\mu+\phi)(\alpha\beta_2+\beta_1(\mu+\delta+\gamma+\gamma_1)).
	\end{align*}
	According to the Routh-Hurwitz criterion \cite{  Stability Bound-11,Stability Bound-16}, all roots of the cubic equation (second factor of the polynomial expression), possess negative real part if and only if $A_1, A_2, A_3 > 0$ and $A_1A_2-A_3>0$. It reflects $A_1>0$, and after expressing $A_2$ in terms of $\mathcal{R}_0$ we have obtained,
	\begin{align*}
		A_2=&(\alpha+\mu)(\mu+\phi)(\mu+\delta+\gamma+\gamma_1)(1-\mathcal{R}_0)+(\mu+\phi)(\alpha+\mu)(\mu+\delta+\gamma+\gamma_1)(1-\mathcal{R}_0)+\\
		&2\mu(\mu+\phi).
	\end{align*}
	Thus, for $A_2>0$, it is necessary that $\mathcal{R}_0<1$. Similarly, we write $A_1$ as in terms of $\mathcal{R}_0$ as follows, 
	\begin{align*}
		A_1=(\alpha+\mu)+(\mu+\phi)+(\alpha+\mu)(\mu+\phi)(\mu+\delta+\gamma+\gamma_1)(1-\mathcal{R}_0),
	\end{align*}
	and using the previous condition if $\mathcal{R}_0<1$ we find that $A_1>0$. Moreover, we write $A_3$ as in terms of $\mathcal{R}_0$,
	\begin{align*}
		A_3=(\alpha+\mu)(\mu+\phi)(\mu+\delta+\gamma+\gamma_1)(1-\mathcal{R}_0).
	\end{align*} 
	Hence, we find that $A_3>0$ if $\mathcal{R}_0<1.$
	Clearly, for $\mathcal{R}_0<1$ we obtain that $A_1A_2-A_3>0$, and the Routh-Hurwitz criteria are satisfied.
	The characteristic roots (eigenvalues) of the Jacobian matrix at $J(\mathcal{E}^0)$ are
	$-\mu,-\mu,-\mu$,
	and the rest of them are determined by the nature of the coefficients of the cubic polynomials described above.
	Thus, $\mathcal{R}_0<1$ implies that all the eigenvalues $\lambda_i\;(i=1,2,\cdots,6)$ of the linearized system are negative. Conversely, for the case $\mathcal{R}_0>1$, then there is at least one positive eigenvalue in the linearized system, and the equilibrium becomes unstable \cite{Stability Bound-14}. Hence, the disease free equilibrium state ($\mathcal{E}^0$) is locally asymptotically stable if $\mathcal{R}_0<1$ and unstable if $\mathcal{R}_0>1$.
\end{proof}
The above analysis reveals that when $\mathcal{R}_0=1$ then the above analysis fails. The scenario $\mathcal{R}_0=1$ is equivalent to, $$\beta_2=\beta_2^{[TC]}=\frac{(\mu+\phi)(\alpha+\mu)(\mu+\delta+\gamma+\gamma_1)-\mu N\beta_1(\mu+\delta+\gamma+\gamma_1)}{\phi N\lambda(\alpha+\mu)+\mu N\alpha}.$$
In the previous section, we have shown that model system \eqref{new_model} passes through transcritical bifurcation at disease-free state $(\mathcal{E}^0)$ when model parameter $\beta_2$ undergoes its critical value $\beta_2=\beta_2^{[TC]}$.

\subsection{Global Stability of Disease-Free Equilibrium State $\mathbf{(\mathcal{E}^0)}$}
\begin{Th}\label{gsDFE_1}
\cite{Stability Bound-8, Stability Bound-21, Stability Bound-16}.	When $\mathcal{R}_0<1$, the DFE $\mathcal{E}^0$ is globally asymptotically stable. When $\mathcal{R}_0>1$, $\mathcal{E}^0$ is unstable.
\end{Th}
\noindent To prove this theorem we use the following lemma.
\begin{Lemma}\label{gsl_DFE}
	\cite{Stability Bound-8}. (Global stability of DFE) Consider the model written in the form,
	\begin{align*}
		\frac{dX_1}{dt} &=F(X_1,X_2),\\
		\frac{dX_2}{dt} &=G(X_1,X_2),\;\;\;\; G(X_1,0)=0.
	\end{align*}
	Let $X_1 \in \mathbb{R}^m$ represent the number of uninfected individuals. Similarly, $X_2 \in \mathbb{R}^m$ signifies the count of infected people, encompassing latent, infectious, and other categories. $X_0=(X_1^*)$ is the DFE for the system \eqref{new_model}. Furthermore, the following conditions (H1) and (H2) are assumed:
	\begin{itemize}
		\item[(H1)] $ X_1^*$ is globally asymptotically stable when $\displaystyle \frac{dX_1}{dt}=F(X_1,0)$.
		\item[(H2)] For $\displaystyle G(X_1,X_2)=AX_2-\hat{G}(X_1,X_2) ,\; \hat{G}(X_1,X_2)\geq 0$ for $(X_1, X_2) \in \Omega$. Where the Jacobian $\displaystyle A=\frac{\partial G}{\partial X_2}(X_1^*,0)$ is an M-matrix (the off-diagonal elements of A are non-negative), and $\Omega$ is the region where the model makes biological sense. Then the DFE $X_0=(X_1^*,0)$ is globally asymptotically stable provided that $\mathcal{R}_0 <1.$
	\end{itemize}
\end{Lemma}
\noindent Next, we will proceed the proof for Theorem \ref{gsDFE_1}, and \ref{gs_DFE2} (Appendix \ref{allproofs}).\\

\begin{Th}\label{gs_DFE2}\cite{Stability Bound-17}.
	The disease-free steady state $(\mathcal{E}^0)$ of the model \eqref{new_model} is globally asymptotically stable when $\mathcal{R}_0 <1$ and the disease dies out.
\end{Th}

\subsection{Local Stability of Endemic Equilibrium State $\mathbf{(\mathcal{E}^*)}$}
\begin{Th}\label{lsEE}
	The endemic equilibrium point $\mathcal{E}^*(S^*,V^*,E^*,I^*,R^*,T^*)$ is locally asymptotically stable if $\mathcal{R}_0<1$ and unstable if $\mathcal{R}_0>1$.
\end{Th}
\begin{proof}
	The analysis for $\mathcal{E}^*$ is similar to that of $\mathcal{E}^0$. The Jacobian of the system at the endemic equilibrium point (Linearizing \eqref{new_model} about $\mathcal{E}^*$) $\mathcal{E}^*(S^*, V^*, E^*, I^*, R^*, T^*)$ is as follows,
	\begin{equation}
		J(\mathcal{E}^*)=
		\begin{pmatrix}
			a_{11} & 0 & -\beta_1S^* & -\beta_2S^* & 0 & 0 \\
			\phi & a_{22} & -\lambda\beta_1 V^* & -\lambda\beta_2 V^* & 0 & 0 \\
			(\beta_1 E^*+\beta_2 I^*) & 0 & a_{33} & \beta_2S^* & 0 & 0 \\
			0 &a_{42} & \alpha+\lambda\beta_1V^* & a_{44} & 0 & 0\\
			0 & 0 & 0 & \gamma & -\mu & 0 \\
			0 & 0 & 0 & \gamma_1 & 0 & -\mu 
		\end{pmatrix}
	\end{equation}
	where $a_{11}= -(\beta_1E^*+\beta_2I^*)-(\mu+\phi),\; a_{22}=-\lambda(\beta_1E^*+\beta_2I^*)-\mu$,
	$a_{33}=\beta_1S^*-(\alpha+\mu),\;a_{42}=\lambda(\beta_1E^*+\beta_2I^*)$, and $a_{44}=\lambda\beta_2V^*-(\mu+\delta+\gamma+\gamma_1).$\\
	At the endemic equilibrium state $\mathcal{E}^*$, by calculating the Jacobian matrix $J$, after that by solving $|(J-x\mathbb{I})|=0$, the two eigenvalues are $-\mu,\;-\mu$ which are negative. Rest are the roots of the characteristic polynomial of the matrix,
	\begin{equation}
		J_1(\mathcal{E}^*)=
		\begin{pmatrix}
			a_{11} & 0 & -\beta_1S^* & -\beta_2S^*\\
			\phi & a_{22} & -\lambda\beta_1 V^* & -\lambda\beta_2 V^*\\
			(\beta_1 E^*+\beta_2 I^*) & 0 & \beta_1S^*-(\alpha+\mu) & \beta_2S^* \\
			0 & \lambda(\beta_1 E^*+\beta_2 I^*) & \alpha+\lambda\beta_1V^* & a_{44}
		\end{pmatrix}
	\end{equation}
	where $a_{11}= -(\beta_1E^*+\beta_2I^*)-(\mu+\phi),\; a_{22}=-\lambda(\beta_1E^*+\beta_2I^*)-\mu$,
	and $a_{44}=\lambda\beta_2V^*-(\mu+\delta+\gamma+\gamma_1).$\\
	 The characteristic equation is obtained as follows,
$$x^4+A_1x^3+A_2x^2+A_3x+A_4=0.$$
	Expanding the expressions and arranging them in terms of powers of $x$, these coefficients of the equation ultimately simplify to,\\
	\begin{align*}
		A_1=&\alpha-S^*\beta_1+I^*\beta_2+\delta+\gamma+\gamma_1+I^*\beta_2\lambda-V^*\beta_2\lambda+E^*\beta_1(1+\lambda)+4(\mu+\phi).\\
  A_2=&-S^*\alpha\beta_2+\alpha\gamma-S^*\beta_1\gamma+\alpha\gamma_1-S^*\beta_1\gamma_1+\alpha\delta-S^*\beta_1\delta+E^{*2}\beta_1^2\lambda-V^*\alpha\beta_2\lambda+I^{*2}\beta_2^2\lambda+\\
		&3\alpha\mu-3S^*\beta_1\mu+3\gamma\mu+3\gamma_1\mu+3\delta\mu-3V^*\beta_2\lambda\mu+6\mu^2+\alpha\phi-S^*\beta_1\phi+\gamma\phi+\gamma_1\phi+\delta\phi-\\
		&V^*\beta_2\lambda\phi+3\mu\phi+I^*\beta_2(\alpha+\gamma+\gamma_1+\delta+\alpha\lambda-S^*\beta_1\lambda-V^*\beta_2\lambda+\gamma\lambda+\gamma_1\lambda+\delta\lambda+3\mu+3\lambda\mu+\lambda\phi)\\
		&+E^*\beta_1(\alpha+\gamma+\gamma_1+\delta+\alpha\lambda-S^*\beta_1\lambda+2I^*\beta_2\lambda-V^*\beta_2\lambda+\gamma\lambda+\gamma_1\lambda+\delta\lambda+3\mu+3\lambda\mu+\lambda\phi).\\
  A_3=& -2S^*\alpha\beta_2\mu+2\alpha\gamma\mu-2S^*\beta_1\gamma\mu+2\alpha\gamma_1\mu-2S^*\beta_1\gamma_1\mu+2\alpha\delta\mu-2S^*\beta_1\delta\mu-2V^*\alpha\beta_2\lambda\mu+3\alpha\mu^2-\\
		&3S^*\beta_1\mu^2+3\gamma\mu^2+3\gamma_1\mu^2+3\delta\mu^2-3V^*\beta_2\lambda\mu^2+4\mu^3+E^{*2}\beta_1^2\lambda(\alpha+\gamma+\gamma_1+\delta+2\mu)+\\
		&I^{*2}\beta_2^2\lambda(2\mu+\alpha+\gamma+\gamma_1+\delta)-S^*\alpha\beta_2\phi+\alpha\gamma\phi-S^*\beta_1\gamma\phi+\alpha\gamma_1\phi-S^*\beta_1\gamma_1\phi+\alpha\delta\phi-\\
		&S^*\beta_1\delta\phi-V^*\alpha\beta_2\lambda\phi+2\alpha\mu\phi-2S^*\beta_1\mu\phi+2\gamma\mu\phi+2\gamma_1\mu\phi+2\delta\mu\phi-2V^*\beta_2\lambda\mu\phi+3\mu^2\phi+\\
		&I^*\beta_2(2\gamma\mu+2\gamma_1\mu+2\delta\mu-2V^*\beta_2\lambda\mu+2\gamma\lambda\mu+2\gamma_1\lambda\mu+2\delta\lambda\mu+3\mu^2+3\lambda\mu^2+\gamma\lambda\phi+\gamma_1\lambda\phi+\delta\lambda\phi+\\
		&2\lambda\mu\phi+\alpha(\gamma+\gamma_1+\delta-S^*\beta_2\lambda-V^*\beta_2\lambda+\gamma\lambda+\gamma_1\lambda+\delta\lambda+2\mu+2\lambda\mu+\lambda\phi)-\\
		&S^*\lambda(-\beta_2\phi+\beta_1(\gamma+\gamma_1+\delta+2\mu+\phi)))+E^*\beta_1(2I^*\beta_2\gamma\lambda+2I^*\beta_2\gamma_1\lambda+2I^*\beta_2\delta\lambda+2\gamma\mu+2\gamma_1\mu+\\
		&2\delta\mu+4I^*\beta_2\lambda\mu-2V^*\beta_2\lambda\mu+2\gamma\lambda\mu+2\gamma_1\lambda\mu+2\delta\lambda\mu+3\mu^2+3\lambda\mu^2+\gamma\lambda\phi+\gamma_1\lambda\phi+\delta\lambda\phi+\\
		&2\lambda\mu\phi+\alpha(\gamma+\gamma_1+\delta+2I^*\beta_2\lambda-S^*\beta_2\lambda-V^*\beta_2\lambda+\gamma\lambda+\gamma_1\lambda+\delta\lambda+2\mu+2\lambda\mu+\lambda\phi)-\\
		&S^*\lambda(-\beta_2\phi+\beta_1(2\mu+\gamma+\gamma_1+\delta+\phi))).\\
  A_4=&I^{*2}\alpha\beta_2^2\gamma\lambda+I^{*2}\alpha\beta_2^2\gamma_1\lambda+I^{*2}\alpha\beta_2^2\delta\lambda+I^*\alpha\beta_2\gamma\mu+I^*\alpha\beta_2\gamma_1\mu+I^*\alpha\beta_2\delta\mu+I^{*2}\alpha\beta_2^2\lambda\mu-I^*S^*\alpha\beta_2^2\lambda\mu\\
		&-I^*V^*\alpha\beta_2^2\lambda\mu+I^*\alpha\beta_2\gamma\lambda\mu-I^*S^*\beta_1\beta_2\gamma\lambda\mu+I^{*2}\beta_2^2\gamma\lambda\mu+I^*\alpha\beta_2\gamma_1\lambda\mu-I^*S^*\beta_1\beta_2\gamma_1\lambda\mu+\\
		&I^{*2}\beta_2^2\gamma_1\lambda\mu+I^*\alpha\beta_2\delta\gamma\mu-I^*S^*\beta_1\beta_2\delta\lambda\mu+I^{*2}\beta_2^2\delta\lambda\mu+I^*\alpha\beta_2\mu^2-S^*\alpha\beta_2\mu^2+\alpha\gamma\mu^2-S^*\beta_1\gamma\mu^2-I^*\beta_2\gamma\mu^2\\
		&+\alpha\gamma_1\mu^2-S^*\beta_1\gamma_1\mu^2+I^*\beta_2\gamma_1\mu^2+\alpha\delta\mu^2-S^*\beta_1\delta\mu^2+I^*\beta_2\delta\mu^2+I^*\alpha\beta_2\lambda\mu^2-V^*\alpha\beta_2\lambda\mu^2-\\
		&I^*S^*\beta_1\beta_2\lambda\mu^2+I^{*2}\beta_2^2\lambda\mu^2-I^*V^*\beta_2^2\lambda\mu^2+I^*\beta_2\gamma\lambda\mu^2+I^*\beta_2\gamma_1\lambda\mu^2+I^*\beta_2\delta\lambda\mu^2+\alpha\mu^3-S^*\beta_1\mu^3+\\
		&I^*\beta_2\mu^3+\gamma\mu^3+\gamma_1\mu^3+\delta\mu^3+I^*\beta_2\lambda\mu^3-V^*\beta_2\lambda\mu^3+\mu^4+E^{*2}\beta_1^2\lambda(\alpha+\mu)(\mu+\delta+\gamma+\gamma_1)+\\
		&I^*\alpha\beta_2\gamma\lambda\phi-I^*S^*\beta_1\beta_2\gamma\lambda\phi+I^*\alpha\beta_2\gamma_1\lambda\phi-I^*S^*\beta_1\beta_2\gamma_1\lambda\phi+I^*\alpha\beta_2\delta\lambda\phi-I^*S^*\beta_1\beta_2\delta\lambda\phi-\\
		&S^*\alpha\beta_2\mu\phi+\alpha\gamma\mu\phi-S^*\beta_1\gamma\mu\phi+\alpha\gamma_1\mu\phi-S^*\beta_1\gamma_1\mu\phi+\alpha\delta\mu\phi-S^*\beta_1\delta\mu\phi+I^*\alpha\beta_2\lambda\mu\phi-\\
		&V^*\alpha\beta_2\lambda\mu\phi-I^*S^*\beta_1\beta_2\lambda\mu\phi+I^*S^*\beta_2^2\lambda\mu\phi+I^*\beta_2\gamma\lambda\mu\phi+I^*\beta_2\gamma_1\lambda\mu\phi+I^*\beta_2\delta\lambda\mu\phi+\alpha\mu^2\phi-S^*\beta_1\mu^2\phi+\\
		&\gamma\mu^2\phi+\gamma_1\mu^2\phi+\delta\mu^2\phi+I^*\beta_2\lambda\mu^2\phi-V^*\beta_2\lambda\mu^2\phi+\mu^3\phi+E^*\beta_1(-S^*\beta_1\gamma\lambda\mu-S^*\beta_1\gamma_1\lambda\mu-\\
		&S^*\beta_1\delta\lambda\mu+\gamma\mu^2+\gamma_1\mu^2+\delta\mu^2-S^*\beta_1\lambda\mu^2-V^*\beta_2\lambda\mu^2+\gamma\lambda\mu^2+\gamma_1\lambda\mu^2+\delta\lambda\mu^2+\mu^3+\lambda\mu^3+\\
		&2I^*\beta_2\lambda(\alpha+\mu)(\gamma+\gamma_1+\delta+\mu)-S^*\beta_1\gamma\lambda\phi-S^*\beta_1\gamma_1\lambda\phi-S^*\beta_1\delta\lambda\phi-S^*\beta_1\lambda\mu\phi+S^*\beta_2\lambda\mu\phi+\\
		&\gamma\lambda\mu\phi+\gamma_1\lambda\mu\phi+\delta\lambda\mu\phi+\lambda\mu^2\phi+\alpha(\delta\mu-S^*\beta_2\lambda\mu-V^*\beta_2\lambda\mu+\delta\lambda\mu+\mu^2+\lambda\mu^2+\delta\lambda\phi+\lambda\mu\phi+\\
		&\gamma(\mu+\lambda\mu+\lambda\phi)+\gamma_1(\mu+\lambda\mu+\lambda\phi))).
	\end{align*}
	
	As a consequence of the Routh-Hurwitz criteria \cite{Stability Bound-8, Stability Bound-16}, every root of this bi-quadratic equation possesses a negative real part if and only if $A_1, A_2,  A_3, A_4 >0,\; A_1A_2>A_3$ and $A_1A_2A_3>A_3^2+A_1^2A_4$.
	For $\mathcal{R}_0>1$,
	\begin{align*}
		S_0\alpha\beta_2+S_0\beta_1(\mu+\delta+\gamma+\gamma_1)+V_0\beta_2\lambda(\alpha+\mu)>(\alpha+\mu)(\mu+\delta+\gamma+\gamma_1).
	\end{align*}
	Now, rewriting $A_1$ in terms of $\mathcal{R}_0$,
	\begin{align*}
		& I^*\beta_2\lambda+E^*\beta_1(1+\lambda)+4(\mu+\phi)-\mu+(\alpha+\mu)(\mu+\delta+\gamma+\gamma_1)(\mathcal{R}_0-1).
	\end{align*}
	Hence, the nature $\mathcal{R}_0>1$ is mandatory, in order to satisfy $A_1>0$ and $A_1^2>0.$
	Again, from the expression of $A_2$, 
	\begin{align*}
		&(\alpha+\mu)(\mu+\delta+\gamma+\gamma_1)(\mathcal{R}_0-1)+\lambda(\gamma+\gamma_1+3\mu+\phi+\delta+\alpha)+2\mu.
	\end{align*}
Further, from the coefficient of $E^*\beta_1$ in $A_2$, 
	\begin{align*}
		&\alpha+\gamma+\gamma_1+\delta+\alpha\lambda-S^*\beta_1\lambda+2I^*\beta_2\lambda-V^*\beta_2\lambda+\gamma\lambda+\gamma_1\lambda+\delta\lambda+3\mu+3\lambda\mu+\lambda\phi\\
		=&\lambda(\gamma+\gamma_1+3\mu+\phi+\delta+\alpha)+2\mu+(\alpha+\mu)(\mu+\delta+\gamma+\gamma_1)(\mathcal{R}_0-1).
	\end{align*}
	Therefore, from the above expression we find that it requires $\mathcal{R}_0>1$ to satisfy $A_2>0.$
	Now, from the expression of $A_3$, the coefficient of $I^*\beta_2$ can be expressed in terms of $\mathcal{R}_0$ as follows,
	\begin{align*}
		&\alpha((\alpha+\mu)(\mu+\delta+\gamma+\gamma_1)(\mathcal{R}_0-1)+\mu+\lambda(\gamma+\gamma_1+2\mu+\delta+\phi))-\lambda(\alpha+\mu)(\mu+\delta+\gamma+\gamma_1)\\
		&(\mathcal{R}_0-1)+\lambda\mu\beta_1.
	\end{align*}
	Thus, we find that it is necessary $\mathcal{R}_0>1$ in order to satisfy $A_3>0$ and $A_3^2>0$. Hence, it shows that $A_1A_2>A_3$ for $\mathcal{R}_0>1$. Finally, from the expression of $A_4$ taking a portion as coefficient of $E^*\beta_1$,
	\begin{align*}
		&\alpha(\delta\mu-S^*\beta_2\lambda\mu-V^*\beta_2\lambda\mu+\delta\lambda\mu+\mu^2+\lambda\mu^2+\delta\lambda\phi+\lambda\mu\phi+\gamma(\mu+\lambda\mu+\lambda\phi)+\gamma_1(\mu+\lambda\mu+\lambda\phi))\\
		=&\alpha(\lambda\mu(\alpha+\mu)(\mu+\delta+\gamma+\gamma_1)(\mathcal{R}_0-1)+\mu^2+\delta\lambda\phi+(\gamma+\gamma_1)(\mu+\lambda\mu+\lambda\phi)).
	\end{align*}
	Clearly, $A_4>0$ for $\mathcal{R}_0>1.$
	Therefore, from the Routh-Hurwitz Criterion \cite{Stability Bound-17, Stability Bound-21} for the Jacobian matrix with characteristic polynomial of degree $n=4$, combining all the above expressions $A_1A_2A_3>A_3^2+A_1^2A_4$ is satisfied for $\mathcal{R}_0>1$.\\
	Therefore, all the roots of the characteristic equation will possess a negative real part. Consequently, when $\mathcal{R}_0 >1$, the endemic equilibrium point will be locally asymptotically stable. In contrast, if $\mathcal{R}_0<1$, suggests that the infected state is unstable as the Jacobian includes at least one positive eigenvalue. This establishes the validity of the conclusion.
\end{proof}

\begin{Th} \cite{Stability Bound-11}
	The endemic equilibrium state $(\mathcal{E}^*)$ of the system \eqref{equi} is globally asymptotically stable for $\mathcal{R}_0 >1$ and the disease persists.
\end{Th}
\begin{proof}
    The proof is excluded as a comparable result can be found in \cite{Stability Bound-11}.
\end{proof}

\begin{Th}\label{gs_EE2}
	The endemic equilibrium state $(\mathcal{E}^*)$ of the model \eqref{new_model} is globally asymptotically stable provided that $\mathcal{R}_0 >1$ \cite{Stability Bound-18, Stability Bound-16}.
\end{Th}
\begin{proof}
	Considering the model \eqref{new_model} and $\mathcal{R}_0 >1$, so that the associated unique endemic equilibrium $\mathcal{E}^*$ of the model exists (which proved already). To examine the global stability of $\mathcal{E}^*$ we have considered the following non-linear Lyapunov function of the Goh-Voltera type,
	\begin{align*}
		V=\left(S-S^{**}-S^{**}\ln\frac{S}{S^{**}}\right)+\left(E-E^{**}-E^{**}\ln\frac{E}{E^{**}}\right)+K\left(I-I^{**}-I^{**}\ln\frac{I}{I^{**}}\right)\\
		\text{where }\; K=\frac{(\beta_1+\beta_2)S^*I^*}{\alpha E^*}.
	\end{align*}
	Notice that, $V$ is non-negative, and becomes identically zero if and only if it is evaluated at the non-negative endemic equilibrium state $\mathcal{E}^*$. By performing the derivative of V along the solution curves of (\ref{new_model}) yields,
	\begin{align*}
		V'=S'\left(1-\frac{S^*}{S}\right)+E'\left(1-\frac{E^*}{E}\right)+KI'\left(1-\frac{I^*}{I}\right).
	\end{align*}
	Here, prime $(')$ denotes the derivatives. Substituting the derivatives $(S', E', I')$ into this equation from (\ref{new_model}) we have,
	\begin{flalign*}
		V'=& [\Lambda-(\beta_1 E+\beta_2 I)S-(\mu +\phi)S]\left(1-\frac{S^*}{S}\right)+[(\beta_1 E+\beta_2 I)S-(\alpha+\mu)E]\left(1-\frac{E^*}{E}\right)+&\\
		&K[\alpha E-(\mu+\delta+\gamma+\gamma_1)I]\left(1-\frac{I^*}{I}\right).
	\end{flalign*}
	Here, we have considered the simplified case of our model by assuming $\lambda=(1-\varepsilon)=0$ i.e., the vaccination rate is $100$ effective.\\
	At the steady state from equation \eqref{equi} we have,
	\begin{align*}
		\Lambda = (\beta_1 E^*+\beta_2 I^*)S^*+(\mu+\phi)S^*,\;
		\text{and}\\ (\beta_1 E^*+\beta_2 I^*)S^*=(\alpha+\mu)E^*.
	\end{align*}
	Substituting these relations in the expression of $V'$ we have obtained that,
	\begin{flalign*}
		V'=&[(\beta_1E^*+\beta_2I^*)S^*+(\mu+\phi)S^*-(\beta_1 E+\beta_2 I)S-(\mu+\phi)S -(\beta_1 E^*+\beta_2 I^*)\frac{S^{*2}}{S}- &\\
		&\frac{(\mu+\phi)S^{*2}}{S}+(\beta_1 E+\beta_2 I)S^*+(\mu+\phi)S^*]+[\beta_1S(E-E^*)+\beta_2IS\left(1-\frac{E^*}{E}\right)-&\\
		&A_1(E-E^*)]+\left[K\alpha E\left(1-\frac{I^*}{I}\right)-KA_2(I-I^*)\right].
	\end{flalign*}
	Here, $(\alpha+\mu)=A_1$, and $(\mu+\delta+\gamma+\gamma_1)=A_2$.
	By collecting all the infected classes without the star $(*)$ and equating to zero we obtain,
	\begin{align*}
		-(\beta_1 E+\beta_2 I)S+(\beta_1 E+\beta_2 I)S^* +\beta_1SE+\beta_2IS-A_1E+K\alpha E-KA_2I=0.
	\end{align*}
	A little perturbation of steady state results in,
	\begin{align}\label{eqn4.6}
		K=\frac{S^*(\beta_1+\beta_2)}{A_2},\; A_1=\frac{(\beta_1+\beta_2)S^*I^*}{E^*},\; \alpha=\frac{A_2I^*}{E^*}.
	\end{align}
	Substituting the expression from \eqref{eqn4.6} into the expression of $V'$,
	\begin{flalign*}
		V'=&\left[(\beta_1 E^*+\beta_2 I^*)S^*+2(\mu+\phi)S^*-(\beta_1 E+\beta_2 I)S\right]+&\\
		&\left[-(\mu+\phi)S-(\beta_1 E^*+\beta_2 I^*)\frac{S^{*2}}{S}-\frac{(\mu+\phi)S^{*2}}{S}+(\beta_1 E+\beta_2 I)S^*\right]+&\\
		&\left[\beta_1S(E-E^*)+\beta_2IS\left(1-\frac{E^*}{E}\right)-\frac{(\beta_1+\beta_2)S^*I^*}{E^*}(E-E^*)\right]+ &\\
		&\left[\frac{S^*(\beta_1+\beta_2)I^*}{E^*}\left(1-\frac{I^*}{I}\right)-S^*(\beta_1+\beta_2)(I-I^*)\right].\\
  =&(\mu+\phi)S^*\left[2-\frac{S}{S^*}-\frac{S^*}{S}\right]+&\\
		&\beta_1E^*S^*\left[1-\frac{S^*}{S}+\frac{E}{E^*}-\frac{S}{S^*}-\frac{E}{E^{*2}}+\frac{1}{E^*}+\frac{I^*}{E^{*2}}-\frac{I^{*2}}{E^{*2}I}-\frac{I}{E^*}+\frac{I^*}{E^*}\right]+&\\
		&\beta_2I^*S^*\left[3-\frac{S^*}{S}+\frac{I}{I^*}-\frac{I}{I^*}\frac{S}{S^*}\frac{E^*}{E}-\frac{E}{E^*}+\frac{1}{E^*}-\frac{I^*}{E^*I}-\frac{I}{I^*}\right].
	\end{flalign*}
	Hence, we have got the new form of $V'$. Now, the coefficient of $\beta_1E^*S^*$ gives,
	\begin{align*}
		\left(1-\frac{S^*}{S}-\frac{S}{S^*}\right)+\left(\frac{E}{E^*}+\frac{1}{E^*}-\frac{E}{E^*}\right)+\frac{I^*}{E^{*2}}\left(1-\frac{I^*}{I}\right)+\left(\frac{I^*}{E^*}-\frac{I}{E^*}\right)\leq 0.
	\end{align*}
	Finally, as the arithmetic mean exceeds the geometric mean. Hence, the following inequality from immediate expression of $V^*$ (for $n=3$) results into,
	\begin{align*}
		\left(2-\frac{S}{S*}-\frac{S^*}{S}\right)\leq 0, \;\text{and}\;\left(3-\frac{S^*}{S}+\frac{I}{I^*}-\frac{I}{I^*}\frac{S}{S^*}\frac{E^*}{E}-\frac{E}{E^*}+\frac{1}{E^*}-\frac{I^*}{E^*I}-\frac{I}{I^*}\right)\leq 0,
	\end{align*}
	which shows that each of the resulting terms is non-positive. Thus, with  these conditions we conclude that $V'(t)\leq 0$ for all positive values of $\{S,E,I\}$ i.e., if $S=S^*,\; E=E^*,\; I=I^*$, and also for $\mathcal{R}_0>1$. Moreover, the strict equality $V'=0$ holds only for $S=S^*,\;E=E^*,\;\text{and}\;I=I^*$.\\
	Hence, the maximum invariance set $\{(S,E,I)\in \Omega : V'(t)=0\}$ is the singleton $\{\mathcal{E}^*\}$, where $\{\mathcal{E}^*\}$ is the endemic equilibrium point.\\
	Thus, as a consequence of LaSalle's Invariance principle, the endemic equilibrium point $(\mathcal{E}^*)$ is globally asymptotically stable in the set $\Omega$ when $\mathcal{R}_0>1$. In other words, every solution to the equations of the model \eqref{new_model} converges to the corresponding unique endemic equilibria $(\mathcal{E}^*)$, of the model as $t\rightarrow \infty$ for $\mathcal{R}_0>1$.
\end{proof}
\subsection{Stability and Persistence of the System}\label{Subsection-Stability and Persistence}
\begin{Th} \cite{Optimal Control-1}.
	The DFE $\mathcal{E}^0$ of the model \eqref{new_model} is a global attractor.
\end{Th}
\begin{proof}
    Omitting the proof, as a similar result can be referenced in \cite{Optimal Control-1}.
\end{proof}

\begin{Th}
	\cite{Stability Bound-8, Optimal Control-1}. If $\mathcal{R}_0 > 1$, the system described by \eqref{new_model} exhibits uniform persistence, implying the existence of a constant $\xi > 0$. This constant ensures that for any initial data $\xi $ in $\Omega$: $\liminf\limits_{t\rightarrow \infty} S(t) > \xi$, $\liminf\limits_{t\rightarrow \infty} E(t) > \xi$, $\liminf\limits_{t\rightarrow \infty} I(t) > \xi$, and $\liminf\limits_{t\rightarrow \infty} V(t) > \xi$.
Remarkably, the value of $\xi$ remains independent of the initial data in $\Omega$.

\end{Th}
\begin{proof}
	When $t\rightarrow \infty$ from system \eqref{new_model}, we have the following limiting system,
	\begin{align}\label{pers_eqn}
		\begin{cases}
			\vspace{0.2cm}
			\displaystyle S'=\Lambda-\frac{\mu}{\Lambda}(\beta_1 \xi E+\beta_2 I)S-(\mu+\phi)S.\\
			\vspace{0.2cm}
			\displaystyle V'=\phi S-\frac{\mu}{\Lambda}(1-\varepsilon)(\beta_1 \xi E+\beta_2 I)V-\mu V.\\
			\vspace{0.2cm}
			\displaystyle E'=\frac{\mu}{\Lambda}(\beta_1 \xi E+\beta_2 I)S-(\alpha+\mu)E.\\
			\displaystyle I'=\alpha E+\frac{\mu}{\Lambda}(1-\varepsilon)(\beta_1 \xi E+\beta_2 I)V-(\mu+\delta+\gamma+\gamma_1)I.
		\end{cases}
	\end{align}	For the case of notation, we still use the notation $\mathcal{E}^0$ to denote the DFE of the equation \eqref{pers_eqn}. Define,
	\begin{align*}
		X =&\{(S,E,I,V) : S \geq 0,E \geq 0,I \geq 0,V \geq 0 \},\\
		X_0 =&\{(S,E,I,V) : S > 0,E > 0,I > 0,V > 0 \},\;\text{and}\\
		\partial X_0= & X\setminus X_0.
	\end{align*}
	It is often suffices to show that \eqref{pers_eqn} is uniformly persistent with respect to $(X_0, \partial X_0)$.\\
	Firstly, by the form of \eqref{pers_eqn}, it is easy to see that both $X$ and $X_0$ are positively invariant. $\partial X_0$ is relatively closed in $X$ and system \eqref{pers_eqn} is point dissipative. Consider the following set using solutions $(S(t),V(t),E(t),I(t))$ of the system \eqref{pers_eqn}.
	\begin{align*}
		M_{\partial}=\{(S(0),E(0),I(0),V(0)) : (S(t),E(t),I(t),V(t))\in \partial  X_0, \forall\;t\geq 0\}.
	\end{align*}
	We now show that,
	\begin{align}\label{MdelSV}
		M_{\partial}=\{S,V,0,0 : S\geq 0 \;,\; V\geq 0\}.
	\end{align}
	Assume that $(S(0),V(0),E(0),I(0)) \in M_{\partial}$. It suffices to show that $E(t)=I(t)=0$ for all $t \geq 0$.\\
	Suppose not, then there exists a $t_0\geq 0$ such that $(E(t_0)>0,I(t_0)=0)$ or, $(E(t_0)=0,I(t_0)>0)$.\\
	For, $(E(t_0)>0,I(t_0)=0)$ we have,
	\begin{align*}
		I'(t_0)=\alpha\varepsilon E(t_0)>0.
	\end{align*}
	It follows that there is an $\epsilon_0>0$ such that $I(t)>0$ for $t_0<t<t_0+\epsilon_0$. Clearly, we can restrict $\epsilon_0>0$ small enough such that $E(t_0)>0$ for $t_0<t<t_0+\epsilon_0$. This means that $(S(t),V(t),E(t),I(t))\notin \partial X_0$ for $t_0<t<t_0+\epsilon_0$, which contradicts the assumption that  $(S(0),V(0),E(0),I(0)) \in M_{\partial}$. For other cases, we can show these contradict the assumption that $(S(0), V(0), E(0), I(0)) \in M_{\partial}$ respectively. Thus (\ref{MdelSV}) holds.\\
	{Note that, $\mathcal{E}^0$ (DFE point) is globally asymptotically stable in Int$M_{\partial}$ (interior of $M_{\partial}$)}. Moreover, $\mathcal{E}_0$ is isolated invariant set in $X$, every orbit in $M_{\partial}$ converges to $\mathcal{E}_0$ and $\mathcal{E}_0$ is acyclic in $M_{\partial}$. We only need to show that $W^S(\mathcal{E}_0)\cap X_0=\emptyset$ if $\mathcal{R}_0>1$.\\
	In the following, we prove that $W^S(\mathcal{E}_0)\cap X_0=\emptyset$. Suppose not, that is, $W^S(\mathcal{E}_0)\cap X_0\neq \emptyset$. Then there exists a positive solution $(\widetilde{S}(t),\widetilde{V}(t),\widetilde{E}(t),\widetilde{I}(t))$,
	with $(\widetilde{S}(0),\widetilde{V}(0),\widetilde{E}(0),\widetilde{I}(0)) \in X_0$ such that $(\widetilde{S}(t),\widetilde{V}(t),\widetilde{E}(t),\widetilde{I}(t))\rightarrow\mathcal{E}_0$ as $t\rightarrow +\infty$. Since, $\mathcal{R}_0>1$, we can choose a $\eta >0 $ small enough such that, $$\mathcal{R}_0-\eta\frac{\mu}{\Lambda}\mathcal{R}_0 > 1.$$
	Thus, when $t$ is sufficiently large, we have,
	\begin{align*}
		S_0-\eta \leq \widetilde{S}(t)\leq S_0+\eta,\; 0\leq  \widetilde{V}(t)\leq \eta,\; 0\leq  \widetilde{E}(t)\leq \eta ,\; 0\leq  \widetilde{I}(t)\leq \eta, \; \text{and}
	\end{align*}
	\begin{align*}
		\begin{cases}
			\vspace{0.2cm}
			\displaystyle V'=\phi(S_0-\eta)-\frac{\mu}{\Lambda}(1-\varepsilon)(\beta_1\xi E+\beta_2 I)(S_0-\eta)-\mu V.\\
			\vspace{0.2cm}
			\displaystyle E'=\frac{\mu}{\Lambda}(\beta_1\xi E+\beta_2 I)(S_0-\eta)-(\alpha+\mu) E.\\
			\displaystyle I'=\alpha E+\frac{\mu}{\Lambda}(1-\varepsilon)(\beta_1\xi E+\beta_2 I)(S_0-\eta)-(\mu+\delta+\gamma+\gamma_1)I.
		\end{cases}
	\end{align*}
	By the comparison principle, it is easy to see that $\widetilde{E}(t)\rightarrow +\infty$, $\widetilde{I}(t)\rightarrow +\infty$ as $t \rightarrow +\infty$, which contradicts $\widetilde{E}(t)\rightarrow 0$, $\widetilde{I}(t)\rightarrow 0$ as $t \rightarrow +\infty$. This proves that, $$W^S(\mathcal{E}^0)\cap X_0 = \emptyset.$$
	Since $W^S(\mathcal{E}^0)\cap X_0 = \emptyset$, $\bigcup_{x\in M_{\partial}}\omega(x)=\{\mathcal{E}^0\}$, $\mathcal{E}^0$ is isolated invariant set in $X$, and $\mathcal{E}^0$ is acyclic in $M_{\partial}$, thus we are able to conclude that the system (\ref{pers_eqn}) is uniformly persistent with respect to $(X_0,\partial X_0)$. Then, the system (\ref{new_model}) is uniformly persistent.
\end{proof}

\section{Numerical Simulation}\label{Section-Numerical-Simulation}
In this section, we analyze the model \eqref{new_model} numerically to support the analytical results presented in previous sections.  Numerical simulations enable the incorporation of real data, such as population demographics, disease parameters, and contact patterns, to create more accurate and realistic models. This also facilitate parameter estimation by comparing model predictions with observed data, optimizing parameters to minimize the difference between simulated and actual outcomes \cite{Stability Bound-4, Stability Bound-25}. By modifying parameters or implementing different control measures within the simulation, researchers can assess the effectiveness of interventions such as vaccination campaigns, social distancing, or contact tracing. Moreover, numerical simulations allow for forecasting the progression of an epidemic under different scenarios and can aid decision-makers in making informed choices regarding resource allocation, healthcare capacity, and policy implementation. Relationship between parameters and $\mathcal{R}_0$ along with different compartments in model \eqref{new_model} are presented graphically. Also, we have presented phase plane analysis of the model \eqref{new_model} to realize the scenario at different situations.
\subsection{Numerical Analysis of Compartments at $\mathcal{R}_0 < 1$}
In epidemiology, the compartmental model typically used to analyze the spread of infectious diseases in the community. We can provide an interpretation of the scenario when the basic reproduction number ($\mathcal{R}_0$) is less than 1 in the context of the SVEIRT model. When ($\mathcal{R}_0$) is less than 1, it means that, on average, each infected individual will infect fewer than one other person. 
Here's how the scenario from Figure \ref{basicgraphless1} may be interpreted for the compartments in model \eqref{new_model}:\\
Initially, the majority of the population is susceptible to the disease. As the epidemic progresses, some individuals will become infected, while others may be vaccinated or naturally acquire immunity. The number of infected individuals will initially increase, indicating the growth phase of the epidemic. However, since $\mathcal{R}_0<1$, the rate of new infections will decrease over time for model parameter values as in Table \ref{tableparameter}. This can be due to various factors such as public health interventions, social distancing measures, or increasing immunity in the population. Also, the number of recovered individuals will gradually increase as infected individuals either recover from the disease or succumb to it. In Figure \ref{basicgraphless1}, we see that the total number of infected human with treatment decreases faster than the total number of infected human with treatment.
\begin{figure}[H]
	\centering  
	\includegraphics[width=4 in]{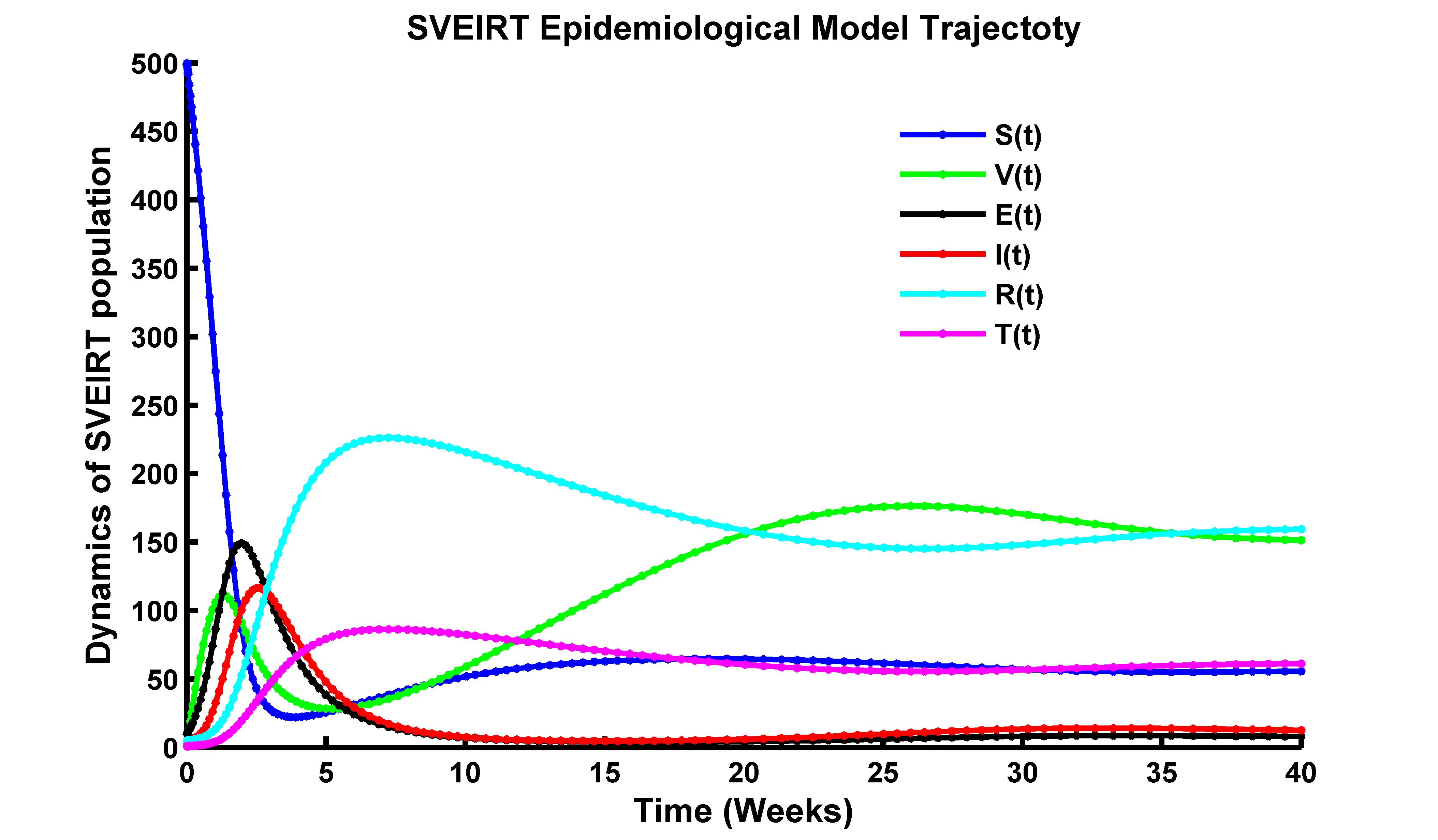} 
	\caption{Simulation of the dynamics of $S(t),\; V(t),\; E(t),\; I(t),\; R(t),\; T(t)$ population over time with $S(0)=500,\; V(0)=1,\; E(0)=1,\; I(0)=1,\; R(0)=0,\; T(0)=1$ and $\mathcal{R}_0<1$, where all the parameter values are taken from Table \ref{tableparameter}.}\label{basicgraphless1}
\end{figure}
\noindent
As the epidemic slows down and the number of new infections decreases, the number of recoveries will eventually surpass the number of new infections, leading to a decline in the active cases. In this case, both exposed and vaccinated individuals contribute to reducing the susceptible population and can further decrease the transmission rate of the disease. Time $(t)$ is an essential aspect of the analysis, as it represents the progression of the epidemic. As $\mathcal{R}_0<1$, the epidemic curve will eventually reach its peak and then start to decline. From Figure \ref{basicgraphless1}, we see that, from 1 to 5 weeks $S(t)$ population decreases to 50, after than increases up to 80. The Exposed population goes pick level after 5 weeks ($E(t)=150$), and same scenario for infected population after 5 weeks ($I(t)=120$). The total duration of the epidemic can vary depending on various factors such as the disease's characteristics, population size, and intervention measures. Thus Figure \ref{basicgraphless1} illustrates that, after a long time total number of exposed and infected population converges to zero level. Susceptible, recovered, treated population remains in the community. Also from Figure \ref{basicgraphless1}, the recovered and treated population gradually increase after 5 weeks, the after 20 weeks $R(t)$ and $T(t)$ population goes parallelly. Hence numerical analysis supports the analytical result at DFE when $\mathcal{R}_0<1$. This indicates that the disease is not spreading efficiently within the population, and the epidemic is likely to decline and eventually die out.

\subsection{Numerical Analysis of Compartments at $\mathcal{R}_0 > 1$}
When the basic reproduction number $\mathcal{R}_0>1$, it indicates that each infected individual, on average, is infecting more than one susceptible individual. Here's a short description of each compartment's behaviour in model \eqref{new_model} in such a scenario:\\
As the epidemic progresses, the Figure \ref{basicgraphgrt1} illustrates that, the number of susceptible individuals decreases as they become infected or vaccinated. In an $\mathcal{R}_0>1$ scenario, the number of infected individuals initially increases rapidly, indicating a growing outbreak. As time progresses and interventions take place, the number of infected individuals may eventually decrease. Further, in the early stages of the epidemic, the number of recovered individuals is low. However, as the epidemic progresses, the number of recoveries gradually increases. The number of exposed individuals initially increases, reflecting the transmission from infected individuals. After the latent period, they transition to the infected compartment. The number of vaccinated individuals increases over time, reducing the susceptible population and potentially slowing down the spread of the disease.

For model parameter values taken from Table \ref{tableparameter}, Figure \eqref{new_model} illustrates that, susceptible population decreases up to 5 weeks, after that goes to a constant level. When control with vaccination started, this $V(t)$ compartment significantly grows to 90, then decreases as $E(t)$ and $I(t)$ increases. After 4 weeks, $E(t)$ and $I(t)$ population goes peak level at 160 and 125 respectively. After 5 weeks, the scenario of $E(t)$ and $I(t)$ is gradually decreasing. Meanwhile, $R(t)$ and $T(t)$ population goes peak level after 5 weeks. After 10 weeks, $R(t)$ and $T(t)$ population becomes parallel level at 180 and 60 respectively.
\begin{figure}[H]
	\centering  
	\includegraphics[width=4 in]{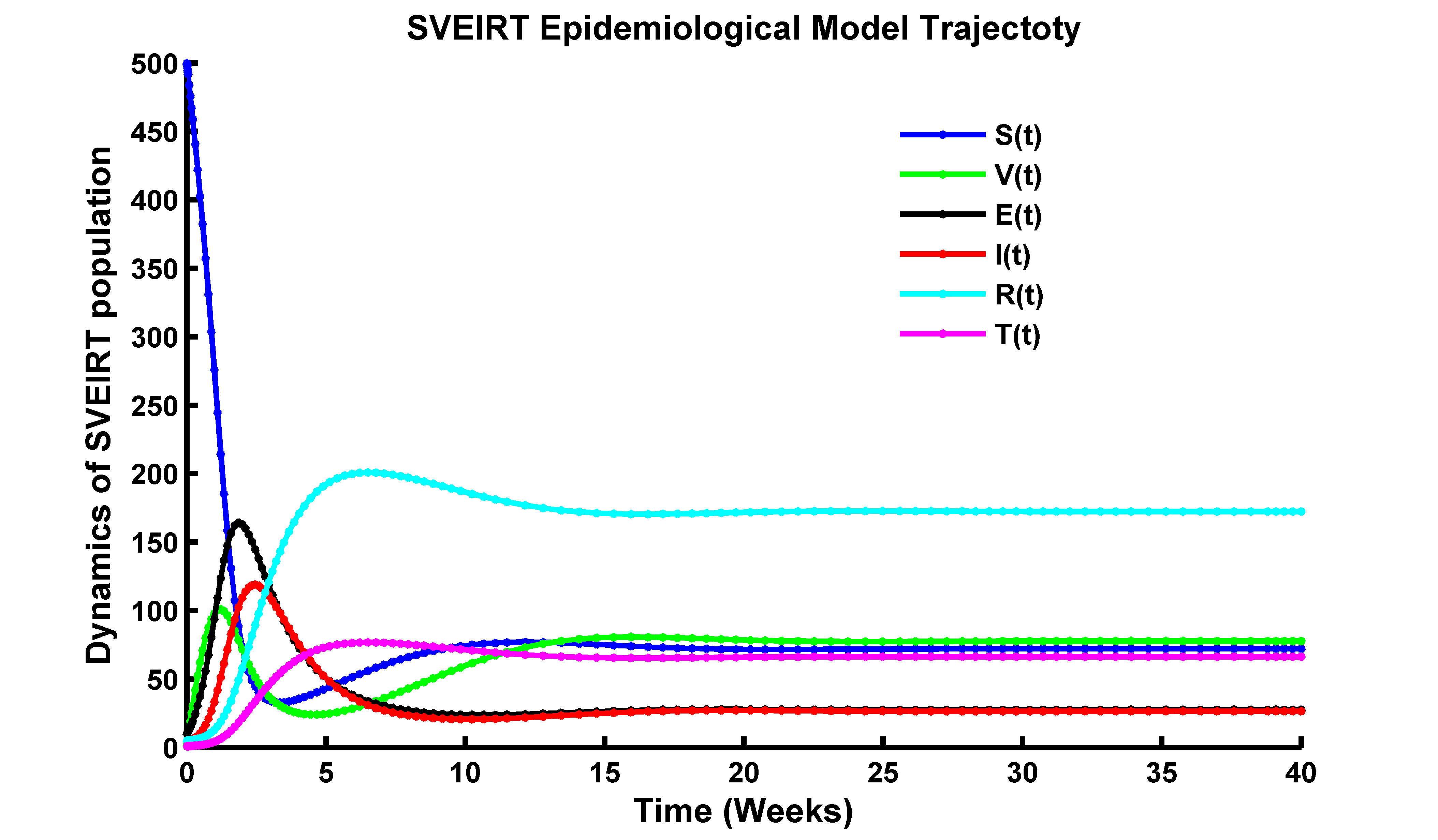} 
	\caption{Simulation of the dynamics of $S(t),\; V(t),\; E(t),\; I(t),\; R(t),\; T(t)$ population over time with $S(0)=500,\; V(0)=1,\; E(0)=1,\; I(0)=1,\; R(0)=0,\; T(0)=1$ and $\mathcal{R}_0>1$, where all the parameter values are taken from Table \ref{tableparameter}.}\label{basicgraphgrt1}
\end{figure}
\noindent
From the Figure \eqref{new_model} we see that, the more $E(t)$ and $I(t)$ decreases, the $R(t)$ and $T(t)$ population grows significantly in the community. Also, we see from the Figure \eqref{new_model} that, when $\mathcal{R}_0>1$ the total susceptible, vaccinated , recovered, treated population remains in the community. After a long time total infected and exposed population became parallel to susceptible population and never approaches to zero level. Therefore, Figure \ref{basicgraphgrt1} presented for $\mathcal{R}_0>1$, supports the stability of EE point numerically. That means, we observe that EE of the model \eqref{new_model} exists and it is locally asymptotically stable for $\mathcal{R}_0>1$. This signifies an epidemic scenario where the disease is spreading rapidly within the population i.e. disease persists in the community. 
\subsection{Phase Plane Analysis of Different Compartments of the Model}
The phase plane is a graphical representation that allows us to analyze the dynamics of a system of differential equations, such as those used in infectious disease modeling. It helps visualize the interactions and trajectories of different compartments (e.g., susceptible, infected) over time. By examining the stability of equilibrium points and the shape of trajectories, phase plane analysis provides valuable information for understanding disease spread, identifying critical thresholds, and evaluating the impact of interventions in a concise and intuitive manner \cite{Stability Bound-25, Bifurcation of R0-7, Bifurcation of R0-8}.
\begin{figure}[H]
	\centering  
	\subfloat[]{\includegraphics[width=2. in]{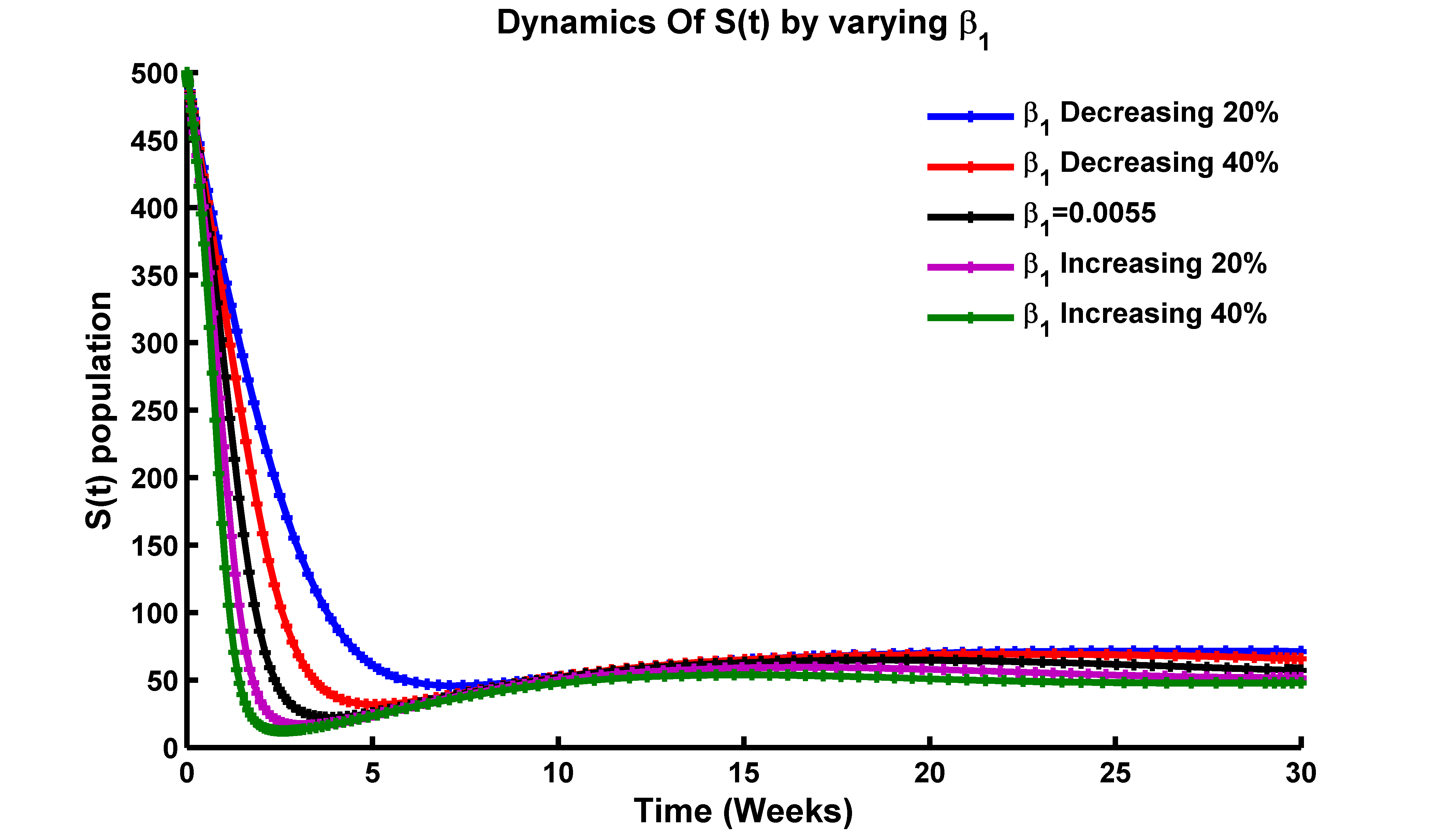}}
	\subfloat[]{\includegraphics[width=2. in]{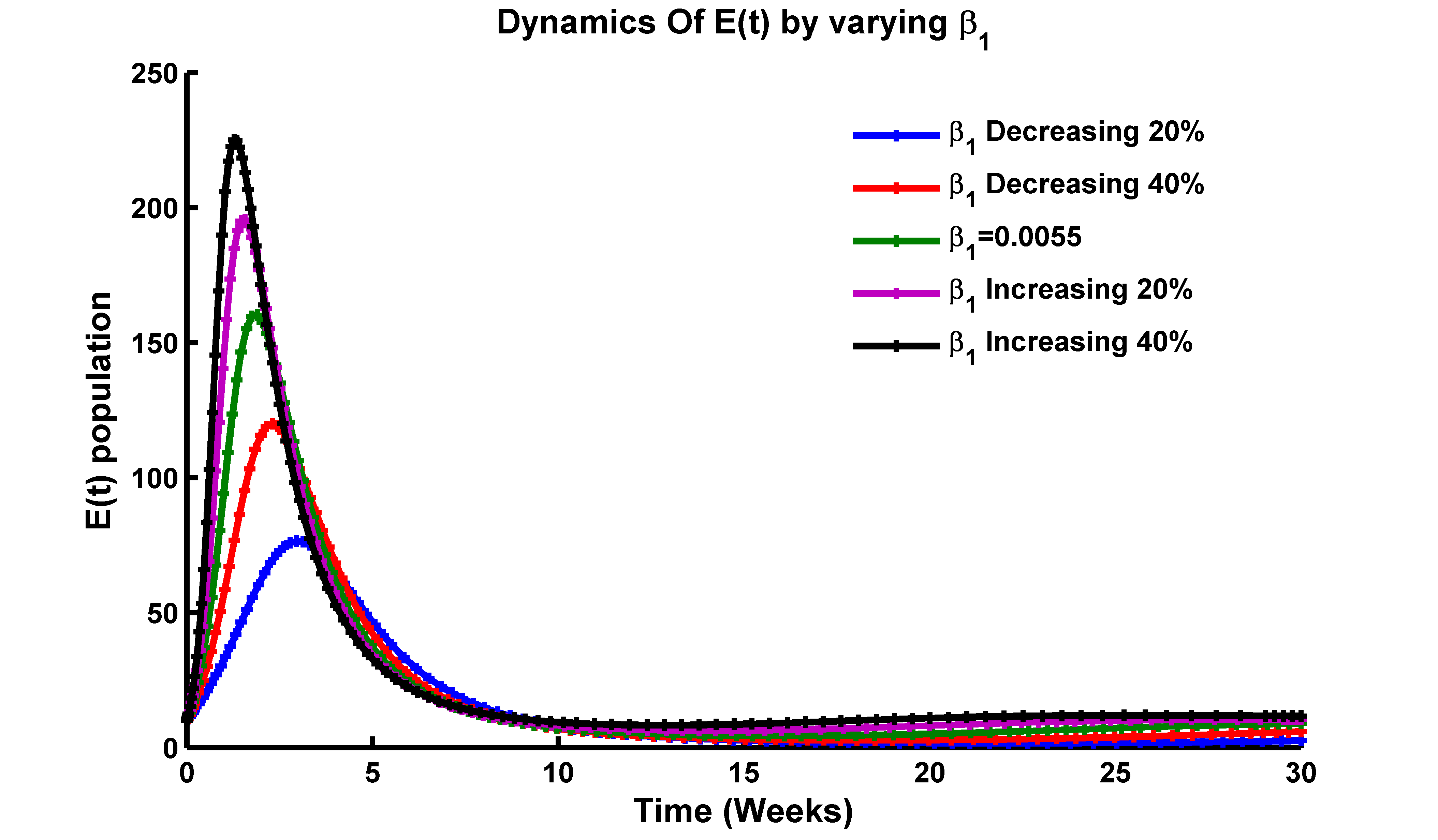}}
	\subfloat[]{\includegraphics[width=2. in]{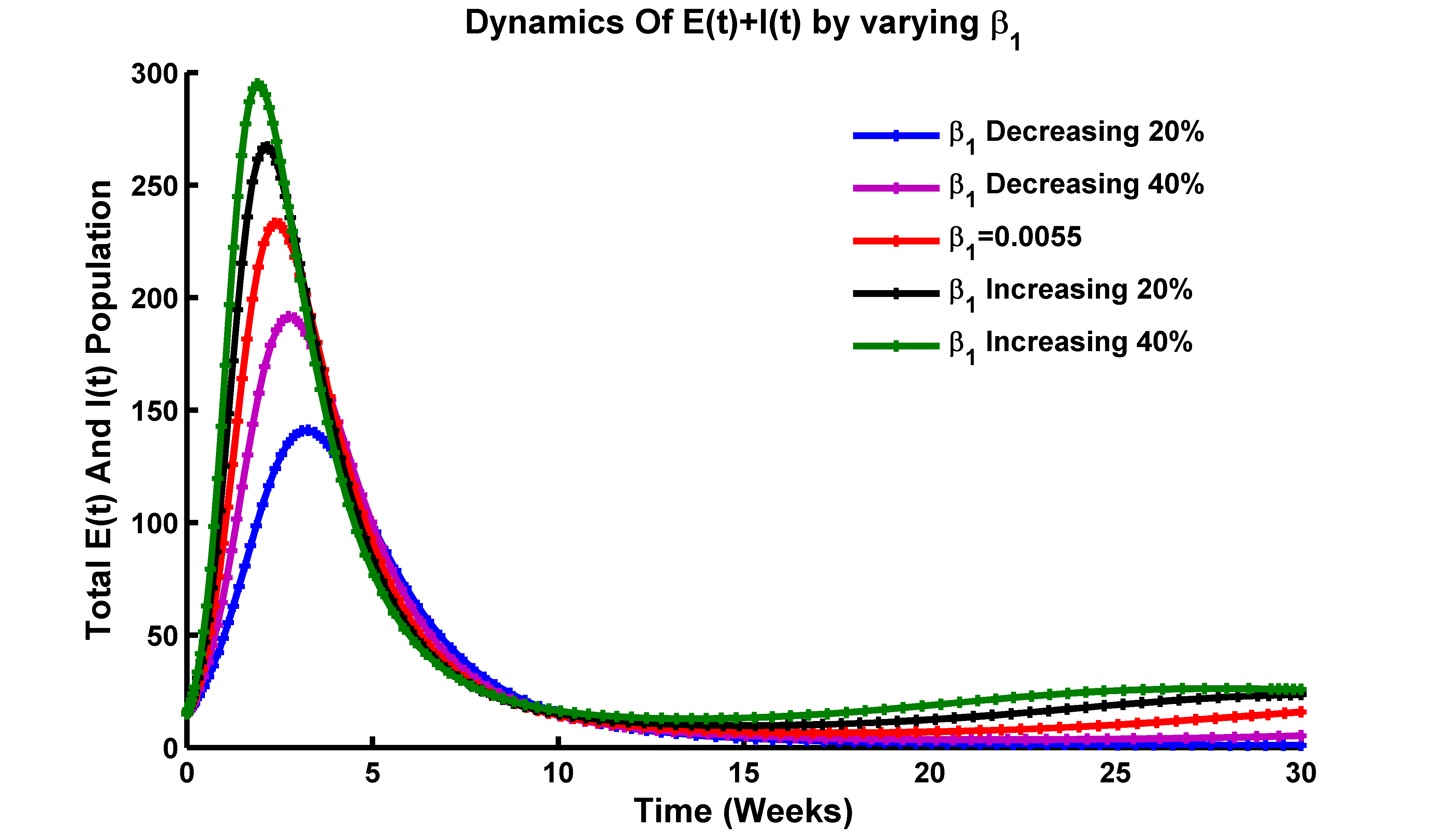}}\\
	\subfloat[]{\includegraphics[width=2. in]{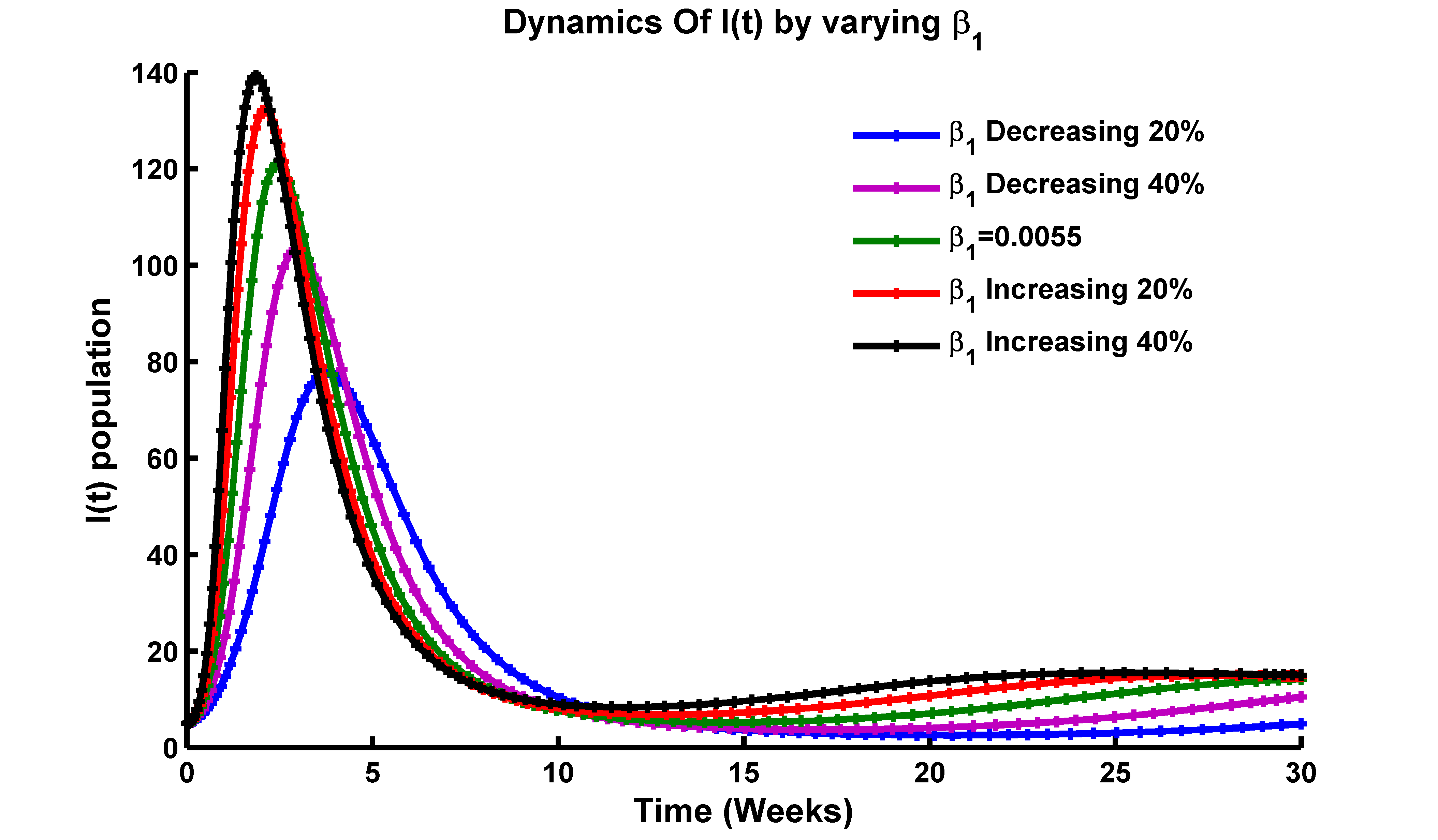}}
	\subfloat[]{\includegraphics[width=2. in]{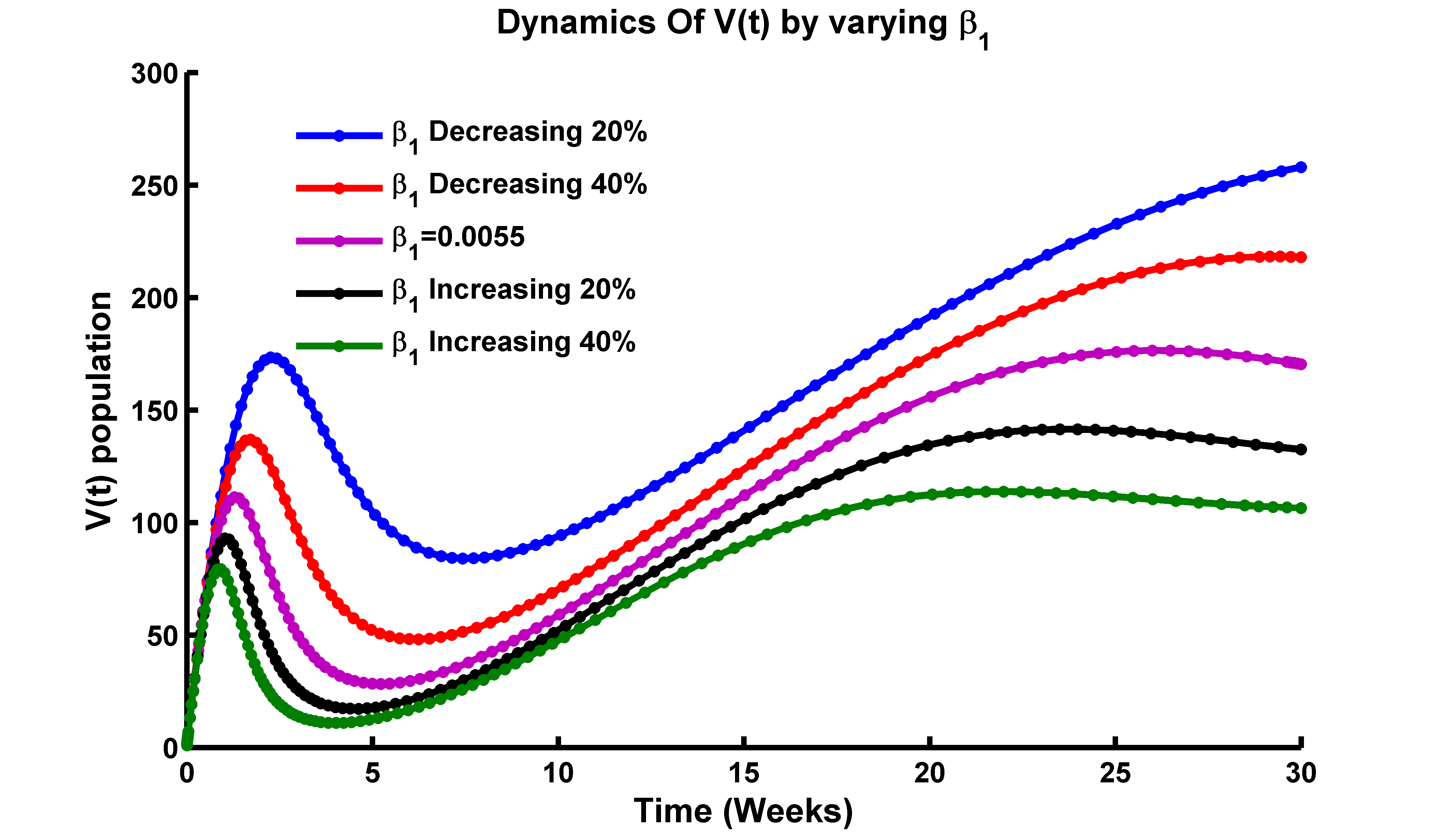}}
	\caption{Effect of $\beta_1$ in phase plane of (a) $S(t)$ compartment, (b) $E(t)$ compartment, (c) $E(t)+I(t)$ compartment, (d) $I(t)$ compartment and (e) $V(t)$ compartment where $\beta_1$ varying in range $[0.0035,0.0065]$ and rest of parameters are taken from Table \ref{tableparameter}.}\label{phase-plane-11-compartment}
\end{figure}
\noindent
Figure \ref{phase-plane-11-compartment}(a) illustrates that when the contact rate increases in the phase plane analysis, it affects the trajectory and behaviour of the susceptible compartment. Specifically, an increase in the contact rate leads to a steeper slope or gradient in the direction of the susceptible compartment axis. This steeper slope indicates that the rate at which individuals transition from the susceptible compartment to the infected compartment increases. In other words, as the contact rate rises, more susceptible individuals become infected at a faster pace. The trajectory of the susceptible compartment in the phase plane will show a more pronounced downward trend as time progresses, indicating a more rapid depletion of the susceptible population. This implies that the disease is spreading more rapidly within the population due to increased contact between infected and susceptible individuals.

Figure \ref{phase-plane-11-compartment}(b), (c) and (d) illustrates that, an increase in the contact rate can lead to a higher influx of individuals into the exposed compartment. The trajectory of the exposed compartment in the phase plane may exhibit a steeper slope, indicating a faster accumulation of individuals in the exposed state. If the contact rate is significantly higher than the recovery rate, the exposed compartment may grow rapidly, reaching a higher peak before declining. On the other hand, with an increased contact rate, the rate of individuals transitioning from the exposed compartment to the infected compartment will be higher. The trajectory of the infected compartment may show a steeper slope, indicating a faster rise in the number of infected individuals. If the contact rate surpasses the recovery rate, the infected compartment may continue to grow without reaching a peak, resulting in a sustained or increasing number of infected individuals.

In phase plane analysis, the trajectory of the vaccinated compartment is typically affected slightly by changes in the contact rate. It grows slowly over time, depending on the model assumptions and the rate of vaccination. Figure \ref{phase-plane-11-compartment}(e) illustrates that, the vaccinated compartment acts as a buffer or barrier, reducing the number of susceptible individuals who can be infected due to increased contact rates. When contact rate $\beta_{1}$ and $\beta_{2}$ rises, vaccination class reduces slightly. The presence of a large and growing vaccinated compartment can lead to a decrease in the overall disease transmission within the population, as vaccinated individuals are less likely to contract and spread the infection.

When the treatment rate increases, it means that infected individuals are progressed to treated class under control from the infection at a faster rate. This has an impact on the trajectory of the infected compartment in the phase plane. Specifically, the trajectory tends to shift towards lower values along the y-axis, indicating a decrease in the number of infected individuals over time. Figure \ref{phase-plane-12-compartment}(a) indicates that, a higher treatment rate results in a steeper slope of the trajectory, suggesting a more rapid decline in the number of infected individuals. This indicates that the duration of the infectious period is shorter, leading to a faster resolution of the infection within the population.

When the recovery rate increases, it means that infected individuals are recovering from the infection at a faster rate. This has an impact on the trajectory of the recovered compartment in the phase plane. Specifically, the trajectory tends to shift towards higher values along the y-axis, indicating an increase in the number of recovered individuals over time.
\begin{figure}[H]
	\centering  
	\subfloat[]{\includegraphics[width=2.5 in]{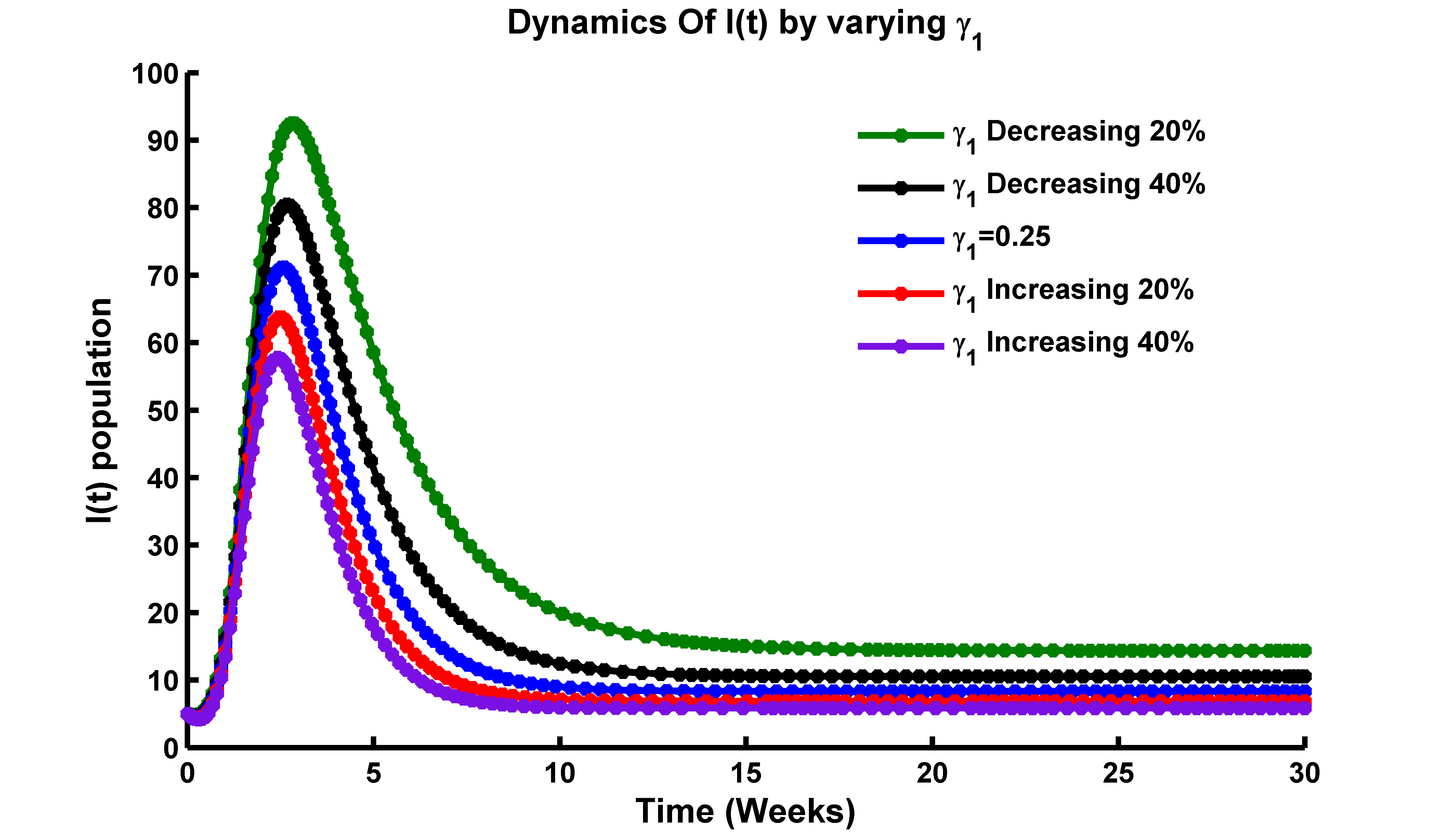}}
	\subfloat[]{\includegraphics[width=2.5 in]{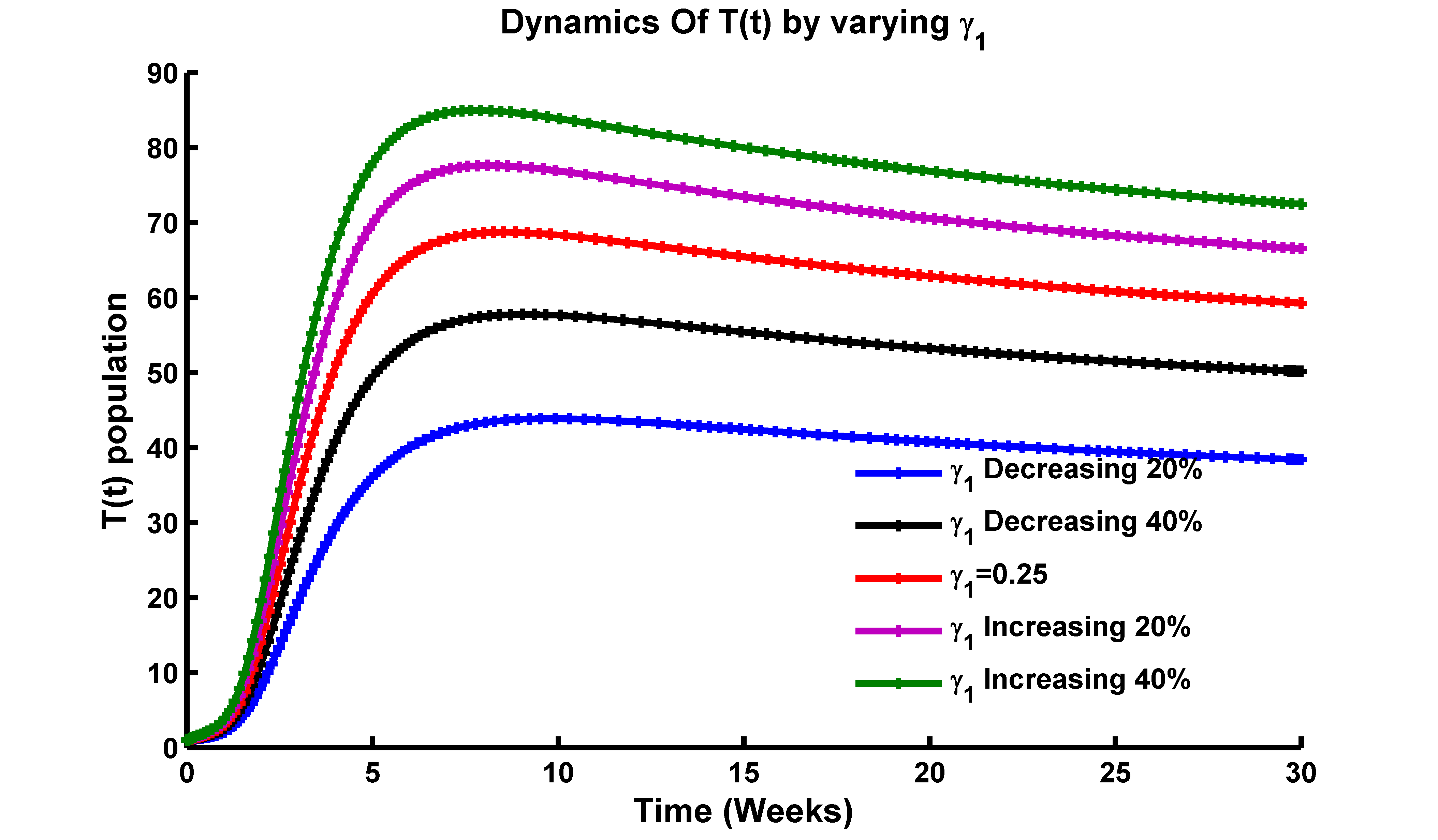}}\\
	\subfloat[]{\includegraphics[width=2.5 in]{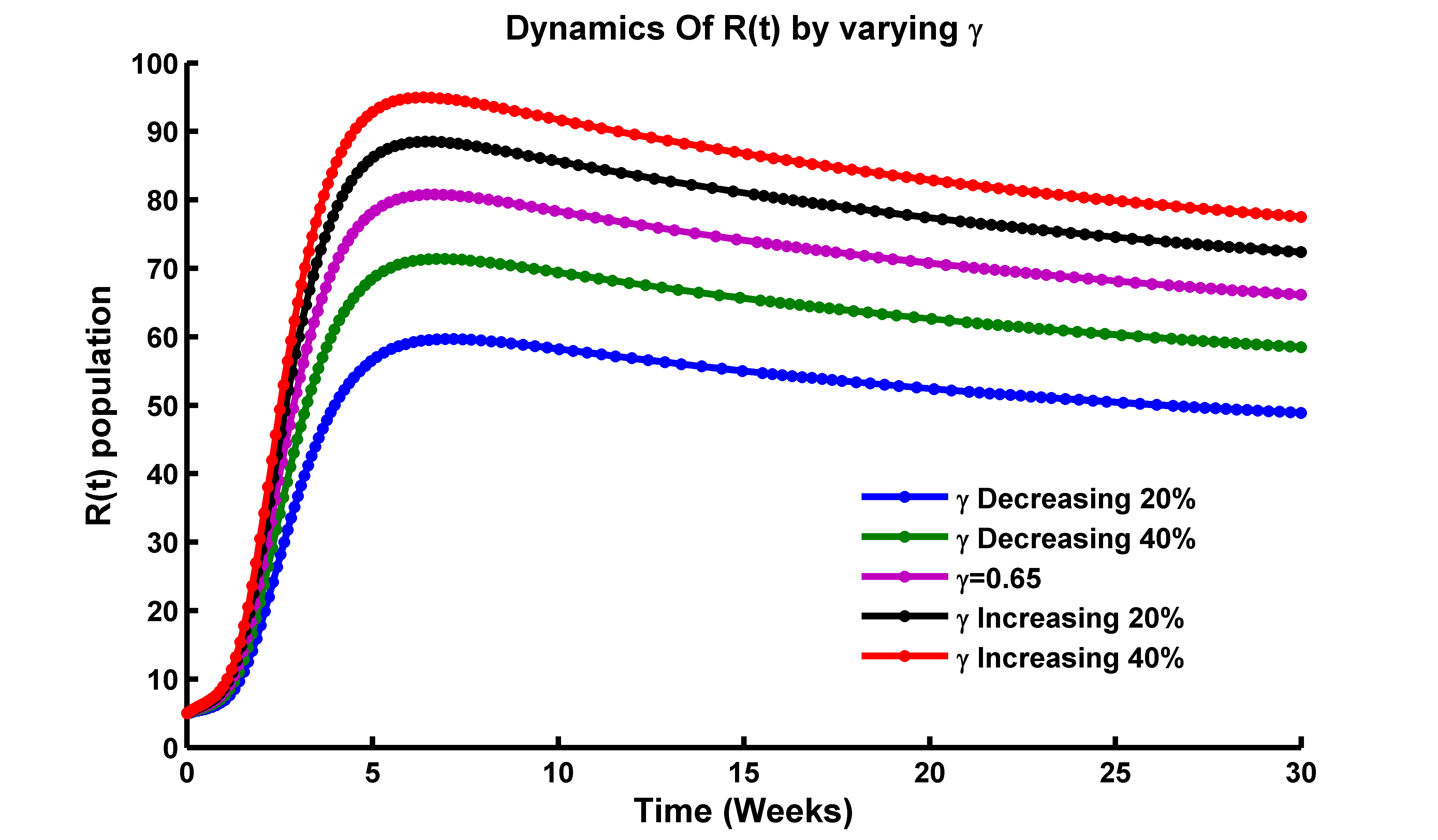}}
	\caption{Effect of $\gamma$ and $\gamma_1$ in the phase plane of (a) $I(t)$ compartment, (b) $T(t)$ compartment, (c) $R(t)$ compartment ,where $\gamma_1$ varying in range $[0.15,0.45]$, $\gamma$ varying in range $[0.35,0.75]$ and rest of parameters are taken from Table \ref{tableparameter}.}\label{phase-plane-12-compartment}
\end{figure}
\noindent
Figure \ref{phase-plane-12-compartment}(c) reflects that, a higher recovery rate results in a steeper slope of the trajectory, suggesting a more rapid increase in the number of recovered individuals. This indicates that individuals are transitioning out of the infected compartment and into the recovered compartment more quickly. An increase in the treatment rate may result in a shift of the treated compartment towards a stable state. This means that the individuals or entities being treated experience positive changes and move towards a more desired condition or outcome. Figure \ref{phase-plane-12-compartment}(b) indicates that, when the treatment rate increases, individuals or entities in the treated compartment may experience faster recovery or improvement. This can be visualized in the phase plane analysis by observing trajectories that move more rapidly towards healthier states or regions associated with improved outcomes.

\subsection{Phase Plane Analysis of Two Compartments}
The number or proportion of individuals in the exposed class influences the rate at which new infections occur. From Figure \ref{phase-plane-21-compartment}(a) we see that, as the exposed individuals become infectious, they transition into the infected class, contributing to the overall number of infected individuals. The interaction and flow of individuals between the exposed and infected classes determine the dynamics of the epidemic, such as the rate of transmission, the speed of disease spread, and the eventual size of the infected population. The more exposed population grows, the faster grows of the infected population. This relationship is crucial for developing effective strategies to control and mitigate the impact of the epidemic. Figure \ref{phase-plane-21-compartment}(b) indicates that, as individuals recover, they transition from the infected class to the recovered class, reducing the number of active infections. The size and dynamics of the recovered class can impact the spread of the disease. A larger number of recovered individuals means a smaller pool of susceptible individuals for the disease to infect, potentially slowing down transmission rates. The rate at which individuals recover and transition to the recovered class affects the overall duration and severity of the epidemic. Faster recovery rates can lead to a quicker decline in the number of active infections. Figure \ref{phase-plane-21-compartment}(c) indicates that, the size and effectiveness of the treatment class can impact the progression of the epidemic. Prompt and effective treatment can help reduce the duration of infectiousness and potentially lower the transmission rates. The interaction between the treatment and infected classes influences the overall burden of the disease on the healthcare system and the potential for reducing morbidity and mortality. The more population progress to the treated class from infected, it reduces the disease burden. Figure \ref{phase-plane-21-compartment}(d) indicates that, vaccination reduces the susceptibility of individuals to the infectious disease, thereby decreasing the likelihood of them transitioning from the susceptible class to the infected class. The higher the vaccination coverage within a population, the lower the number of individuals in the susceptible class, leading to a reduced pool of potential hosts for the disease. As the number of vaccinated individuals increases, the infected class may experience a decline in its size, resulting in a decline in the overall disease transmission rate. Thus, vaccination also plays a crucial role in reducing the transmission of the disease from infected individuals to susceptible ones, further limiting the spread of the epidemic.

Figure \ref{phase-plane-22-compartment}(a) reflects that, initially, the infected class increases rapidly, indicating a high rate of new infections, while the cumulative infected class starts from zero. As time progresses, the infected class may reach a peak and start to decline, while the cumulative infected class continues to increase as new infections add to the total count. The shape and pattern of the trajectory in the phase plane can reveal important information about the dynamics of the epidemic, such as the effectiveness of control measures or the presence of multiple waves.
\begin{figure}[H]
	\centering  
	\subfloat[]{\includegraphics[width=2.5 in]{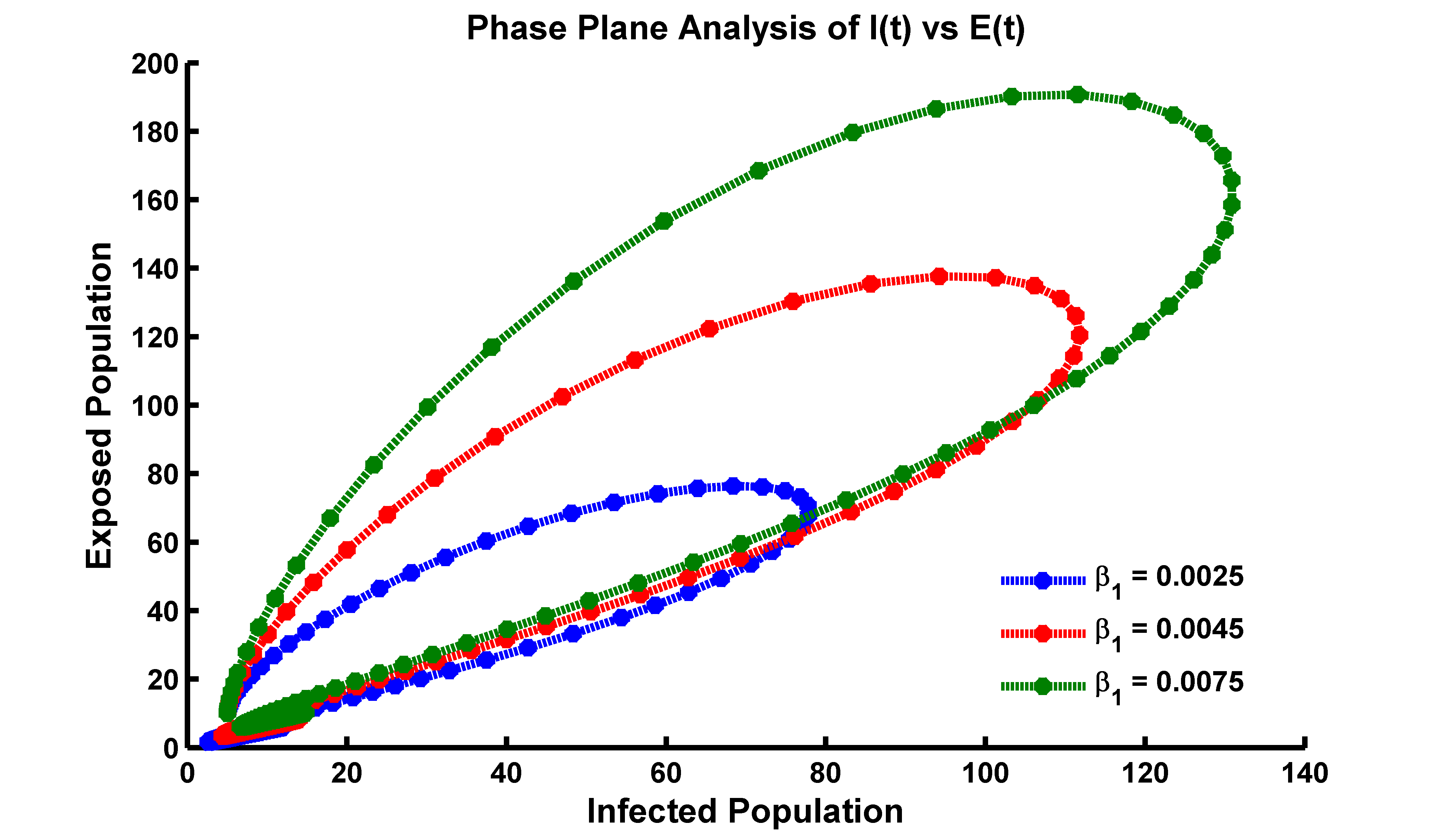}}
	\subfloat[]{\includegraphics[width=2.5 in]{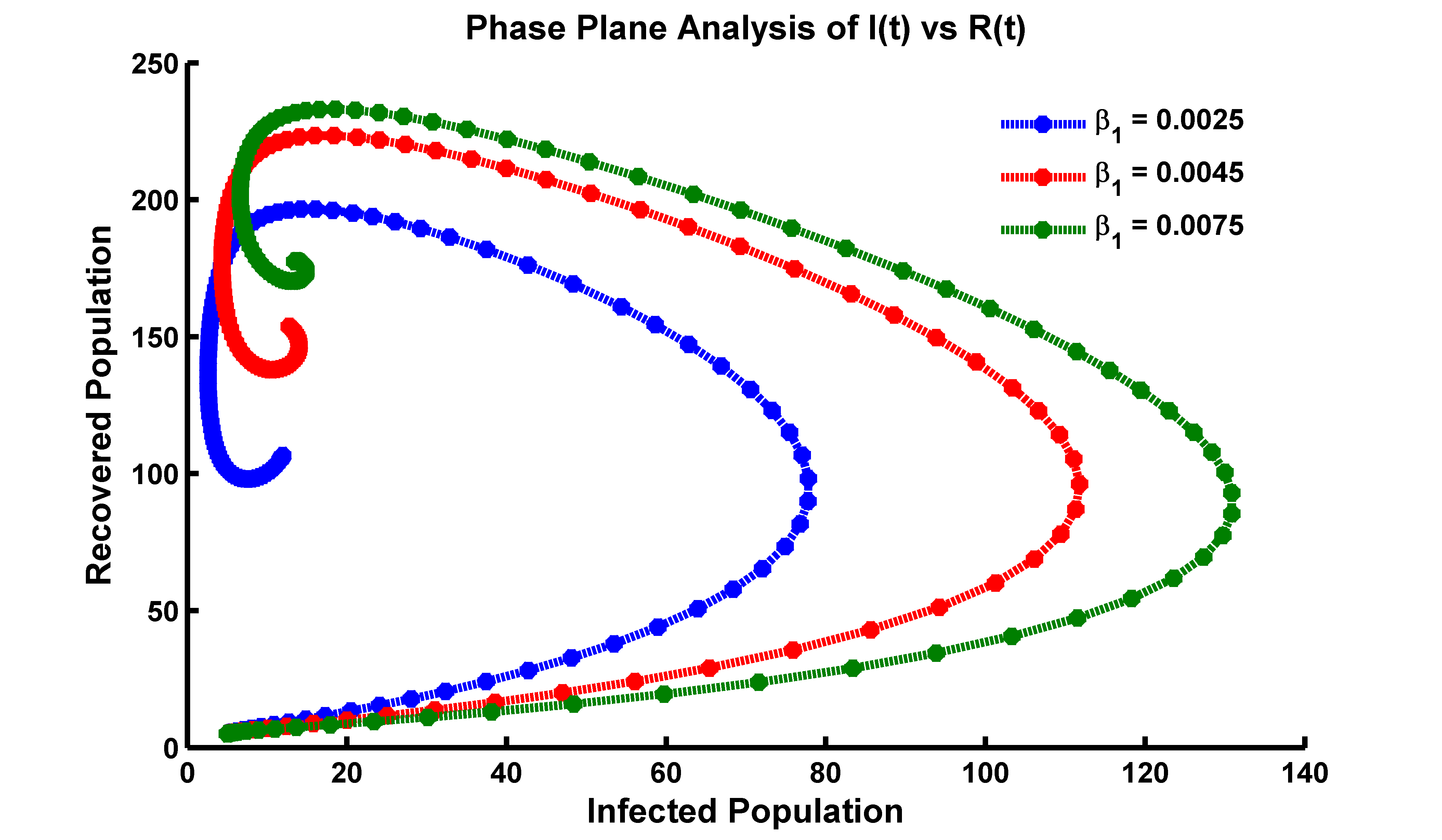}}\\
	\subfloat[]{\includegraphics[width=2.5 in]{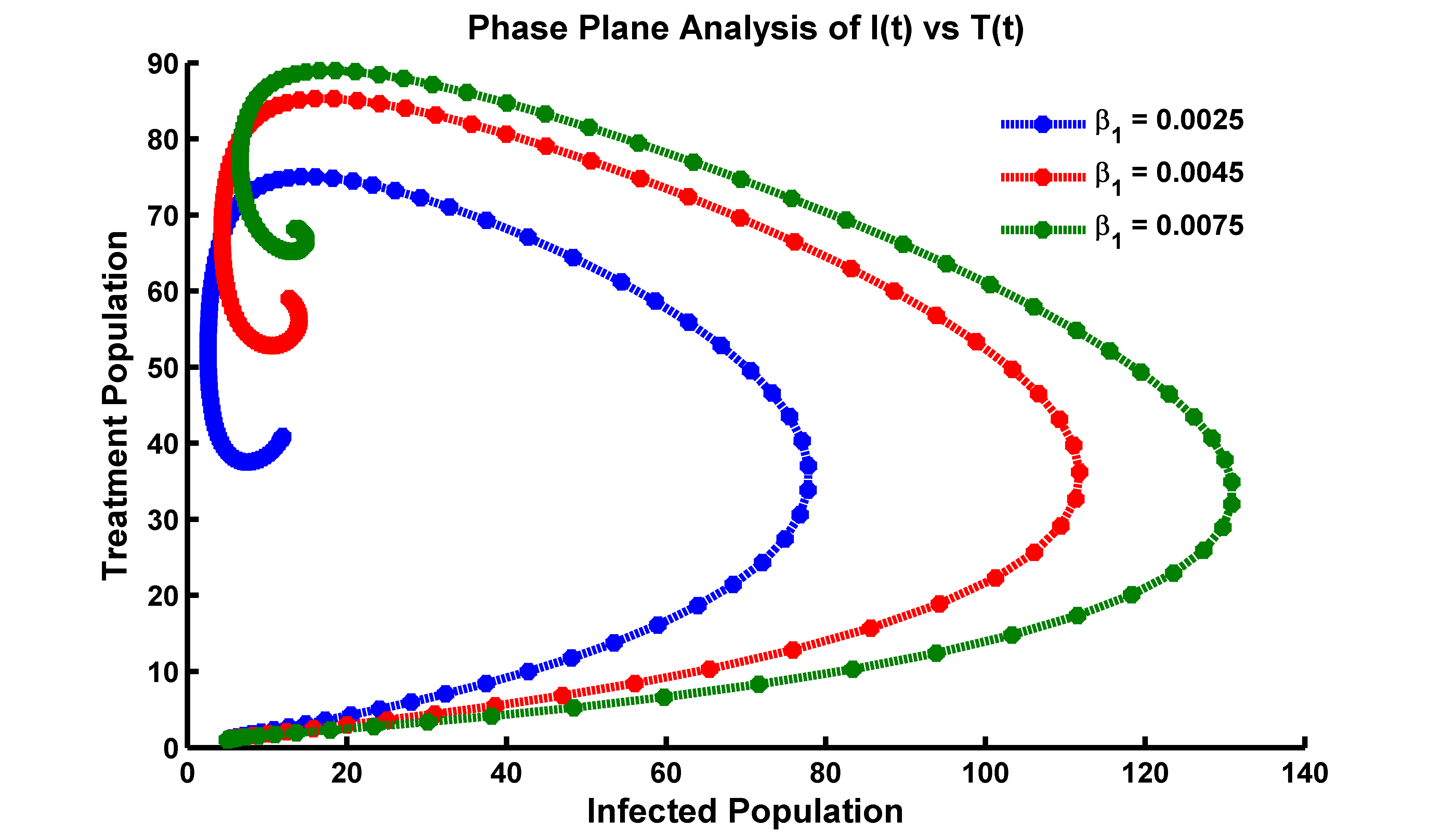}}
	\subfloat[]{\includegraphics[width=2.5 in]{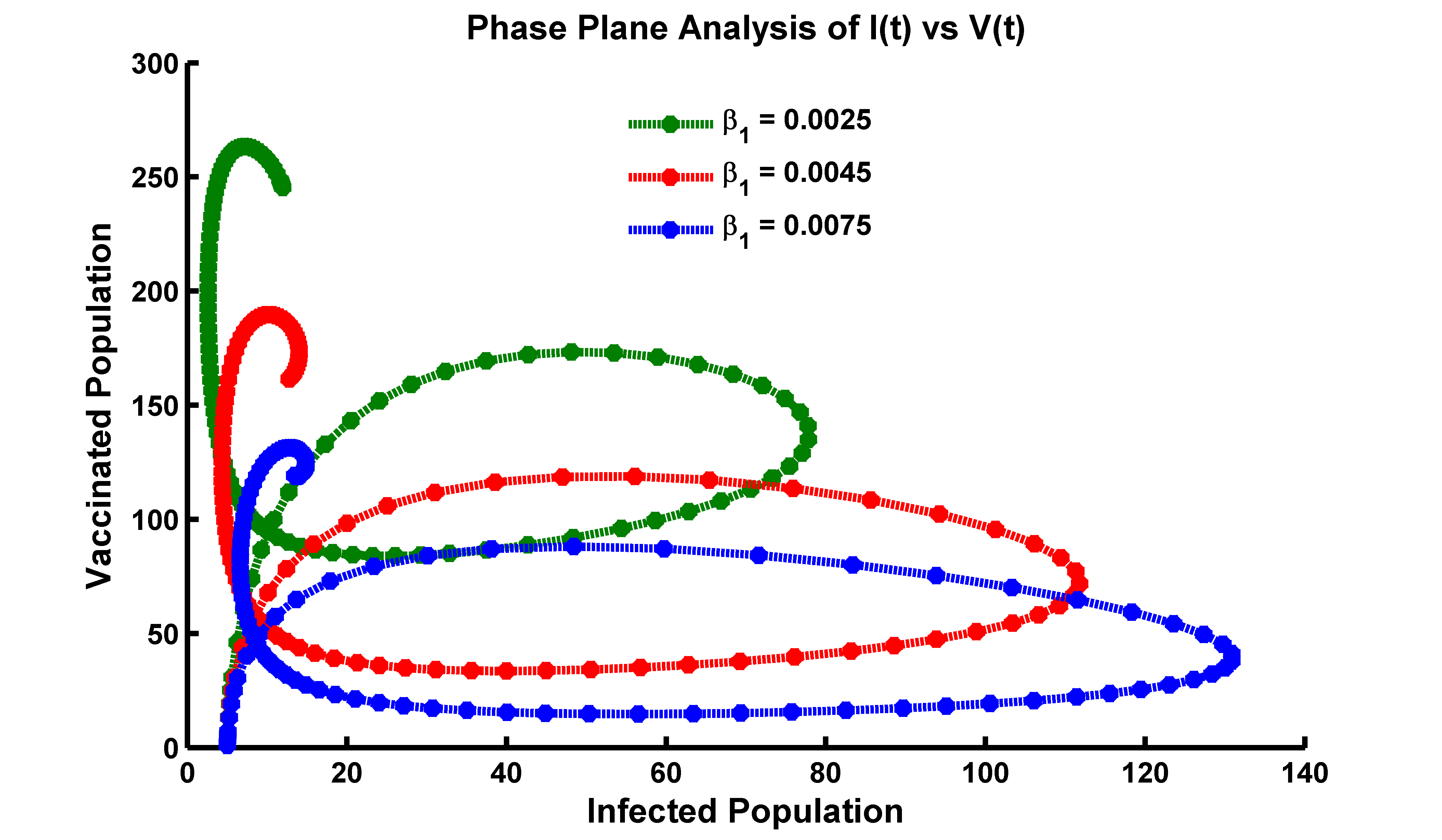}}
	\caption{Phase plane of (a) $I(t)$ vs $E(t)$ compartment,(b) $I(t)$ vs $R(t)$ compartment,(c) $I(t)$ vs $T(t)$ compartment and (d) $I(t)$ vs $V(t)$ compartment where all parameters are taken from Table \ref{tableparameter}.}\label{phase-plane-21-compartment}
\end{figure}
\noindent
From Figure \ref{phase-plane-22-compartment}(b) we see that, the model \eqref{new_model} equations describe how the number of susceptible individuals changes over time as they become exposed to a particular infectious agent. The phase plane plot allows us to visualize the dynamics of this relationship by plotting the susceptible population on one axis and the exposed population on the other. Trajectories in the phase plane represent the flow of individuals between the susceptible and exposed states, providing insights into the progression of an infectious disease. The more susceptible individuals reduces, it progress to the exposed class and exposed population increase gradually. From Figure \ref{phase-plane-22-compartment}(c), these trajectories represent the movement of the system over time, considering the two variables: susceptible population and infected population. The direction and shape of the trajectories reveal the dynamics of the epidemic. Typically, when the susceptible population is high and the infected population is low, the trajectories move towards the susceptible axis. As the infected population increases, the more susceptible population reduces from community, the trajectories shift towards the infected axis, indicating the spread of the disease. The intersection or equilibrium point between the two axes represents the steady state, where the epidemic reaches a balance between susceptible and infected individuals. 
\begin{figure}[H]
	\centering
	\subfloat[]{\includegraphics[width=2.5 in]{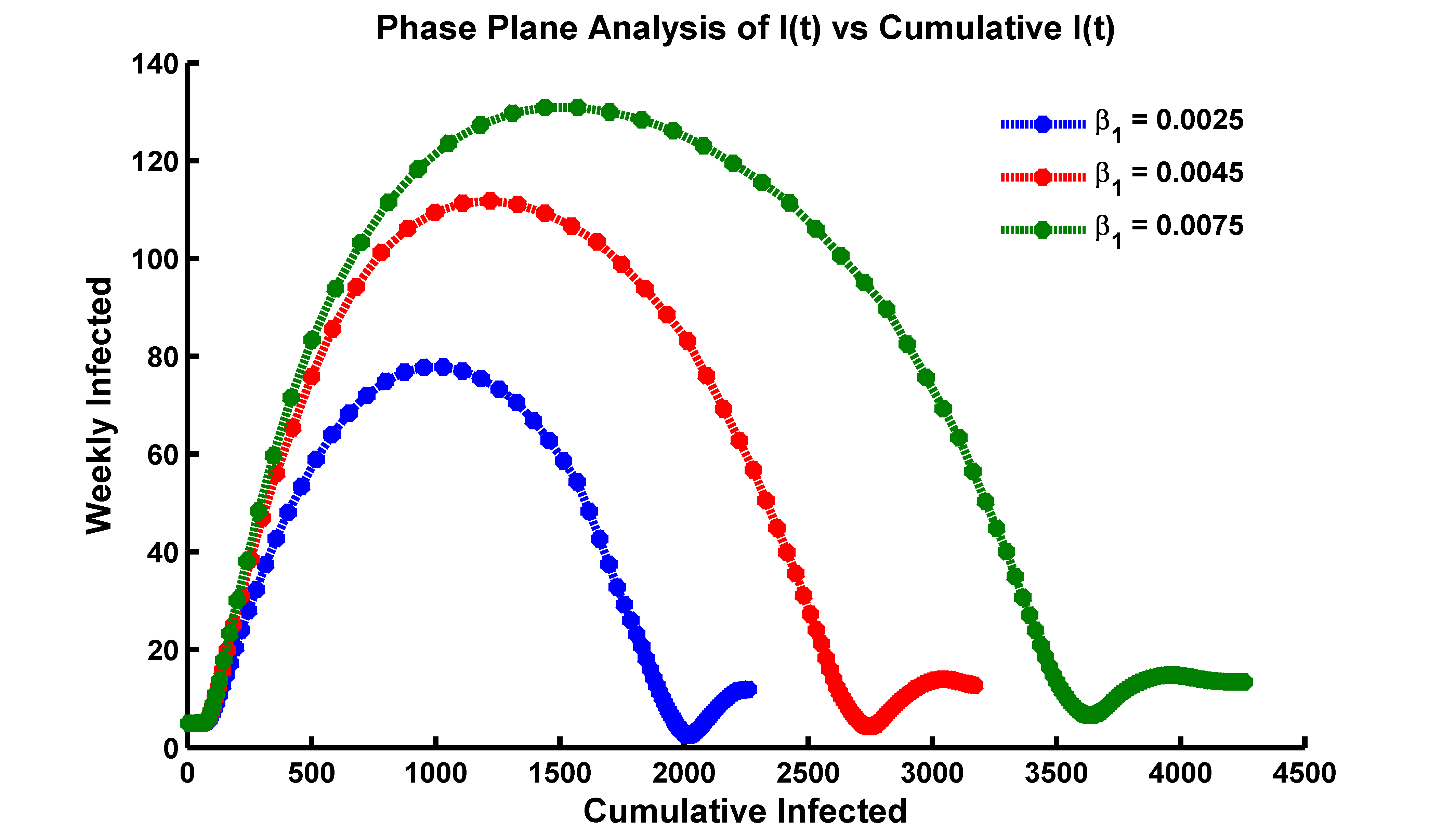}}
	\subfloat[]{\includegraphics[width=2.5 in]{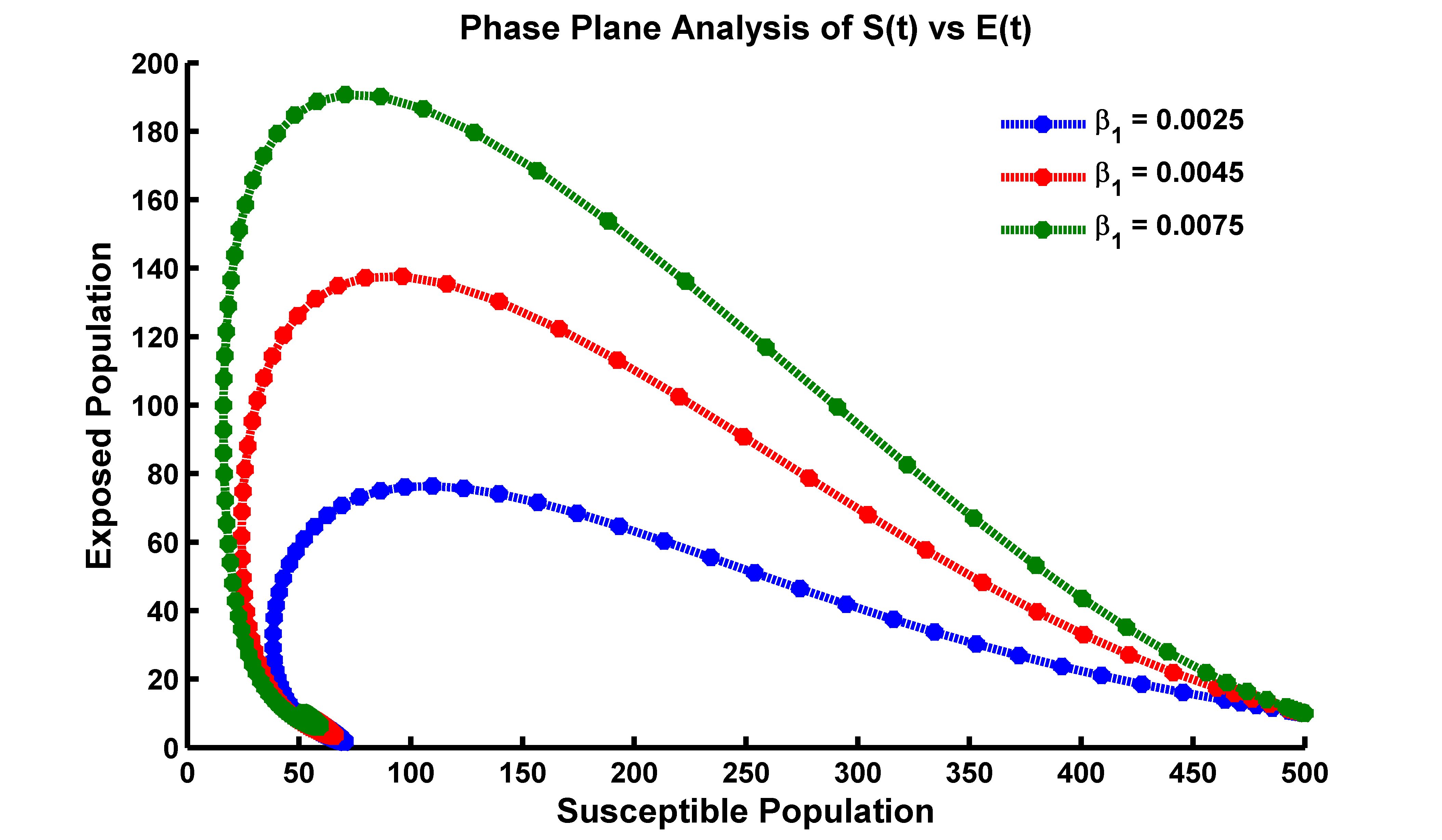}}\\
	\subfloat[]{\includegraphics[width=2.5 in]{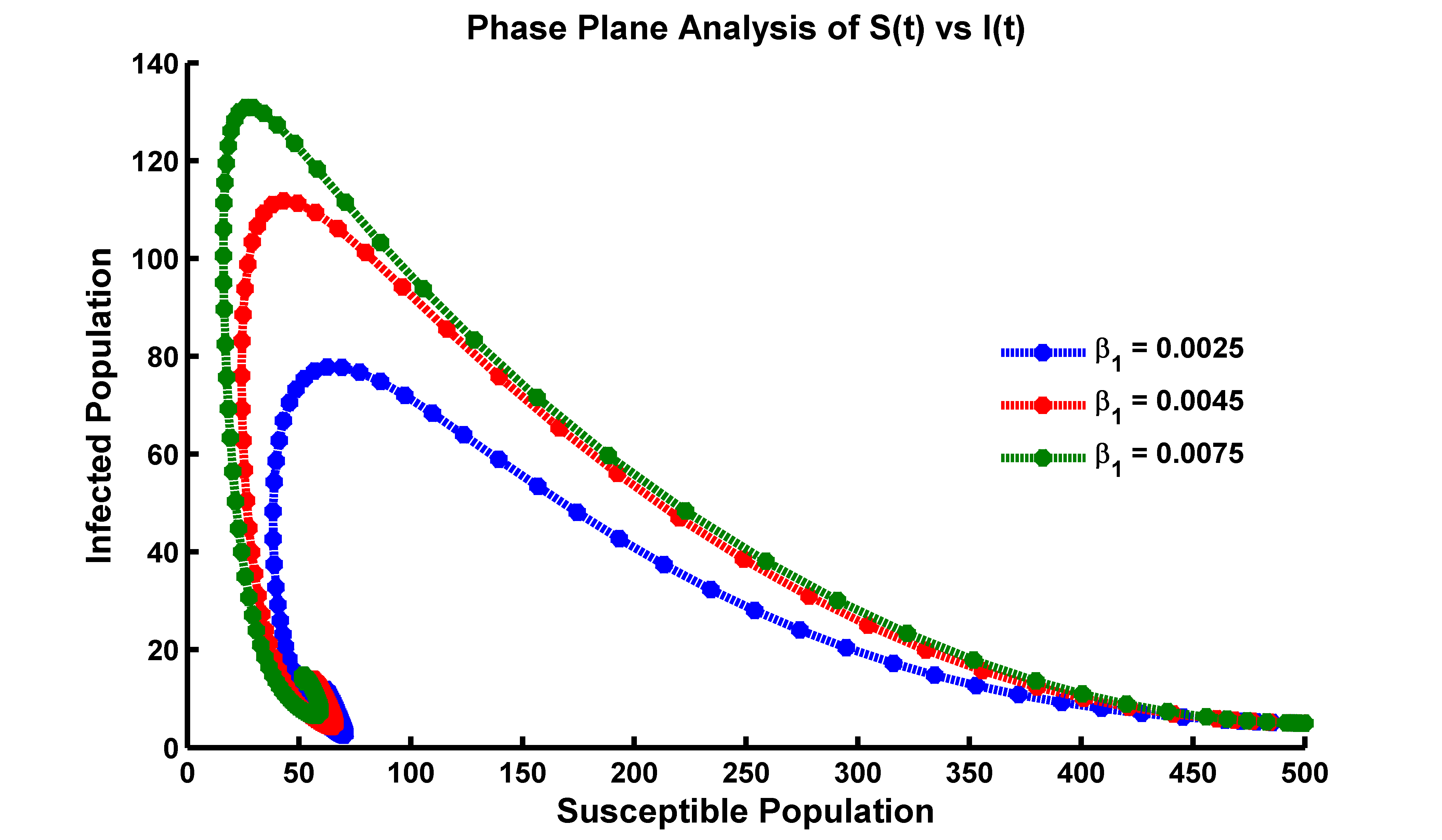}}
	\subfloat[]{\includegraphics[width=2.5 in]{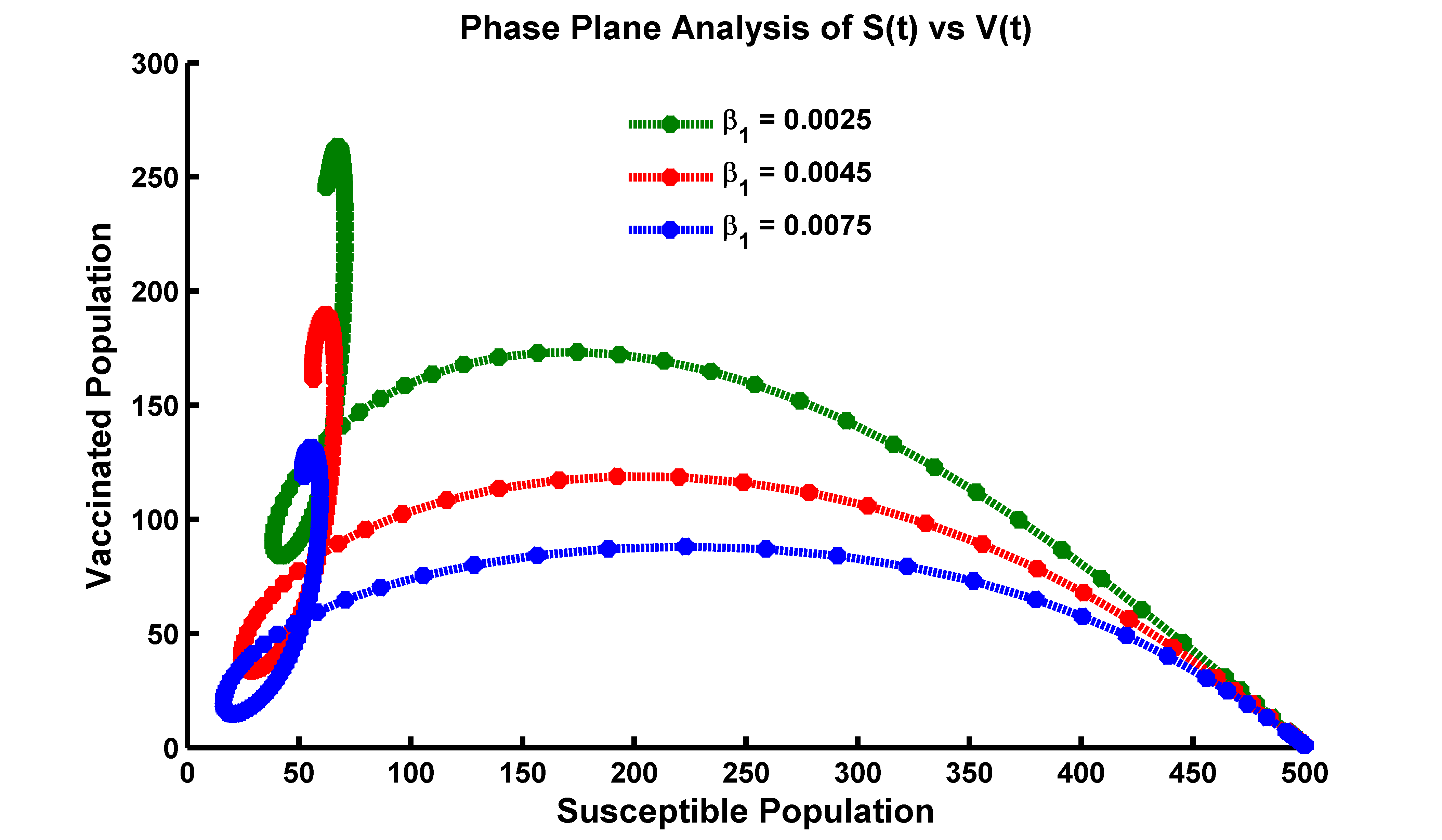}}
	\caption{Phase plane of (a) $I(t)$ vs Cumulative $I(t)$ compartment (b) $S(t)$ vs $E(t)$  compartment (c) $S(t)$ vs $I(t)$ compartment and (d) $S(t)$ vs $V(t)$ compartment where all parameters are taken from Table \ref{tableparameter}.}
	\label{phase-plane-22-compartment}
\end{figure}
\noindent
From Figure \ref{phase-plane-22-compartment}(d), it reflects that, these trajectories demonstrate the dynamics of the system over time, considering the two variables: susceptible population and vaccinated population. When the susceptible population is high and the vaccinated population is low, the trajectories move towards the susceptible axis. As the vaccinated population increases with the decreasing of susceptible population, the trajectories shift towards the vaccinated axis, indicating the impact of vaccination on reducing susceptibility. That means, by vaccination, more portion of people progress to the control strategy. The intersection or equilibrium point between the two axes represents the steady state, where a balance is achieved between susceptible and vaccinated individuals. The analysis helps assess the effectiveness of vaccination in mitigating the spread of the disease.

\subsection{Contour Plot Analysis of $\mathcal{R}_0$}\label{Subsection-Contour_Plot}
A contour plot of $\mathcal{R}_0$ (basic reproduction number) with respect to two parameters in epidemiology provides valuable insights into the spread and control of infectious diseases. By creating a contour plot, we can visualize how changes in two specific parameters affect the value of $\mathcal{R}_0$. This helps us in this thesis work to understand the dynamics of disease transmission and make informed decisions regarding interventions and control measures \cite{Bifurcation of R0-3, Bifurcation of R0-7}. The contour plot allows us to identify regions where $\mathcal{R}_0$ remains low or high based on the parameter values. It helps identify critical thresholds or tipping points that may lead to an outbreak or epidemic. For example, if the contour plot shows a steep increase in $\mathcal{R}_0$ as parameter values cross a certain threshold, it suggests that particular factors significantly influence disease transmission and warrant attention. Furthermore, the contour plot can guide decision-making by identifying areas where interventions or modifications to the parameters could effectively reduce $\mathcal{R}_0$. By manipulating the parameters within favorable regions, public health measures can be tailored to control or prevent the spread of infectious diseases more efficiently \cite{Bifurcation of R0-8}. Moreover, contour plot serves as a powerful tool to understand the relationship between key factors, assess the potential for disease spread, and devise targeted strategies for disease control and prevention.\\
Additionally, The colour bar in a contour plot indicates the values associated with different colours in the plot. It provides a visual representation of the magnitude or level of the variable being displayed. This typically represents the values of $\mathcal{R}_0$. Each colour on the colour-bar corresponds to a specific range or interval of $\mathcal{R}_0$ values. The colour intensity or shading within the contour plot indicates the relative magnitude of $\mathcal{R}_0$ at different points on the plot.

In this section, we have presented the relationship if the parameters by contour plot of the basic reproduction number $\mathcal{R}_0$ as a function of parameters. 
Figure \ref{contourplot-analysis}(a) depicts the simulation of the model \eqref{new_model} by showing the contour plot of the threshold quantity $\mathcal{R}_0$ as a function of disease transmission rate $\beta_{1}$ from $E(t)$ compartment, and recovery rate $\gamma$. In general it shows that, $\mathcal{R}_0$ value decreases when $\beta_{1}$ increases, and $\mathcal{R}_0$ value increases when $\beta_{1}$ decreases. Also, $\mathcal{R}_0$ value decreases significantly, as $\gamma$ increases. But the rate of progression of $\mathcal{R}_0$ is faster than that of $\gamma$.
\begin{figure}[H]
	\centering  
	\subfloat[]{\includegraphics[width=2.5 in]{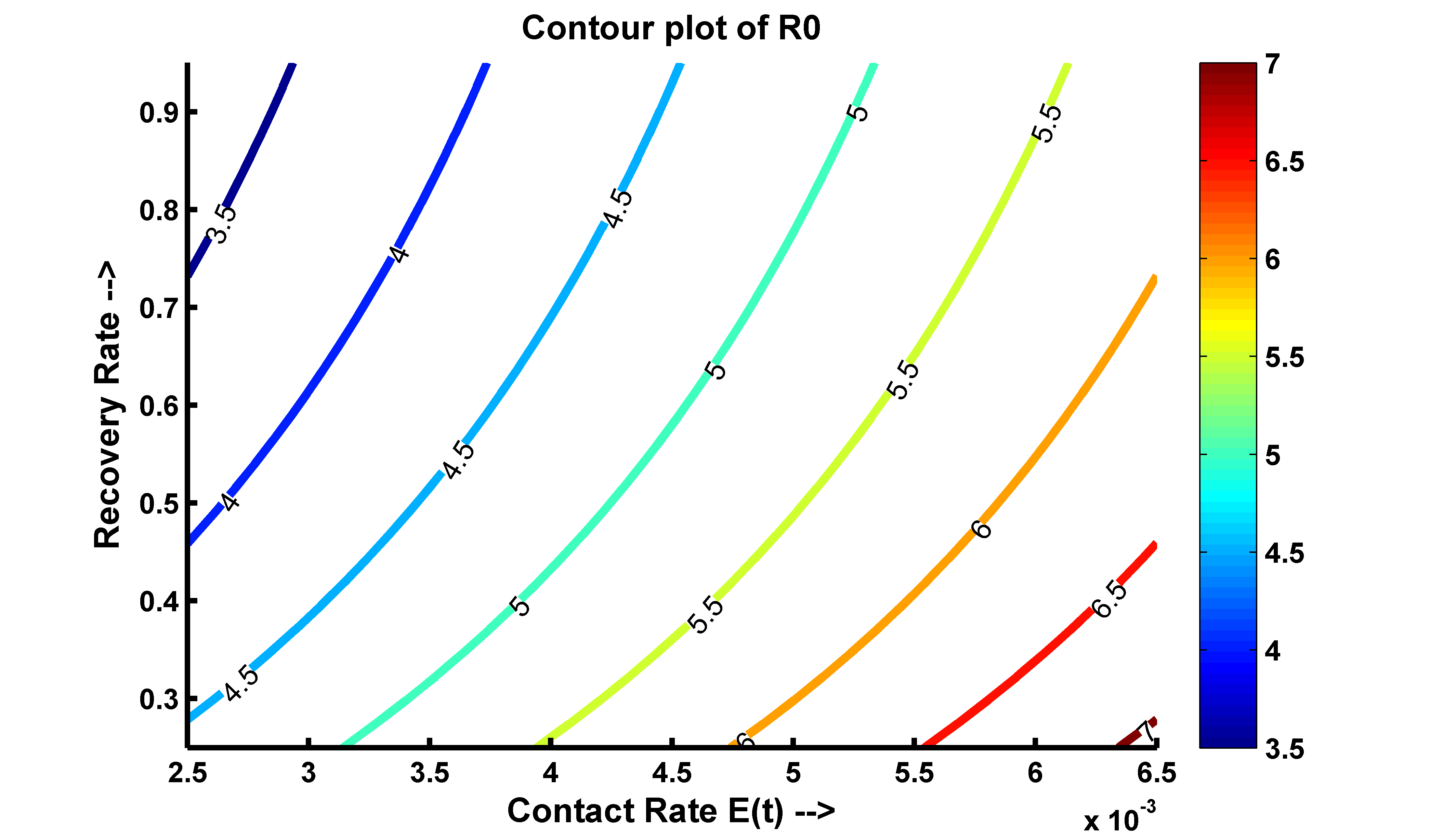}}
	\subfloat[]{\includegraphics[width=2.5 in]{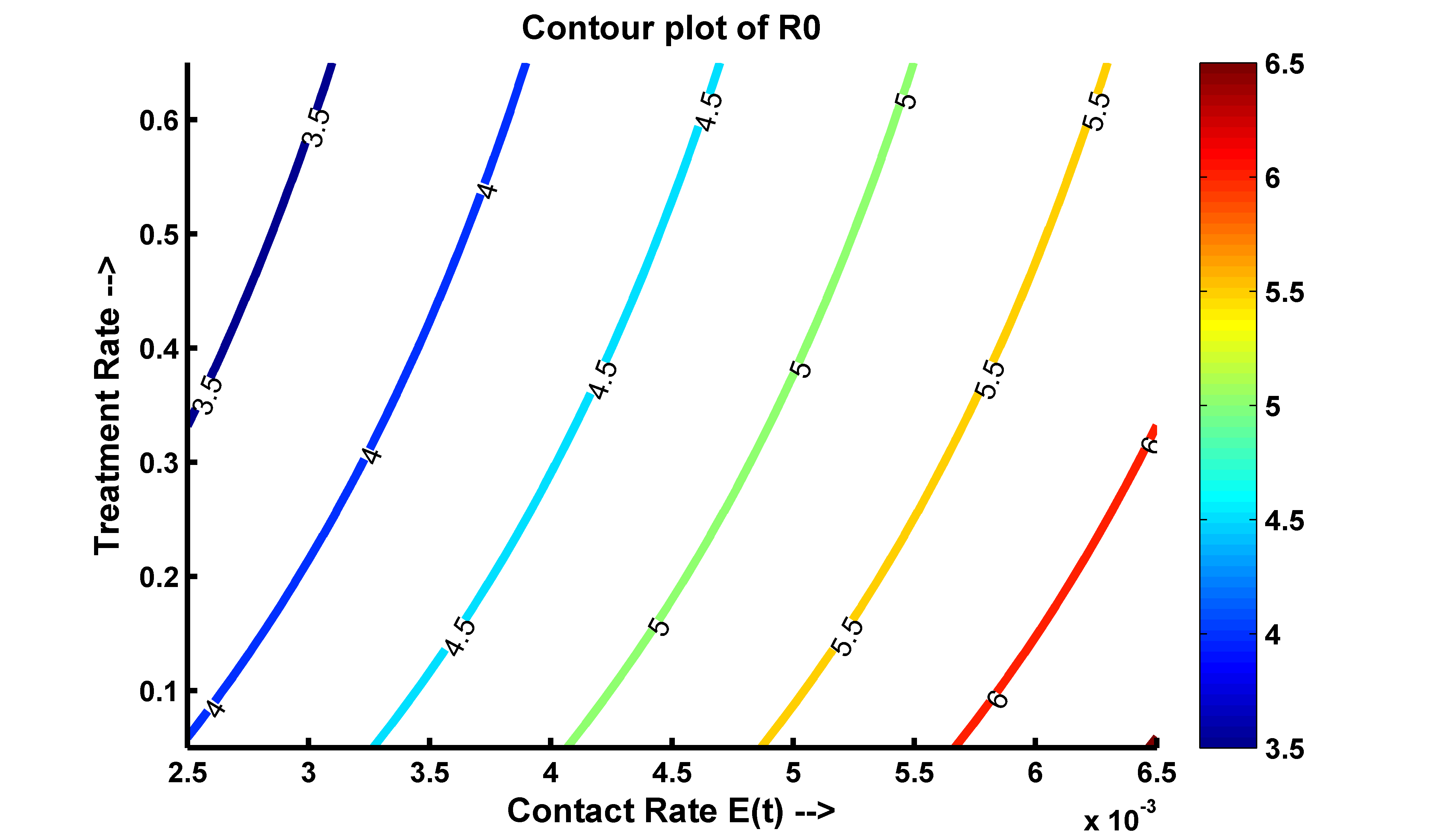}}\\
	\subfloat[]{\includegraphics[width=2.5 in]{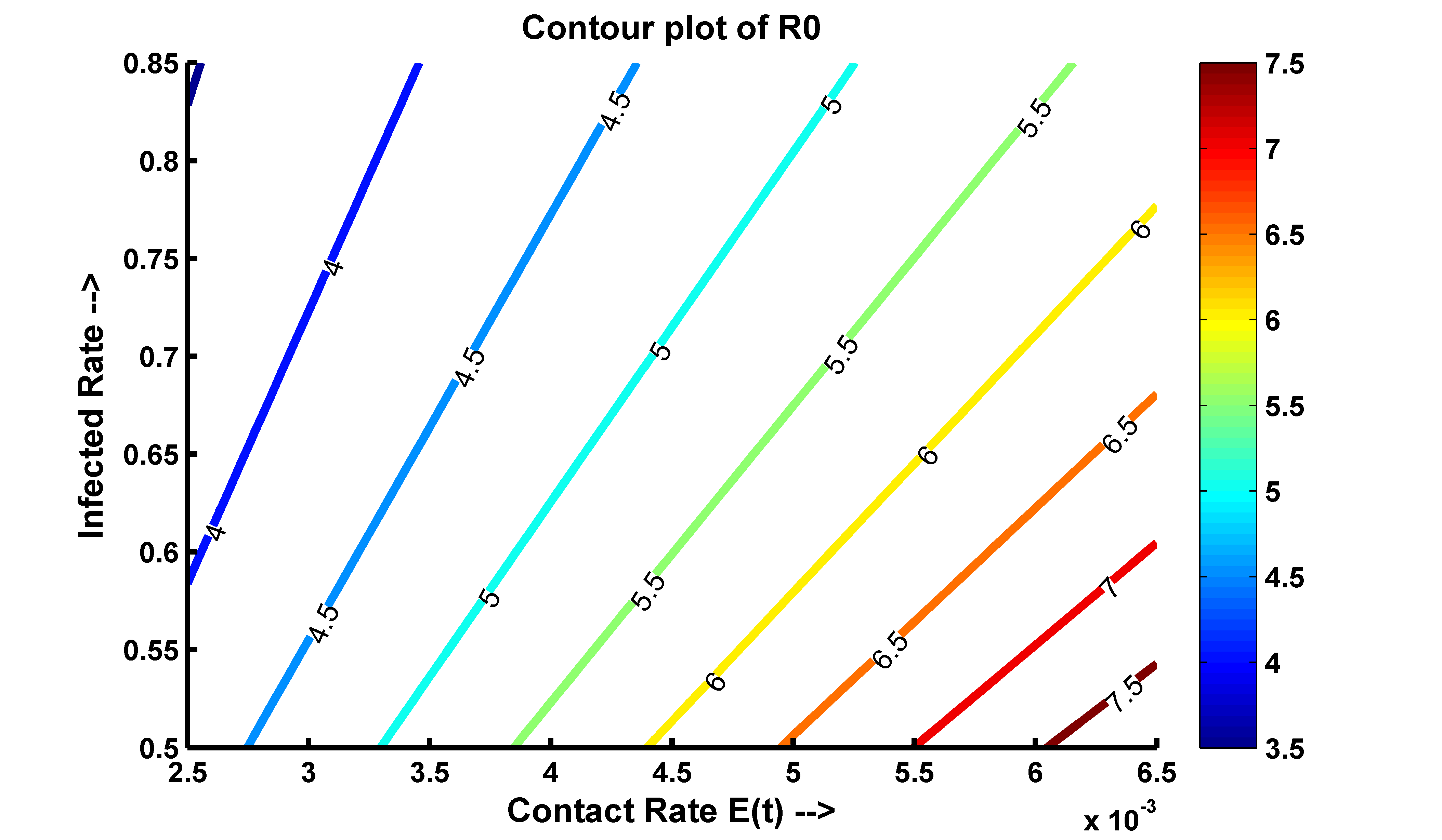}}
	\subfloat[]{\includegraphics[width=2.5 in]{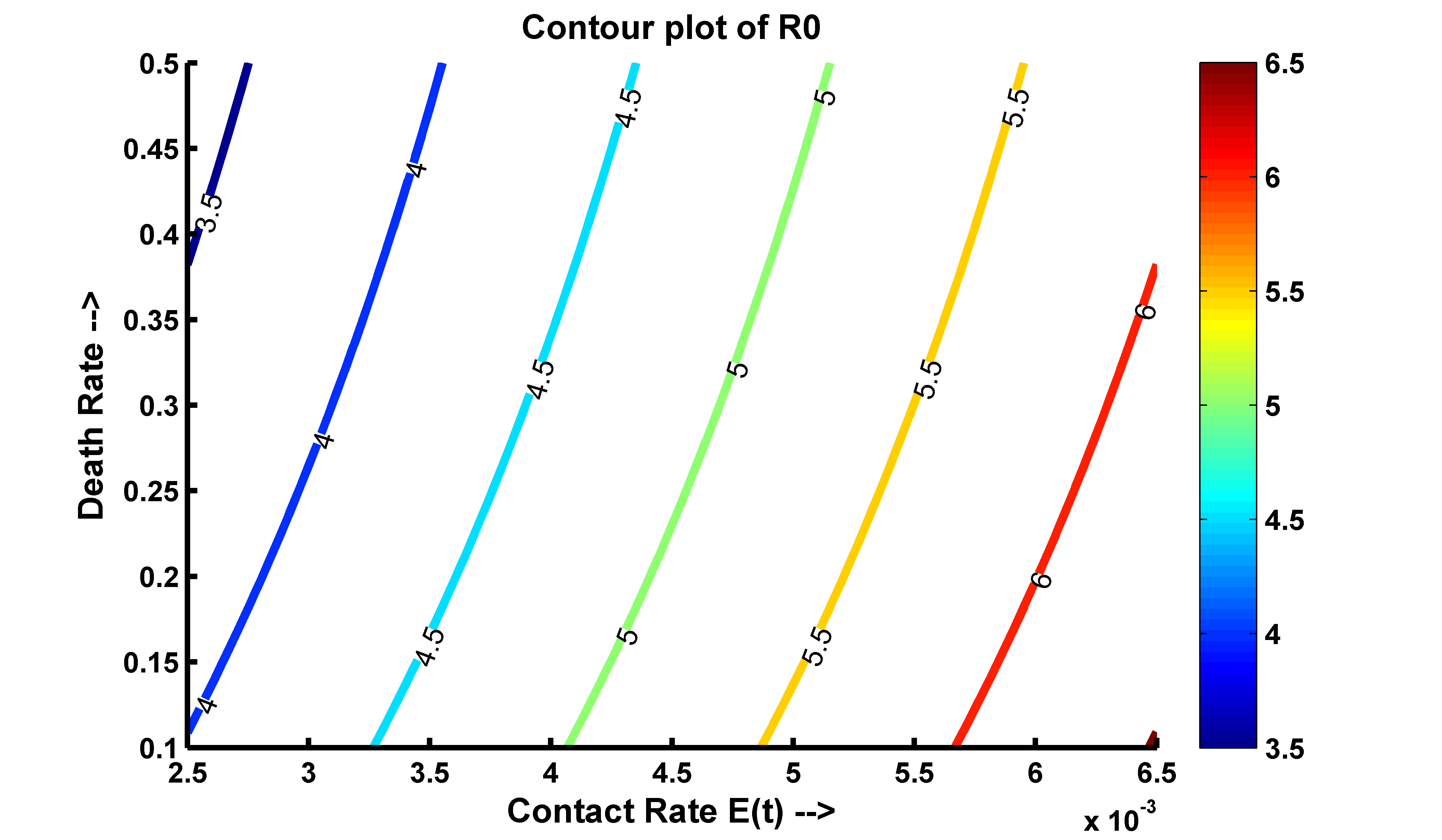}}\\
	\subfloat[]{\includegraphics[width=2.5 in]{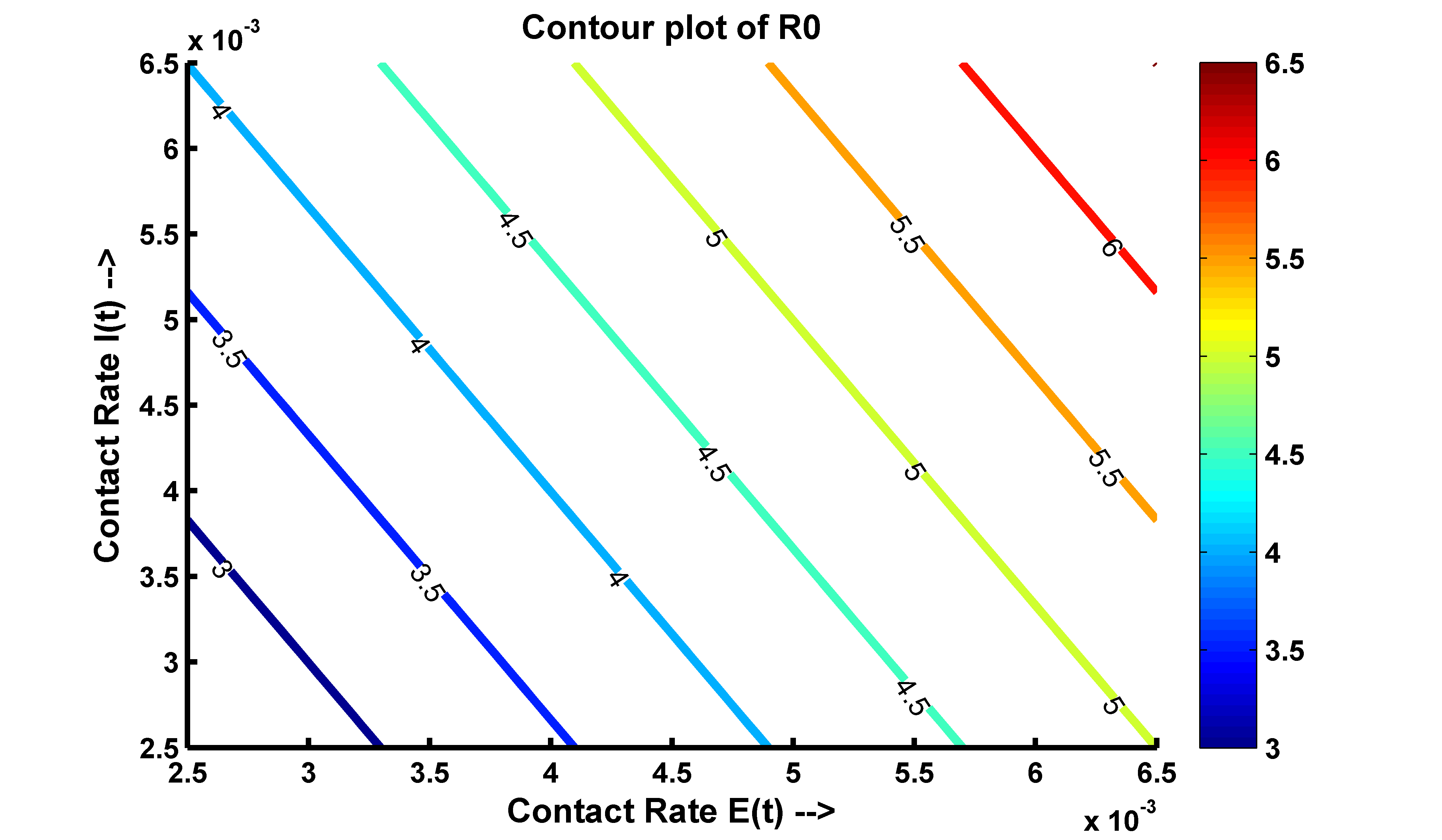}}
	\caption{Contour plots of $\mathcal{R}_0$ as a function of (a) parameter $\beta_1$ vs $\gamma$ (b) parameter $\beta_1$ vs $\gamma_1$ (c) parameter $\beta_1$ vs $\alpha$ (d) parameter $\beta_1$ vs $\delta$ and (e) parameter $\beta_1$ vs $\beta_2.$, parameter values are taken from Table \ref{tableparameter}.}\label{contourplot-analysis}
\end{figure}
\noindent
Simulation of the model \eqref{new_model} is presented in Figure \ref{contourplot-analysis}(b) showing the contour plots of $\mathcal{R}_0$ as function of $\beta_{1}$ and the treatment rate $\gamma_1$. It indicates that the value of $\mathcal{R}_0$ increases (or decreases) gradually when the value of $\gamma_1$ increases (or decreases). That means if we can take treatment strategy, we can mitigate the disease burden. Thus, one prevention strategies can be vaccination and proper isolation of the infected individuals.

From the numerical simulation of the model presented in Figure \ref{contourplot-analysis}(c) showing contour plot of $\mathcal{R}_0$ as a function of parameter $\beta_1$ and infection rate $\alpha$. It shows that $\mathcal{R}_0$ increases (or decreases) rapidly when $\alpha$ increases (or decreases). Thus by vaccination and treatment strategies, if progression to the infected compartment can be reduced, the disease burden can be minimized. Figure \ref{contourplot-analysis}(d) showing contour plot of $\mathcal{R}_0$ as a function of parameter $\beta_1$ and disease induced death rate $\delta$ and natural death rate $\mu$. That indicates that, $\mathcal{R}_0$ has reverse relation with the progression rate of $\delta$ and $\mu$. Thus, separation of the infected and exposed individuals has a significant impact to extinction of the disease from the community. Further, Figure \ref{contourplot-analysis}(e) showing contour plot of $\mathcal{R}_0$ as a function of parameter $\beta_1$ and parameter $\beta_2$. In general, it shows that the threshold quantity increases (or decreases) quickly if $\beta_{1}$ and $\beta_{2}$ increases (or decreases). But the increasing (or decreasing) rate of $\mathcal{R}_0$ with respect to $\beta_{2}$ is faster than that of $\beta_{1}$.  Thus, $\beta_{1}$ and $\beta_{2}$ are the crucial parameters which have significant impact to persist the disease in the community. Epidemiological meaning of the simulation is that, the increasing rate of $\gamma$ and $\gamma_1$ and decreasing rate of $\beta_{1}$, $\beta_{2}$ and $\alpha$ can reduce the spread of disease and mitigate the burden from the community.

\subsection{Box Plot Analysis}\label{Subsection-BoxPlot}
A box plot is a type of graph used to visualize the distribution of a dataset, particularly its median, quartiles, and outliers. In the context of basic reproduction number, a box plot analysis can be used to understand the variability of the basic reproduction number ($\mathcal{R}_0$) across different groups or time periods. When analyzing $\mathcal{R}_0$ using box plots, the median represents the typical value of $\mathcal{R}_0$, while the box indicates the interquartile range (IQR) of $\mathcal{R}_0$ values. The whiskers of the box plot represent the range of $\mathcal{R}_0$ values, while any points beyond the whiskers are considered outliers. In this section we have carried out the box plot analysis of $\mathcal{R}_0$ as a function of two parameters \cite{Bifurcation of R0-3, Bifurcation of R0-7, Bifurcation of R0-8}.

Figure \ref{Box-plot-analysis}(a) reflects valuable insights into the relationship between parameters $\beta_{1}$,$\beta_{2}$ (contact rate) and $\gamma$ (recovery rate) with the threshold quantity $\mathcal{R}_0$. The box visually displays the distribution of $\mathcal{R}_0$ at different combinations of $\beta_1$,$\beta_2$ and $\gamma$. We observe that, the tendency of $\mathcal{R}_0$ is decreasing (or increasing) for the value of parameter $\gamma$ increasing (or decreasing). But progression of $\mathcal{R}_0$ is proportional with $\beta_1$ and $\beta_2$. From, the whiskers extending from the box we have observed interquartile range of $\mathcal{R}_0$ is 5 to 7 when median value is 6 (for increasing 30\% of $\gamma$). Meanwhile, by increasing $\gamma$ by 60\% and 90\%, the median value of $\mathcal{R}_0$ decrease to 5 and 4.5 respectively; also the interquartile range.
\begin{figure}[H]
	\centering  
	\subfloat[]{\includegraphics[width=2.5 in]{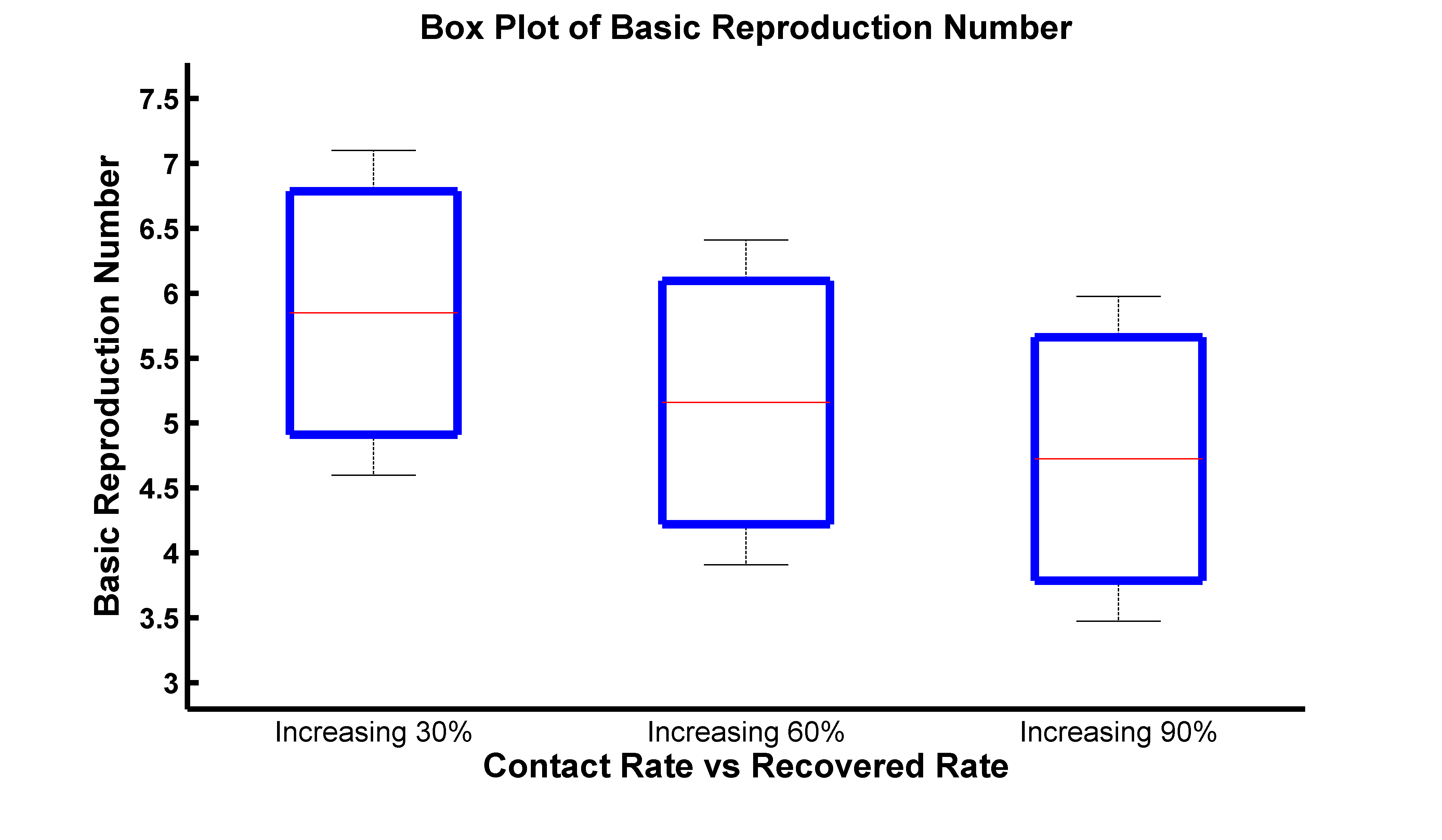}}
	\subfloat[]{\includegraphics[width=2.5 in]{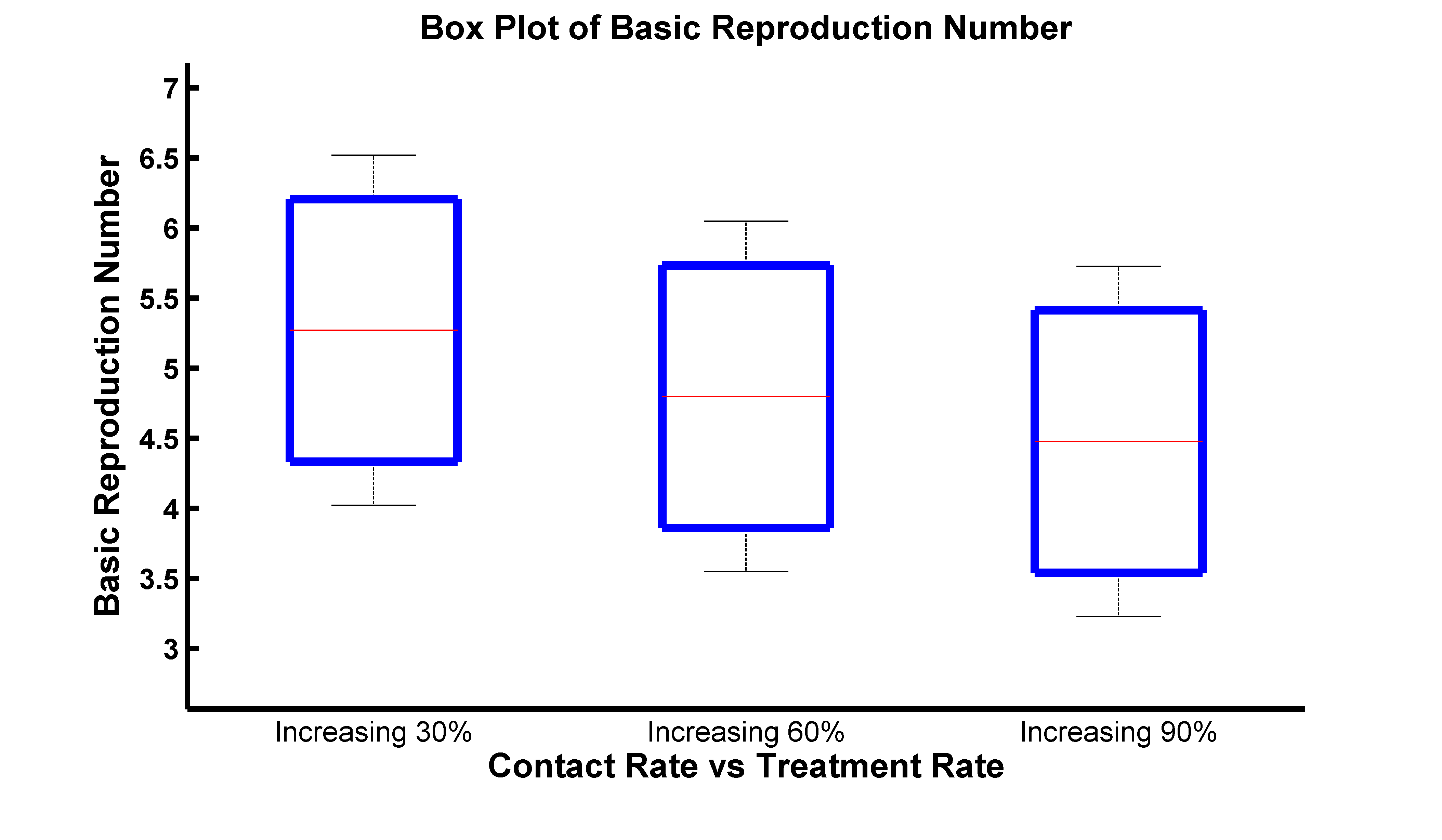}}\\
	\subfloat[]{\includegraphics[width=2.5 in]{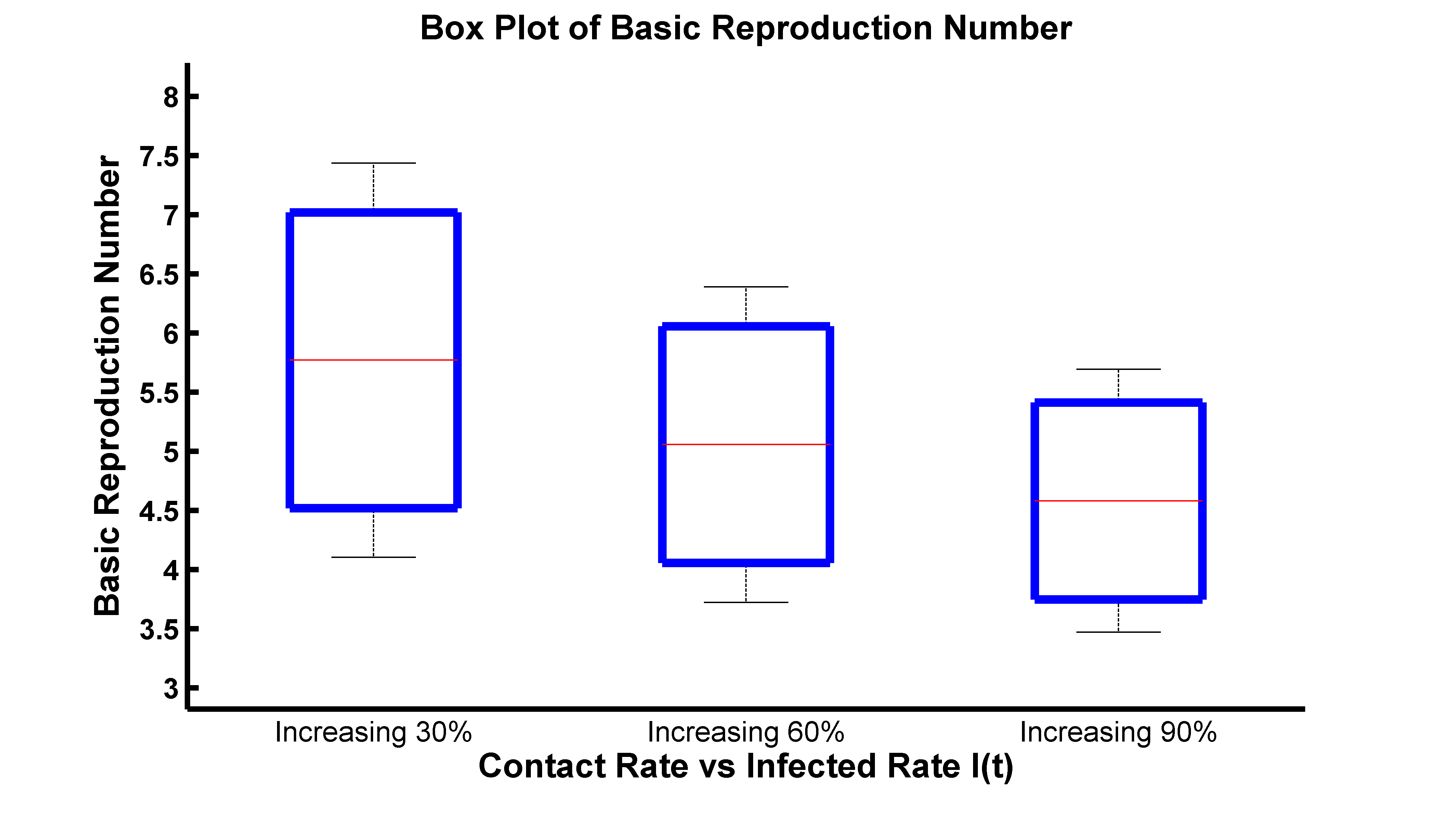}}
	\subfloat[]{\includegraphics[width=2.5 in]{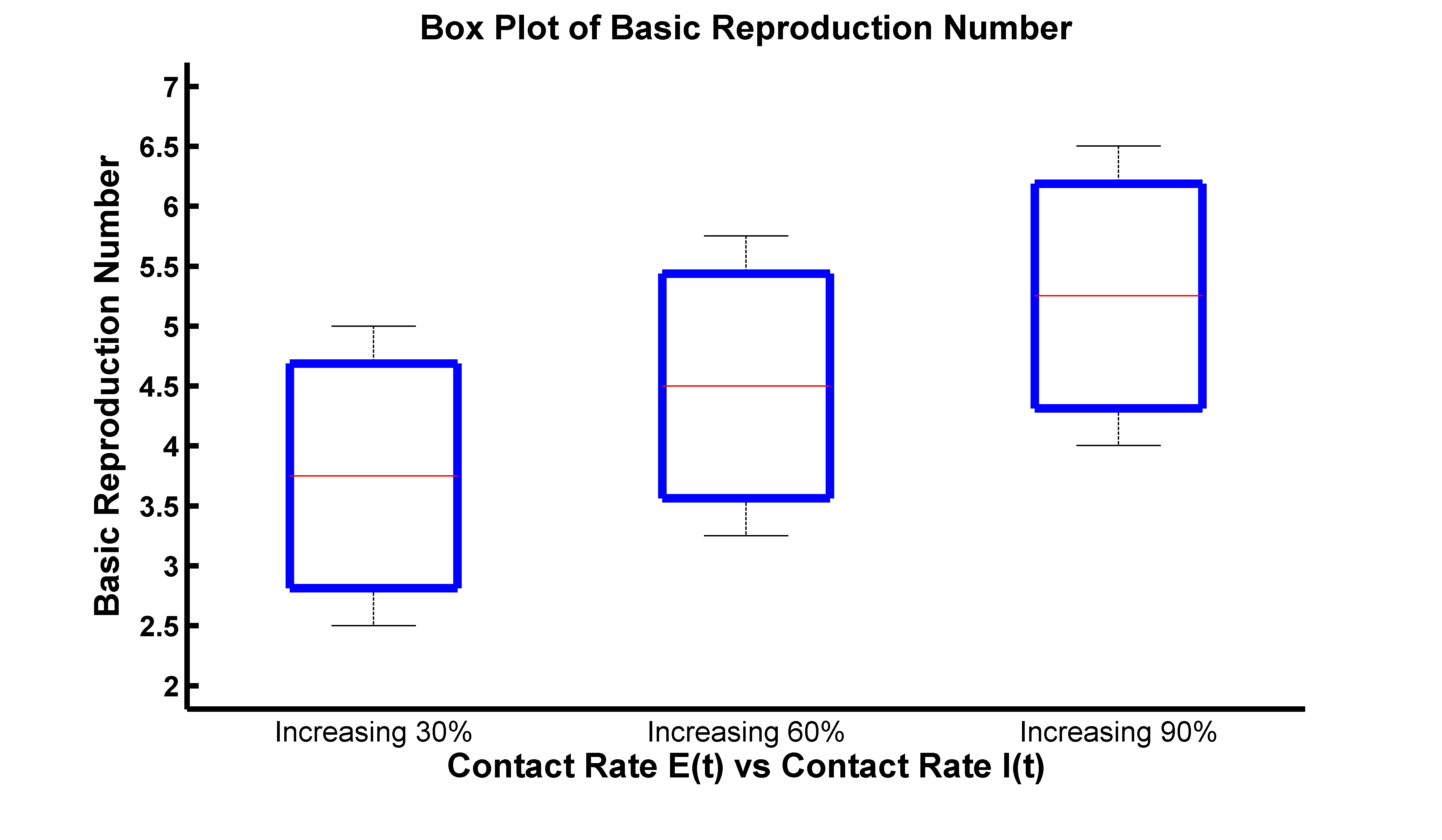}}
	\caption{Box plot analysis of $\mathcal{R}_0$ with the (a) parameter $\beta_1$ vs $\gamma$ (b) $\beta_1$ vs $\gamma_1$ (c) $\beta$ vs $\alpha$ and (d) $\beta_1$ vs $\beta_2$ , where all the parameter values are taken from Table \ref{tableparameter}.}\label{Box-plot-analysis}
\end{figure}
\noindent
Figure \ref{Box-plot-analysis}(b) reflects the influence of contact rates $\beta_1$, $\beta_2$ and treatment rate $\gamma_1$ on basic reproduction number $\mathcal{R}_0$. It reflects that, by increasing rate of $\gamma_1$, reduces the median value of $\mathcal{R}_0$. When, $\gamma_1$ increases 60\%, the median value of $\mathcal{R}_0$ is 4.7. From, the whiskers extending from the box we have observed interquartile range of $\mathcal{R}_0$ is 3.75 to 5.8.

Box plot analysis of $\mathcal{R}_0$ as a function of parameters $\beta_2$ and infection rate $\alpha$ is presented in Figure \ref{Box-plot-analysis}(c). This indicates that, with the progression of $\alpha$ and $\beta_2$, positive impact on the interquartile range of $\mathcal{R}_0$. When $\alpha$ and $\beta_2$ both increase 60\%, the median of $\mathcal{R}_0$ is 5, along with the interquartile range 4 to 6. But with the increasing of $\alpha$ and $\beta_{2}$ both 30\%, the interquartile range of $\mathcal{R}_0$ increases in range $[4.5, 7]$ and median is 5.75. Moreover, Figure \ref{Box-plot-analysis}(d) reveals box plot of $\mathcal{R}_0$ as a function of parameter $\beta_1$ and $\beta_2$. This indicates that, $\beta_1$ and $\beta_2$ have strong positive influence on the growth of $\mathcal{R}_0$. When both these contact rates on crease 30\%, then median value of $\mathcal{R}_0$ is 3.75, while the interquartile range is $[2.75,4.6]$. After the progression of $\beta_1$ and $\beta_2$ at 60\%, median value of $\mathcal{R}_0$ increase to 4.5 and whiskers reflects the increasing of interquartile range in $[3.5,5.5]$. After that, when $\beta_1$ and $\beta_2$ increases up-to 90\%, the median of $\mathcal{R}_0$ is 5.2 and the interquartile range fall in $[4.32, 6]$. This analysis suggests that, disease is more likely to spread rapidly with the increasing amount to transmission rates $\beta_1$, $\beta_2$ and infection rate $\alpha$.

Medicine of Influenza virus is available in several countries. Also experimental vaccine trials are going through in Africa sun continent regions \cite{CDC, WHO}. The objective is to determine the impact of treatment rate $\gamma_1$ on the threshold quantity $\mathcal{R}_0$. By taking the parameters value from Table \ref{tableparameter}, box plot analysis shows that the treatment to exposed and infected individuals reduces the quantity $\mathcal{R}_0$ effectively. If the infected individuals at sub-acute phase are considered as infectious, Figure \ref{Box-plot-analysis} reveals that the treatment and vaccination control strategy has a significant impact to reduce the disease burden. 

\section{Conclusion}\label{Section-Concluding-Remarks}
We have considered a deterministic SVEIRT epidemic model for influenza infection by taking vaccination and treatment strategies.
Theoretically, it is justified that the model is well-posed. 
The DFE is also globally asymptotically stable, which is established by Lyapunov functions and the LaSalle Invariance Principle whenever $\mathcal{R}_0<1$. When the basic reproduction number, $\mathcal{R}_0$ is greater than one, there exists a unique EE point of the model. Local stability of the EE is shown by the Routh-Hurwitz process when $\mathcal{R}_0>1$. Also, the global stability of EE is proved by using the non-linear Lyapunov function and LaSalle Invariance Principle, whenever $\mathcal{R}_0>1$. 
	Moreover, numerical simulations of the model support the existence and stability of DFE and EE.
 From the simulations of the model, we see that if the transmission rates $\beta_1$, $\beta_2$, and rate of acquiring infection $\alpha$, increases that results in the rapid increase of the disease burden. We see that if the rate of acquiring infection increases, then the disease burden will increase. Also, we observed the growth rate in $E(t)$ compartment arises the disease burden. Meanwhile, treatment rate and vaccination rate mitigate the factor $\mathcal{R}_0$. Moreover, contact rate $\beta_1$ and $\beta_2$ resulted into the rapid increase of the disease in the community. Also treatment of infected individuals in acute phase is significant. Human recruitment rate have a great influence on controlling the disease. According to the report of CDC \cite{CDC}, flu vaccination reduces the risk of flu illness by between 40\% and 60\% among the overall population during seasons when most circulating flu viruses are well-matched to those used to make flu vaccines. Several different brands of standard dose flu shots are available, including Afluria Quadrivalent, Fluarix Quadrivalent, FluLaval Quadrivalent, and Fluzone Quadrivalent as a active ingradient \cite{WHO}. Vaccination programme can reduce the quantity $(1-\varepsilon)$ which can reduce the rate $\mathcal{R}_0$.

 Quarantine policy, using masks, hand sanitizer, and droplets in a proper way is important to increase the recovery and reduction of interaction rate. 
 We have validated our framework by comparing its predictions with simulation results. In
 the end, vaccination , treatment program, preventive maintenance in lifestyle, and confirming
 an adequately supportive medical care system for all can provide significant result to reduce
 the outbreak of the disease. Thus, our desired model assures the effectiveness of the strategy to reduce and gradually prevent disease outbreaks.

\begin{appendices}
\appendix
\section*{Appendix}
\section{ Demonstrations of the Analytical Findings}\label{allproofs}
This appendix represents the proofs for the analytical results outlined in Section \ref{Section-Mathematical Model SVEIRT}, and \ref{Section-Local-global-stability of DFE EE}.\\

\noindent{\bf Proof of Theorem \ref{theorem01}.}
	The Picard-Lindelof Theorem asserts that concerning the initial value issue,
	$y'(t)= f(y(t)),\;y(t_0)=y_0,\;t \in [t_0-\epsilon,\;t_0+\epsilon]$, if $f$ is continuous in $t$ and locally Lipschitz in $y$, then for some value $\epsilon > 0,$ a unique solution $y(t)$ exists to the initial value problem within the range $[t_0-\epsilon,t_0+\epsilon].$ Because the system of ODEs is autonomous, it suffices to demonstrate that the function $\mathbf{f}: \mathbb{R}^6\rightarrow \mathbb{R}^6$ determined by,
	\[\mathbf{f}(\mathbf{y})=\begin{pmatrix}
		\Lambda-(\beta_1 E+\beta_2 I)S-(\mu+\phi)S \\ \phi S-(1-\varepsilon)(\beta_1 E+\beta_2 I)V -\mu V \\ (\beta_1 E+\beta_2 I)S-(\alpha+\mu)E \\ \alpha E+(1-\varepsilon)(\beta_1 E+\beta_2 I)V-(\mu+\delta+\gamma+\gamma_1)I \\
		\gamma I-\mu R \\ \gamma_1 I-\mu T
	\end{pmatrix}\]
	is locally Lipschitz in its $y$ argument. The Jacobian matrix, 
	\[\nabla\mathbf{f}(\mathbf{y})=\begin{pmatrix}
		a_{11}&0&-\beta_1S&-\beta_2S&0&0\\
		\phi & a_{22}&-\lambda\beta_1V&-\lambda\beta_2V&0&0\\
		\beta_1E+\beta_2I & 0& a_{33}&\beta_2 S&0&0\\
		0& a_{42}&\alpha+\lambda\beta_1 V&a_{44}&0&0\\
		0&0&0&\gamma&-\mu&0\\
		0&0&0&\gamma_1&0&-\mu
	\end{pmatrix}\]
	is linear in $\mathbf{y} \in \mathbb{R}^6$.
	Where \\ $a_{11}=-(\beta_1E+\beta_2I)-(\mu+\phi),\;a_{22}=-\lambda(\beta_1E+\beta_2I)-\mu $,\\
	$a_{33}=\beta_1 S-(\alpha+\mu),\; a_{42}=\lambda(\beta_1E+\beta_2I)$,
	and $a_{44}=\lambda\beta_2V-(\mu+\delta+\gamma+\gamma_1).$
 
	\noindent So, on a closed interval, $\nabla\mathbf{f}(\mathbf{y})$ is continuous while differentiable on an open interval $I_1\in \mathbb{R}^6$. According to the Mean Value Theorem, 
	\begin{equation*}
		\frac{|\mathbf{f}(\mathbf{y_1})-\mathbf{f}(\mathbf{y_2})|}{|\mathbf{y_1}-\mathbf{y_2}|}\leq |\nabla \mathbf{f}(\mathbf{y^*})|
	\end{equation*}
	for $\mathbf{y}^*\in I_1$. Let $|\nabla \mathbf{f}(\mathbf{y^*})|=K$, we get $|\mathbf{f}(\mathbf{y_1})-\mathbf{f}(\mathbf{y_2})|\leq K|\mathbf{y_1}-\mathbf{y_2}|$ for all $\mathbf{y_1},\mathbf{y_2}\in I_1$, and therefore for every $\mathbf{y}  \in \mathbb{R}^6$, $\mathbf{f}(\mathbf{y})$ is locally bounded.As a result, $\mathbf{f}$ is locally Lipschitz in $\mathbf{y}$ since it has a continuous, bounded derivative on any compact subset $\mathbb{R}^3$. The Pichard-Lindelof theorem states that for any time $t_0>0$, there is a unique solution, $y(t)$, to the ordinary differential equation $y'(t)=f(y(t))$ with starting value $y(0)=y_0$ on $[0,t_0]$ \cite{Stability Bound-11}.\\
 
\noindent{\bf Proof of Theorem \ref{th1}.}
	Assume that 
	$$\widehat{t} =\sup \{t >0 : S(t) \geq 0, V(t) \geq 0, E(t) \geq 0,I(t) \geq 0, R(t) \geq 0, \text{and} \;T(t) \geq 0 \} \in [0,t].$$
	Since the solution is continuous and each initial condition is non-negative, there must be a period while the outcome is still positive, and we observe that $\widetilde{t}>0$. Then, each term is calculated on the interval $[0,t]$.
	Thus, $\widehat{t} >0 $ and results from the equation of system \eqref{new_model} that,
	\begin{align*}
		\dfrac{\mathrm{d}S}{\mathrm{d}t} \geq \Lambda - (\lambda_1 +\mu)S. \quad
		\left[\text{where}\;\lambda_1 =(\beta_1 E + \beta_2 I)\right]
	\end{align*}
	
	\noindent This inequality can be resolved by applying the integrating factor approach.
	\begin{align*}
		\frac{d}{dt}\left\{S(t) \exp \left[\mu (t) + \int_0^t \lambda_1 (s) dS\right] \right\}
		\geq \pi \exp \left[\mu t + \int_0^t \lambda_1 (s) dS\right].
	\end{align*}
	Integrating both sides yields,\\
	\begin{align*}
		S(\widehat{t}) \exp \left[\mu \widehat{t} + \int_0^{\widehat{t}} \lambda_1 (s) dS\right] \geq \int_0^{\widehat{t}} \pi \exp \left[\mu \widehat{t} + \int_0^{\widehat{t}} (\lambda_1 (w)) dw\right]d\widehat{t}+C.
	\end{align*}
	Where C is the integration constant depending on the upper limit of $\lambda_1,\;\mu$, and $S(0)$. Hence,
	\begin{flalign*}
		&S(\widehat{t})\geq S(0) \exp\left[-\left(\mu \widehat(t)+\int_0^{\widehat{t}} (\lambda_1(S) dS)\right)\right] + &\\
		&\exp\left[-\left(\mu \widehat(t)+\int_0^{\widehat{t}} (\lambda_1(S) dS)\right)\right]\centerdot \left(\int_0^{\widehat{t}} \pi \exp \left[(\mu \widehat(t)+\int_0^{\widehat{t}} (\lambda_1(w) dw))\right] d\widehat{t}\right) >0.
	\end{flalign*}
	So, $S(\widehat{t}) \geq 0 ,\; \forall\;\widehat{t} \geq 0$.\\
	Next, from the positivity of the solutions place bounds on other compartments \cite{Stability Bound-14}. Here, 
	\begin{align*}
		\frac{dV}{dt} \geq & -((1-\varepsilon)\lambda_1 +\mu)V\\
		\Rightarrow V(\widehat{t}) \geq & V(0) \exp\left[-\left\{ \mu \widehat{t} + \int_0^{\widehat{t}} (1-\varepsilon)\lambda_1(s) dS \right\} \right]>0,\; \forall\;\widehat{t} \geq 0.
	\end{align*}
	
	It is examined that,
	\begin{align*}
		E(\widehat{t}) \geq &  E(0) e^{ -(\mu + \alpha)\widehat{t}} >0.\\
		I(\widehat{t}) \geq &  I(0)  e^{-(\mu + \delta +\gamma +\gamma_1)\widehat{t}} >0.\\
		R(\widehat{t}) \geq &  R(0) e^{-\mu \widehat{t}} >0.\\
		T(\widehat{t}) \geq &  T(0) e^{-\mu \widehat{t}} >0.
	\end{align*}
	for $\widetilde{t}\in [0,t]$. As a result, an upper limit can be set for $S(t), V(t), E(t), I(t), R(t)$, and $T(t)$.
	Therefore, all the solutions of the system \eqref{new_model} will stay non-negative for $t \geq 0$, encompassing at time $\widetilde{t}$. According to continuity, there must exist $t>\widetilde{t}$ such that $S(t), V(t), E(t), I(t), R(t)$, and $T(t)$ are strictly positive on the entire interval $[0,t]$. Further extending the interval of existence is possible because all functions stay bounded on this same interval \cite{Stability Bound-17}. The bounds on the compartments $S, V, E, I, R$, and $T$ that were derived previously hold for any brief time (compact interval). Consequently, we can extend the existence of the solution to $[0,t]$ for every $t>0$. Based on the aforementioned reasoning, the solutions continue to be positive and confined to $[0,t]$.\\

\noindent{\bf Proof of Theorem \ref{th2}.}
	We can write in vector form, 
	$$X=(S,V,E,I,R,T)^T \in \mathbb{R}^6.$$ 
	We define, \[F(X)= \\
	\begin{pmatrix}
		F_1(X) \\F_2(X)\\F_3(X) \\F_4(X)\\F_5(X)\\F_6(X)
	\end{pmatrix}=
	\begin{pmatrix}
		\Lambda-(\beta_1 E(t)+\beta_2 I(t))S(t)-(\mu+\phi)S(t)\\
		\phi S(t)-(1-\varepsilon)(\beta_1 E(t)+\beta_2 I(t))V(t)-\mu V(t)\\
		(\beta_1 E(t)+\beta_2 I(t))S(t)-(\alpha+\mu)E(t)\\
		\alpha E(t)+(1-\varepsilon)(\beta_1 E(t)+\beta_2 I(t))V(t)-(\mu+\delta+\gamma+\gamma_1)I(t)\\
		\gamma I(t)-\mu R(t)\\
		\gamma_1 I(t)-\mu T(t)
	\end{pmatrix}
	\] 
	where $F : \mathbf{C}_+ \rightarrow \mathbb{R}^6$, and $F \in \mathbf{C}^\infty(\mathbb{R}^6)$.\\
 
	\noindent Now,  $$\dot X=F(X_t),$$ \\
	where $\displaystyle \centerdot \equiv \frac{d}{dt}$ including $X_t(\theta) = X(t+\theta),\; \theta \in [0,\tau]. $\\
	It is simple to verify that whenever we desire $X(\theta) \in \mathbf{C}_+$ such that $X_i = 0 $, then we acquire $F_i(X) \vert_{X_i(t)=0}$, $X_{t} \in \mathbf{C}_+ \geq 0, \; i=1,2,\cdots,6.$
	Any result to the model's equation, alongside $X_{t}(\theta) \in \mathbf{C}_+$ say, $X(t)= X(t,X(0))$ is such that $X(t) \in \mathbb{R}_{0+}^6$ for all $t>0$.
	The size of the population, $N=S+V+E+I+R+T$
	with initial conditions, $S(0) \geq 0,\;V(0) \geq 0,\;E(0) \geq 0,\;I(0) \geq 0,\;R(0) \geq 0,\;T(0) \geq 0.$\\
	Now, for the boundedness of the solution we define, 
	\begin{align*}
		\frac{dN}{dt}=&
		\leq   \Lambda-\mu N.
	\end{align*}	
	This indicates that $N(t)$ is bounded, and so are $S(t), V(t), E(t), I(t), R(t) \; \text{and} \; T(t).$
	\begin{align*}
		\text{Here, }\;\;  
		& N \leq N_0e^{-\mu t}+ \frac {\Lambda}{\mu}\left(1-e^{-\mu t}\right)
	\end{align*}
	from this expression when $\displaystyle t \rightarrow \infty, \; \text{and}\; N(t) \leq \frac{\Lambda}{\mu}.$
	The system will be examined in biologically feasible regions as follows. As a result, we can consider the feasible region
	$\displaystyle \Omega=\left\{(S,V,E,I,R,T) \in \mathbb{R}_{+}^6 :S \leq \frac{\Lambda}{\mu} , V, E, I, R, T \geq 0 \right\}.$\\

\noindent{\bf Proof of Theorem \ref{th3}.}
Here,
	\begin{align*}
		\frac{dN}{dt}= \Lambda-\mu N.
	\end{align*}
	The omission of influenza infection ensures that, 
	\begin{align*}
		\frac{dN}{dt} \leq \Lambda-\mu N.
	\end{align*}
	Now,
	\begin{align}\label{th3-2.6}
		N(t) \leq \frac{\Lambda}{\mu}+\left(N(0)-\frac{\Lambda}{\mu}\right)\exp(-\mu t).
	\end{align}
	From (\ref{th3-2.6}), we examine that as $\displaystyle t\rightarrow \infty,\; N(t)\rightarrow \frac{\Lambda}{\mu}$. So, if $\displaystyle N(0)\leq \frac{\Lambda}{\mu}$ then $\displaystyle \lim_{t\to\infty} N(t)=\frac{\Lambda}{\mu}$. On the other hand, if $\displaystyle N(0) > \frac{\Lambda}{\mu}$, then total population $N$ will decrease to $\displaystyle \frac{\Lambda}{\mu}$ as $t\rightarrow \infty$. Particularly, $\displaystyle N(t) < \frac{\Lambda}{\mu}$ if $\displaystyle N(0) < \frac{\Lambda}{\mu}$. This means that $\displaystyle N(t)\leq \max \left\{N(0),\frac{\Lambda}{\mu}\right\}$. Hence, no solution path crosses any borders of $\Omega$, making the region $\Omega$ a positively invariant set of the model \eqref{new_model}. This demonstrates that the developed model is applicable from a mathematical and epidemiological perspective \cite{Stability Bound-14, Stability Bound-18}. The model is appraised in the biologically feasible region which means the considered model is well-defined. Further, $\displaystyle N(t)>\frac{\Lambda}{\mu}$, then the other solution enters $\Omega$ in finite time, or $N(t)$ approaches to $\displaystyle \frac{\Lambda}{\mu}$, and the variables $E(t),\; I(t),\;R(t),\;V(t)\;\text{and} \; T(t)$ approach to zero. Hence the region $\Omega$ is attracting.\\
	Therefore, $N(t)$ is bounded above. Subsequently, $S(t),\;V(t),\;E(t),\;I(t),\;R(t)\; \text{and} \;\;T(t)$ are all bounded above. Thus, in $\Omega$, system \eqref{new_model} is well-posed and global attractor of the system.\\

\noindent {\bf Proof of Theorem \ref{gsDFE_1}.}
 Utilizing Lemma \ref{gsl_DFE} on the model \eqref{new_model}, we acknowledge that $X_1= (S, R)$ and $X_2=(E, I)$ when the system at the DFE. At this DFE, the state variables are given by $X_1^*=(N,0)$. It is significant to remember that:
	\begin{equation*}
		\frac{dX_1}{dt}=F(X_1,0)=
		\begin{pmatrix}
			\mu N-(\mu+\phi)S\\
			-\mu R
		\end{pmatrix}
	\end{equation*}
	is linear and its result can be easily identified as, 
	$$R(t)= R(0)e^{-\mu t},\; S(t)=N-(N-S(0))e^{-\mu t}.$$
	Evidently, as $t\rightarrow \infty$, both $R(t)$ tends to 0 and $S(t)$ tends to $N$, regardless of the initial values of $R(0)$ and $S(0)$. Thus, the equilibrium point $X_1^*=(N,0)$ is globally asymptotically stable, and condition (H1) is satisfied.
	Next,
	\begin{equation*}
		G(X_1,X_2)=
		\begin{pmatrix}
			(\beta_1 E+\beta_2 I)S -(\alpha+\mu)E \\
			\alpha E-(\mu+\delta+\gamma+\gamma_1)I
		\end{pmatrix}.
	\end{equation*}
	We can obtain,
	\begin{equation*}
		A=
		\begin{pmatrix}
			\beta_1 N-(\alpha+\mu) & \beta_2 N\\
			\alpha & -(\mu+\delta+\gamma+\gamma_1)
		\end{pmatrix}	
	\end{equation*}
	with all non-negative off-diagonal elements. Consequently, 
	\begin{equation*}
		\hat{G}(X_1,X_2)=
		\begin{pmatrix}
			\beta_1 E(N-S)+\beta_2 I(N-S)\\
			0
		\end{pmatrix}.
	\end{equation*}
	Since, $0\leq S \leq N$, it is obvious that $\hat{G}\geq 0$. That leads to the global stability of DFE for $\mathcal{R}_0<1$.\\

 \noindent {\bf Proof of Theorem \ref{gs_DFE2}.}
 To examine the global stability of $\mathcal{E}^0$, we assume a Lyapunov functional $U_1(t)$,
	$$U_1=\bar S F\left(\frac{S}{\bar{S}}\right)+E+I = \left(S-\bar{S}-\bar{S}\ln\frac{S}{\bar{S}}\right)+E+I.$$
	Here, $U_1$	is continuous, well-defined, and positive definite for all $(S,V,E,I,R,T) > 0$ and $\theta \in [0,\tau]$.\\
	It illustrates that $U_1$ is always non-negative, and $U_1$ equals zero exclusively when assessed at the non-infective equilibrium point $\mathcal{E}^0$. Moreover, the global minimum of $U_1$ is achieved at $\mathcal{E}^0$. Consequently, all outcomes converge toward the infection-free steady state $\mathcal{E}^0$. Additionally, the functions $U_1$ along the system's trajectories adhere to the following relations:

	\begin{flalign*}
		 \frac{dU_1}{dt} =&\left(1-\frac{\bar{S}}{S}\right)(\Lambda-(\beta_1 E+\beta_2 I)S-(\mu+\phi)S)+(\beta_1 E+\beta_2 I)S-(\mu+\alpha)E+\alpha E+\\
		&\lambda(\beta_1 E+\beta_2 I)V- (\mu+\delta+\gamma+\gamma_1)I.	
	\end{flalign*}
	Utilizing the infection-free steady state of the model \eqref{new_model}, $\Lambda=(\mu+\phi)\bar{S}$ in above expression,
	then the equation becomes,
	\begin{align*}
		\frac{dU_1}{dt} &\leq -\frac{(\mu+\phi)}{S}(S-\bar{S})^2+\beta_1 E\bar{S}+\beta_2 I\bar{S}-(\mu+\phi)E+\alpha E -(\mu+\delta+\gamma+\gamma_1)I\\
		&\leq-\frac{(\mu+\phi)}{S}(S-\bar{S})^2+\left(\frac{\alpha}{\alpha+\mu}\left(\frac{S_0\left[\alpha\beta_2+\beta_1(\gamma+\gamma_1+\mu+\delta)\right]}{(\alpha+\mu)(\gamma+\gamma_1+\delta+\mu)}\right)-1\right)(\mu+\delta+\gamma+\gamma_1)I-\\
		&\;\;\;\;(\beta_1\bar{S}-\mu)E\\
		&\leq-\frac{(\mu+\phi)}{S}(S-\bar{S})^2+\left(\frac{\alpha}{\alpha+\mu}\mathcal{R}_0-1\right)(\mu+\delta+\gamma+\gamma_1)I-\left(\beta_1\frac{\mu N}{\mu+\phi}-\mu\right)E\\
		&\leq 0.
	\end{align*}
	Here, we substituted the DFE value $\displaystyle S_0=\frac{\mu N}{\mu+\phi}$. If $\mathcal{R}_0<1$, then $\displaystyle \frac{dU_1}{dt}$ is negative. Additionally, $\displaystyle \frac{dU_1}{dt}=0$ if and only if $S(t)=\bar{S}$ and $E(t)=I(t)=R(t)=0$. Therefore, based on the Lasalle invariance principle, the infection-free equilibrium point $\mathcal{E}^0$ is globally asymptotically stable over $\Omega$ in the scenario where $\mathcal{R}_0<1$.
\end{appendices}
\begin{appendices}
    \section{Endemic Equilibrium Point}\label{endemic_calculation}
    This appendix presents the brief calculation for computing fixed points in Section \ref{Section-Determination Fixed Points}.\\
    
    The variables for the endemic equilibrium (EE) are substituted as 
$(\widetilde{S},\widetilde{V},\widetilde{E},\widetilde{I},\widetilde{R},\widetilde{T}) \equiv (S^{*},V^{*},E^{*},I^{*},R^{*},T^{*}), $ where $ E^{*} > 0, \; \text{and} \; I^{*} > 0$\ also $E^{*}\neq 0,\; I^{*}\neq 0$.
And we have the following system as follows,

\begin{align}\label{equiEE}
	\begin{cases}
		\vspace{0.2cm}
		\displaystyle\Lambda - \left(\beta_1E^{*}+\beta_2I^{*}\right) S^{*}-(\mu+\phi) S^{*}=0. \\
		\vspace{0.2cm}
		\displaystyle\phi S^{*}-(1-\varepsilon)\left(\beta_1E^{*}+\beta_2I^{*}\right) V^{*}-\mu V^{*}=0.\\
		\vspace{0.2cm}
		\displaystyle\left(\beta_1E^{*}+\beta_2I^{*}\right)S^{*}-(\alpha+\mu) E^{*}=0.\\
		\vspace{0.2cm}
		\displaystyle\alpha E^{*}+ (1-\varepsilon)\left(\beta_1E^{*}+\beta_2I^{*}\right) V^{*}-(\mu+\delta+\gamma+\gamma_1) I^{*}=0.\\
		\vspace{0.2cm}
		\displaystyle\gamma I^{*}-\mu R^{*}=0.\\
		\displaystyle\gamma_1I^{*}-\mu T^{*}=0.
	\end{cases}
\end{align}
Now, adding the first and third equation of \eqref{equiEE} we have,
\begin{align}\label{expS}
	\begin{cases}
		\vspace{0.2cm}
		& \displaystyle\Lambda-(\mu+\phi) S^{*} =  (\alpha+\mu)E^{*}\\
		\vspace{0.2cm}
		& \displaystyle\Rightarrow S^{*}=\frac{\Lambda-(\alpha+\mu)E^{*}}{(\mu+\phi)}= \frac{\Lambda-a_1 E^{*}}{a_2}. 
	\end{cases}
\end{align}
$\text{Here,}\; \displaystyle a_1=\alpha+\mu, \; \text{and} \; a_2=\mu +\phi$. Now, from the second equation of \eqref{equiEE} we have,
\begin{align}\label{expV}
	\begin{cases}
		\vspace{0.2cm}
		&\displaystyle \phi S^{*}=\left\{\mu+\lambda(\beta_1 E +\beta_2 I)\right\} V\\
		\vspace{0.2cm}
		& \displaystyle\Rightarrow V^*=\frac{\phi S^{*}}{\mu+\lambda(\beta_1 E^{*} +\beta_2 I^{*})} = \frac{\phi\left(\frac{\Lambda-a_1 E^{*}}{a_2} \right)}{\mu+\lambda(\beta_1 E^{*} +\beta_2 I^{*})}\\
		&\displaystyle\Rightarrow V^*=\frac{\phi(\Lambda-a_1 E^{*})}{a_2\left\{ \mu +\lambda(\beta_1 E^{*}+\beta_2 I^{*}) \right\}} \;\;\;\; \text{where} \; \lambda=(1-\varepsilon).
	\end{cases}
\end{align}	
Now, from \eqref{expS} substituting the value of $S^*$, we get from third equation of \eqref{expS}, 
\begin{align*}
	& \Rightarrow (\beta_1E^{*}+\beta_2I^{*})=(\alpha+\mu)E^*=a_1 E^{*}\\
	& \Rightarrow \left( \frac{\Lambda-a_1 E^{*}}{a_2}\right) (\beta_1E^{*}+\beta_2I^{*}) = a_1 E^{*}\\
	&\Rightarrow a_1\beta_1  E^{*2}+E(a_1a_2+a_1\beta_2 I^{*}-\Lambda \beta_1)-\Lambda \beta_2 I^{*}=0.
\end{align*}
Now, 
\begin{align}\label{critexpee}
	E^*=\frac{(\Lambda \beta_1-a_1a_2-a_1\beta_2 I^{*})\pm \sqrt{(\Lambda \beta_1-a_1a_2-a_1\beta_2 I^{*})^2 + 4\Lambda\beta_2I^* a_1\beta_1}}{2a_1\beta_1}.
\end{align}
Two roots of the EE point of \eqref{critexpee} will be real if and only if,
\begin{align*}
	&  (\Lambda \beta_1-a_1a_2-a_1\beta_2 I^{*})^2> -4\Lambda\beta_2I^* a_1\beta_1.
\end{align*}
Here, from the expression of \eqref{critexpee}, one root will be always positive, other will be positive if and only if $4\Lambda\beta_2I^* a_1\beta_1 <0.$
Now, 
\begin{align*}
	&\frac{S_0 \alpha \beta_2 +S_0 \beta_1(\gamma +\gamma_1+\delta+\mu)+V_0 \beta_2\lambda(\alpha+\mu)}{(\alpha+\mu)(\gamma +\gamma_1+\delta+\mu)} > 1\\
	& \Rightarrow S_0\alpha\beta_2 +S_0\beta_1 a_3 +V_0\beta_2\lambda a_1 > a_1a_3.
\end{align*}
Where we let $ a_3= (\gamma +\gamma_1+\delta+\mu).$
Now, from the second equation of \eqref{equiEE} we have,
\begin{align*}
	&V^{*}= \frac{\phi(\Lambda-a_1 E^{*})}{a_2(\mu+\lambda\lambda_1)}.
\end{align*}
Where we let $ \lambda_1= \beta_1 E^{*}+\beta_2 I^{*}$, which is the force of infection of the model \eqref{equiEE}.
Similarly, 

\begin{align*} \label{eeI}
	R^* = \frac{\lambda I^*}{\mu}, \; \text{and}\;
	T^* = \frac{\lambda_1 I^*}{\mu}.
\end{align*}
Thus, at the endemic equilibrium $S^{*},V^{*},E^{*},I^{*},R^{*},\; \text{and} \; T^{*}$ depends on the nature of $I^{*}$.
Now, putting all expressions from above from the third equation of \eqref{equiEE} we have,
\begin{align}
	(\beta_1E^* +\beta_2I^*)S^* = (\alpha+\mu)E^*
	\Rightarrow \frac{(\alpha+\mu)E^*-\beta_1E^*S^*}{\beta_2S^*}= I
	\Rightarrow \frac{(\alpha+\mu)E^* -\beta_1E^*\left(\frac{\Lambda-a_1 E^{*}}{a_2}\right)}{\beta_2\left(\frac{\Lambda-a_1 E^{*}}{a_2}\right)} =I^*\nonumber\\
 \Rightarrow \displaystyle a_1E^* - \beta_1E^*\left(\frac{\Lambda-a_1 E^{*}}{a_2}\right)-\beta_2I^*\left(\frac{\Lambda-a_1 E^{*}}{a_2}\right)=0.
\end{align}

Now, from the fourth equation of \eqref{equiEE}  we have,	
\begin{align}
	\begin{cases}
		\vspace{0.2cm}
		&\displaystyle \alpha E^* +\lambda(\beta_1E^*+\beta_2I^*)V^*-(\mu+\delta+\gamma+\gamma_1)I^*=0.\\
		&\displaystyle \alpha E^*+\lambda(\beta_1E^*+\beta_2I^*)\left\{ \frac{\phi(\lambda-a_1E^*)}{a_2\{\mu+\lambda(\beta_1 E^{*}+\beta_2 I^{*})\}} \right\}-(\mu+\delta+\gamma+\gamma_1)I^*=0.
	\end{cases}
\end{align}
Let, the force of infection
$$\lambda_1=(\beta_1E^*+\beta_2I^*),\;
a_3=(\mu+\delta+\gamma+\gamma_1),\; \text{and}\;
a_4=(\mu+\lambda\lambda_1).$$
Now, by solving and simplifying the above two expressions using Mathematica we have obtained,
\begin{flalign*}
	&I^*= \frac{G_1+G_2}{(2a_1a_2a_3a_4(-a_2a_4(a_3\beta_1+\alpha\beta_2)+a_1\beta_2\lambda\lambda_1\phi))}. &\\
	&\text{Where}&\\ &\mathbb{K}=(a_2^2a_4(-2a_1a_2a_3a_4(a_3\beta_1+\alpha\beta_2)\Lambda+a_4(a_3\beta_1+\alpha\beta_2)^2\Lambda^2+a_1^2a_3(a_2^2a_3a_4+4\beta_2\lambda\Lambda\lambda_1\phi))),&\\
	&G_1=-(a_1^2a_2^2a_3a_4\lambda\lambda_1+a_2a_4\alpha(a_2a_4(a_3\beta_1+\alpha\beta_2)\Lambda+\sqrt{\mathbb{K}})), \; \text{and}&\\
	&G_2=a_1(a_2^3a_3a_4^2\alpha+a_2a_4(-a_3\beta_1+\alpha\beta_2)\lambda\Lambda\lambda_1\phi+\lambda\lambda_1\phi\sqrt{\mathbb{K}}).
\end{flalign*}
Thus, for the threshold parameter $\mathcal{R}_0 >1$ we have, 
$$a_4(a_3\beta_1+\alpha\beta_2)^2\Lambda^2+a_1^2a_3(a_2^2a_3a_4+4\beta_2\lambda\Lambda\lambda_1\phi) > 2a_1a_2a_3a_4(a_3\beta_1+\alpha\beta_2)\Lambda.$$
Hence, the expression of the endemic equilibrium point of the model \eqref{new_model} is obtained.\\
Since different virus particles and infected cells are present in varied amounts, we refer to this as vital persistence. We can also shorten the point as $\mathcal{E}^*=\left(S^*, V^*, E^*, I^*, R^*, T^*\right)$.\\
In mathematical biology, $\mathcal{E}^0$ represents a short-lived infection that naturally clears from the body. In contrast, $\mathcal{E}^*$ signifies a situation where the body can't eliminate the illness on its own. In this case, the influenza infection becomes more noteworthy over time \cite{Stability Bound-18, Stability Bound-16}.  Consequently, more sophisticated models accounting for latent infection, the impact of macrophages, the cytotoxic immune response (CLT), or spatial dependence become essential to explain the dynamics of influenza spread throughout the body and its evolution toward an outbreak.\\
If the system explained by \eqref{new_model} reaches an equilibrium point, it will persist throughout the remaining period. Alternatively, the system is not required to reach these equilibrium values. However, it may approach the equilibrium, deviate from it, or oscillate between definite values. Conducting a comprehensive stability study of the system is essential for precisely predicting its behavior and understanding how it will interact with the equilibrium.

\end{appendices}
\begin{appendices}

    \section{Basic Reproduction Number (BRN)}\label{BRN}
    A brief illustration of BRN for both controlling and without controlling strategies is included here for Section \ref{Section-Reproduction Number}.\\
    {\bf Concise computation of BRN with control}.
    From the model~(\ref{new_model}), in the presence of vaccination class $ \mathcal{R}_{0V} $ is known as basic reproduction number with control. With vaccination the DFE of \eqref{equi} is,
$$\mathcal{E}^0\equiv \left( \dfrac{\mu N}{\mu + \phi},\dfrac{\phi N}{\mu + \phi}, 0, 0, 0,0 \right).$$
We now apply the next generation matrix method to the model \eqref{new_model} and modeling only the exposed and infected compartments $E(t)$ and $I(t)$ is necessary since we are only interested in cells that disseminate infection. Hence, considering subpopulation $E(t)$ and $I(t)$ containing new infection terms and disease transmission terms, we can obtain the following subsystem,
\begin{align}\label{new-repro-syst}
	\frac{dE}{dt}= &(\beta_1E+\beta_2I)S-(\alpha+\mu)E. \nonumber\\
	\frac{dI}{dt}= &\alpha E+(1-\varepsilon)(\beta_1E+\beta_2I)V- (\mu+\delta+\gamma+\gamma_1)I.
\end{align}
From the system \eqref{new-repro-syst}, we obtain,\\   
\[F = 
\begin{pmatrix}
	\beta_{1} S_{0}  &  \beta_{2} S_{0}  \\
	\lambda \beta_{1} V_{0}  &  \lambda \beta_{2} V_{0} \\
\end{pmatrix}, \;\; \text{and}\;\; V = 
\begin{pmatrix}
	\mu +\alpha   &  0  \\
	-\alpha   &  \mu+\delta+\gamma+\gamma_1  \\
\end{pmatrix}.\]\\
Therefore, 
\[V^{-1} = 
\begin{pmatrix}
	\frac{\displaystyle 1}{\displaystyle \alpha+\mu}   &   0                            \\
	\frac{\displaystyle \alpha}{\displaystyle (\alpha+\mu)(\mu+\delta+\gamma+\gamma_1)}   &   \frac{\displaystyle 1}{\displaystyle (\mu+\delta+\gamma+\gamma_1)}   \\
\end{pmatrix}.
\]
\\
Here, $F$ and $V$ stand for new infection term and transferred terms, respectively.
Thus, the  next-generation matrix $FV^{-1}$ is,
\[FV^{-1} = 
\begin{pmatrix}
	\frac{\displaystyle S_0\beta_1}{\displaystyle (\alpha+\mu)}+ \frac{\displaystyle S_0\alpha\beta_2}{\displaystyle (\alpha+\mu)(\mu+\delta+\gamma+\gamma_1)} &  \frac{\displaystyle S_0\beta_2}{\displaystyle (\mu+\delta+\gamma+\gamma_1)}    \\
	\frac{\displaystyle V_0\beta_1\lambda}{\displaystyle \alpha+\mu}+\frac{\displaystyle V_0\alpha\beta_2\lambda}{\displaystyle (\alpha+\mu)(\mu+\delta+\gamma+\gamma_1)}   &  \frac{\displaystyle V_0\beta_2\lambda}{\displaystyle (\mu+\delta+\gamma+\gamma_1)} \\
\end{pmatrix}.
\]
Here, the eigenvalues of $FV^{-1}$ are 
$\displaystyle \left\{ 0,\frac{S_0\alpha\beta_2+S_0\beta_1\gamma+S_0\beta_1\gamma_1+S_0\beta_1\delta+V_0\alpha\beta_2\lambda+S_0\beta_1\mu+V_0\beta_2\lambda\mu}{(\alpha+\mu)(\mu+\delta+\gamma+\gamma_1)} \right\}.$\\
Hence, the controlled basic reproduction number $ \mathcal{R}_{0V} $ which is the spectral radius of $FV^{-1}$ is given as follows,
\begin{align} \label{repro1}
	\mathcal{R}_{0V} = \rho(FV^{-1}) &=\frac{S_0\alpha\beta_2+S_0\beta_1\gamma+S_0\beta_1\gamma_1+S_0\beta_1\delta+V_0\alpha\beta_2\lambda+S_0\beta_1\mu+V_0\beta_2\lambda\mu}{(\alpha+\mu)(\mu+\delta+\gamma+\gamma_1)}\nonumber  \\
	&=\frac{\mu N\left[\alpha\beta_2+\beta_1(\mu+\delta+\gamma+\gamma_1)\right]}{(\mu+\phi)(\alpha+\mu)(\mu+\delta+\gamma+\gamma_1)}+\frac{N\phi\beta_2\lambda}{(\mu+\phi)(\mu+\delta+\gamma+\gamma_1)}.
\end{align}

\noindent {\bf Concise computation of BRN without control.}
From the model~\eqref{new_model}, in absence of vaccination class i.e. when $\lambda=(1-\varepsilon)=0,\;\varepsilon=1$, then the threshold quantity $ \mathcal{R}_{0V} $ becomes $ \mathcal{R}_{0} $ which is known as basic reproduction number without control. Without vaccination the DFE of \eqref{equi} is,
$\mathcal{E}^0\equiv \left( \dfrac{\mu N}{\mu + \phi},\dfrac{\phi N}{\mu + \phi}, 0, 0, 0,0 \right).$
\noindent We can extract the two following matrices from the system \eqref{new-repro-syst}  which are $F$ and $V$ by replacing $\varepsilon=0$. They are presented as follows,
\[F = 
\begin{pmatrix}
	\beta_{1} S_{0}  &  \beta_{2} S_{0}  \\
	0  & 0 \\
\end{pmatrix},
\;\;\text{and}\;\; V = 
\begin{pmatrix}
	\mu +\alpha   &  0  \\
	-\alpha   &  \mu+\delta+\gamma+\gamma_1  \\
\end{pmatrix}.\]\\
Thus,
\[V^{-1} = 
\begin{pmatrix}
	\frac{\displaystyle 1}{\displaystyle \alpha+\mu}   &   0                            \\
	\frac{\displaystyle \alpha}{\displaystyle (\alpha+\mu)(\mu+\delta+\gamma+\gamma_1)}   &   \frac{\displaystyle 1}{\displaystyle (\mu+\delta+\gamma+\gamma_1)}   \\
\end{pmatrix}.
\]
Here, $F$ and $V$ stand for new infection term and transferred terms, respectively.
Thus, the  next-generation matrix $FV^{-1}$ is,
\[FV^{-1} = 
\begin{pmatrix}
	\frac{\displaystyle S_0\beta_1}{\displaystyle (\alpha+\mu)}+ \frac{\displaystyle S_0\alpha\beta_2}{\displaystyle (\alpha+\mu)(\mu+\delta+\gamma+\gamma_1)} &  \frac{\displaystyle S_0\beta_2}{\displaystyle (\mu+\delta+\gamma+\gamma_1)}    \\
	0  &  0 \\
\end{pmatrix}.
\]
Here, the eigenvalues of $FV^{-1}$ are
$\displaystyle \left\{ 0,\frac{S_0\alpha\beta_2+S_0\beta_1\gamma+S_0\beta_1\gamma_1+S_0\beta_1\delta+S_0\beta_1\mu}{(\alpha+\mu)(\mu+\delta+\gamma+\gamma_1)} \right\}.$\\
Hence, the basic reproduction number $ \mathcal{R}_{0} $ which is the spectral radious of $FV^{-1}$ is given by,
\begin{align} \label{repro2}
	\mathcal{R}_{0} = \rho(FV^{-1}) &=\frac{S_0\alpha\beta_2+S_0\beta_1\gamma+S_0\beta_1\gamma_1+S_0\beta_1\delta+S_0\beta_1\mu}{(\alpha+\mu)(\mu+\delta+\gamma+\gamma_1)} \nonumber \\	&=\frac{S_0\left[\alpha\beta_2+\beta_1(\gamma+\gamma_1+\mu+\delta)\right]}{(\alpha+\mu)(\gamma+\gamma_1+\delta+\mu)}.
\end{align}
The basic reproduction number ($\mathcal{R}_0$) is defined as the average number of secondary infections resulting from the introduction of a single virus cell into a host where every target cell is susceptible. In our model, $\mathcal{R}_0$ depends on two variables: the average number of target cells per unit of time (considering natural death) and the rate of disease transmission by an infective cell \cite{Stability Bound-25, Stability Bound-26}.
\end{appendices}

\section*{Acknowledgments}
The research by M. Kamrujjaman was partially supported by the University Grants Commission (UGC), 
and the  University of Dhaka, Bangladesh.

\section*{Conflict of interest}
The authors declare no conflict of interest. 


\section*{Data sharing}
There is no available data in this study. 


\end{document}